\documentclass[aps,prx,reprint,nofootinbib,twoside,superscriptaddress,longbibliography,floatfix]{revtex4-2}
\usepackage[english]{babel}
\usepackage[utf8]{inputenc}
\usepackage{amsfonts,amssymb,amsmath}
\usepackage{mathtools}
\usepackage{amsthm}
\usepackage{bm}

\usepackage{tikz}
\usetikzlibrary{arrows.meta,positioning,calc}
\usepackage{xcolor}
\usepackage{hyperref}

\newcommand{\figlink}[2]{\hyperref[#1]{\textcolor{blue}{#2}}}

\usepackage[normalem]{ulem}
\hypersetup{
  colorlinks=true,
  citecolor=blue!55!black,
  linkcolor=blue!55!black,
  urlcolor=blue!55!black,
  pdftitle={A Complete Classification of Complex Hadamard Matrices of Order Six},
  pdfauthor={Mateo Cárdenes Wuttig and Joseph Tindall}
}

\newtheorem{theorem}{Theorem}
\newtheorem{lemma}[theorem]{Lemma}
\newtheorem{proposition}[theorem]{Proposition}
\newtheorem{corollary}[theorem]{Corollary}

\theoremstyle{definition}
\newtheorem{definition}[theorem]{Definition}

\newcommand{\T}{\mathbb T}
\newcommand{\C}{\mathbb C}
\newcommand{\R}{\mathbb R}
\newcommand{\I}{\mathrm i}

\newcommand{\HH}{\mathcal H}
\newcommand{\KK}{\mathcal K}

\newcommand{\FCAtlas}{\mathcal A}
\newcommand{\tr}{\operatorname{tr}}
\newcommand{\RePart}{\operatorname{Re}}

\newcommand{\guide}[1]{}
\providecolor{LLMGREEN}{RGB}{0,128,0}
\DeclareRobustCommand{\LLMgreen}[1]{{\color{LLMGREEN}#1}}
\makeatletter
\def\LLM@blue2#1{{\color{blue}#1}}
\DeclareRobustCommand{\LLM}{\@ifnextchar2{\LLM@blue}{\LLMgreen}}
\makeatother

\begin{document}

\newcommand{\MainCiteHaagerupSzollosi}{%
  ~\cite{Haagerup1997,Szollosi2012FourParameter}}
\newcommand{\MainCiteSzollosi}{~\cite{Szollosi2012FourParameter}}
\newcommand{\MainCiteBasu}{~\cite{BasuPollackRoy2006}}
\newcommand{\MainCiteBondal}{~\cite{BondalZhdanovskiy2016}}

\title{A Complete Classification of Complex Hadamard Matrices of Order Six}

\author{Mateo C\'ardenes Wuttig}
\thanks{Contact author: \href{mailto:mateo.cardeneswuttig@yale.edu}%
{mateo.cardeneswuttig@yale.edu}}
\affiliation{Department of Applied Physics, Yale University, New Haven, Connecticut 06520, USA}
\affiliation{Yale Quantum Institute, Yale University, New Haven, Connecticut 06520, USA}
\author{Joseph Tindall}
\affiliation{Center for Computational Quantum Physics, The Flatiron Institute, 162 5th Avenue, New York, New York 10010, USA}

\ifdefined\SupplementOnly
\maketitle
\tableofcontents
\clearpage
\hypertarget{main-article-numbering}{}
\else
\begin{abstract}
Complex Hadamard matrices encode perfectly balanced unitary transformations.
Their classification is complete through order five, but order six -- the first dimension in which several continuous families coexist with an isolated solution -- has remained open for decades. Here, we give a complete and exact finite-incidence classification of order-six complex Hadamard matrices up to standard equivalence.
We supply the global step missing from Sz\"oll\H{o}si's dilation method, which allows us to prove an even stronger version of his conjecture: every complex Hadamard matrix of order six can be recovered algebraically from a suitable, dephased $3 \times 3$ corner defined by four initial phases.
We then describe the geometry of the reconstruction from these phases and show that, except for Tao's isolated matrix and a single explicit Karlsson matrix, every class admits a representative obtained by solving one quadratic and one cubic equation in both the horizontal and vertical directions.
Our work resolves the classification problem and provides a rigorous framework for further investigating order-six Hadamards, with applications to balanced six-mode interferometers and the study of mutually unbiased bases.
\end{abstract}
\maketitle

\section{Introduction}
Complex Hadamard matrices are square matrices with entries of unit magnitude and mutually orthogonal rows. They arise directly in quantum mechanics as equal-amplitude unitary changes of basis. The subject grew from Sylvester's constructions and Hadamard's determinant problem for real sign matrices with orthogonal rows~\cite{Sylvester1867,Hadamard1893}. 

A physical realization of complex Hadamards is provided by multimode quantum interferometry. After normalization, an order-six complex Hadamard matrix is a perfectly balanced six-port unitary, with identical single-particle transition probabilities but generally different multiparticle interference determined by its internal phases~\cite{Carolan2015Universal,Laing2012BerryPhase,Tichy2010ZeroTransmission,Crespi2015Sylvester,Crespi2016FFT,Shchesnovich2015Partial}. Complex Hadamards are also utilized as coherence-generating gates~\cite{Yao2016MaximalCoherence}, mixing operations in quantum walks~\cite{Mackay2002QuantumWalks,Lorz2019PhotonicWalks}, elements of dual-unitary dynamics~\cite{Gutkin2020DualUnitary,Claeys2022StateDesigns}, and as core components in teleportation, error correction, and contextuality protocols~\cite{Werner2001Teleportation,Musto2016QuantumLatin,Lisonek2019KochenSpecker}. 

Complex Hadamards are also intimately involved in the study of mutually unbiased bases. Two orthonormal bases $\{e_j\}$ and $\{f_k\}$ of the complex Hilbert space $\C^d$ are mutually unbiased when $|\langle e_j,f_k\rangle|=1/\sqrt d$ for every $j,k$. After one basis is
chosen as the standard basis, multiplying the transition matrix to the other basis by $\sqrt d$ gives a complex Hadamard matrix~\cite{TadejZyczkowski2006,WoottersFields1989,BengtssonEtAl2007}.
Determining how many mutually unbiased bases exist in $\C^6$ is one of the central open problems of quantum information theory~\cite{HorodeckiRudnickiZyczkowski2022,McNultyWeigert2026,BengtssonEtAl2007,MatolcsiEtAl2026MUBTriplets}. A complete classification of the $d =6$ Hadamards would provide a natural framework for addressing this problem.

The classification of all complex Hadamard matrices is complete through order five. In dimensions $d=2,3$ and $5$, they are fully classified by the Fourier matrix, while in $d=4$ there is, up to standard equivalence, one continuous one-parameter family~\cite{Haagerup1997,TadejZyczkowski2006}.
Order six is the first, until now, unresolved case, and it is also the first order at which several inequivalent continuous families coexist with a genuinely isolated matrix~\cite{TadejZyczkowski2006,Szollosi2012FourParameter,BengtssonEtAl2007,Tadej2008}. The introduction of Sz\"oll\H{o}si's dilation algorithm~\cite{Szollosi2012FourParameter} represents a significant milestone in the study of order six. The algorithm starts with a dephased $3\times3$ corner of the full matrix, finds candidate blocks completing its first three rows and columns, and then forces the remaining block by linear algebra. Sz\"oll\H{o}si proved that this procedure is exhaustive for a fixed corner when the normalized candidate sets of invertible blocks are finite~\cite{Szollosi2012FourParameter}. He conjectured that, apart from Karlsson's family~\cite{Karlsson2011H2,Karlsson2011ThreeParameter} and Tao's isolated 
matrix~\cite{Tao2004,BengtssonEtAl2007,Tadej2008}, every order-six complex Hadamard could be recovered in this way. 
Bondal and Zhdanovskiy later proved the existence of a family of order-six complex Hadamard matrices of real dimension four~\cite{BondalZhdanovskiy2016}. More recently, several special subclasses have been studied, including non-$H_2$-reducible examples, matrices with prescribed column patterns, eigenvalue-constrained two-parameter families, and matrices with exactly three distinct entries~\cite{LiangEtAl2024,MatszangoszSzollosi2024,HuangChen2026Eigenvalues,HuangLiangChen2026ThreeElements}.

Two obstacles remained between Sz\"oll\H{o}si's fixed-corner result~\cite{Szollosi2012FourParameter} and a complete classification.
First, Sz\"oll\H{o}si's published formulas use divisions and generic polynomial degrees, so special finite branches require separate care.
Second, a chosen corner may have infinitely many candidate blocks that complete it, so fixed-corner completeness does not guarantee that every matrix has a finite corner.

In this work, we close both gaps. We solve exactly for every normalized side block compatible with the chosen corner, using division-free equations so that valid solutions on vanishing-denominator branches are not lost. We then prove that every order-six complex Hadamard matrix has at least one $3\times3$ corner with only finitely many compatible normalized side blocks. At such a corner the actual adjacent blocks occur in the finite candidate sets and the remaining block is uniquely forced. The resulting branch-complete dilation procedure is therefore sound and exhaustive, giving an exact finite-incidence classification and proving Sz\"oll\H{o}si's Conjecture~4.2~\cite{Szollosi2012FourParameter}.

The result of our branch-complete dilation procedure is an explicit classification of all order-six complex Hadamard matrices, up to standard equivalence.
Except for the equivalence classes associated with Tao's isolated matrix and a single explicit Karlsson matrix, every class is shown to admit a representative matrix with a \textit{product-regular} corner specified by just four phase parameters. 
By product-regular, we mean the matrix can be reconstructed using the four phases by solving one quadratic and one cubic equation in each of the horizontal and vertical directions. Compatible solutions determine the adjacent blocks and uniquely force the final $3\times3$ block, whose entries are automatically of unit magnitude.

This work is structured as follows. In section~\ref{sec:definition_hadamard}, we introduce complex Hadamard matrices, the Tao and Karlsson families, and outline the main ideas underlying the classification proof.
Section~\ref{sec:mainclassification} defines our branch-complete finite-dilation extension of Sz\"oll\H{o}si's algorithm before going on to prove the classification.
Section~\ref{sec:atlas} studies the resulting reconstruction over the four-phase seed torus, introducing the notion of product regularity, the corresponding physical domain over the four seed space and relates this local geometry to the full class space.
Section~~\ref{sec:szollosi-conjecture} defines the output of Sz\"oll\H{o}si's published fixed-corner construction and utilizes our classification theorem to prove his conjecture that every order-six class outside the Karlsson and Tao sectors can be recovered by his fixed-corner construction.
In section~\ref{sec:outlook} we conclude by discussing the implications of this classification for the geometry of the order-six Hadamard space, mutually unbiased bases, multiphoton interferometry, and related applications in quantum information and quantum dynamics.
The Supplemental Material~\cite{SupplementalMaterial} contains the complete technical proofs. An accompanying Lean~4 formalization independently verifies the complete classification without relying on any previous classification results~\cite{deMouraUllrich2021,CardenesWuttig2026HadamardRepository}. The separate four-phase geometry of Sec.~\ref{sec:atlas} is supported by exact certificates in the same Github repository~\cite{CardenesWuttig2026HadamardRepository}.

\section{Complex Hadamard Matrices}
\label{sec:definition_hadamard}
\begin{definition}[Complex Hadamard matrix]
\label{def:hadamard}
The unit circle in the complex plane $\C$ is defined as
\begin{equation}
\T=\{z\in\C:|z|=1\}.
\label{eq:unitcircle}
\end{equation}
For $m\geq1$, the torus
\begin{equation}
    \T^m=\{(z_1,\ldots,z_m)\in\C^m:|z_j|=1\ \text{for all }j\},
\end{equation}
is the space of $m$ independent phases. 
A matrix is called unimodular if all of its entries lie in $\T$.
A complex Hadamard matrix of order $n$ is a matrix
$H\in\T^{n\times n}$ satisfying
\begin{equation}
HH^{\dagger}=nI_n,
\label{eq:hadamarddef}
\end{equation}
where $H^{\dagger}$ is the conjugate transpose and $I_n$ is the identity matrix of order $n$. Since $H$ is square, Eq.~\eqref{eq:hadamarddef} is equivalent to
$H^{\dagger}H=nI_n$~\cite{TadejZyczkowski2006}.
\end{definition}

\begin{definition}[Standard equivalence and the class space]
\label{def:equivalence}
Two order-$n$ complex Hadamard matrices are equivalent, written $H\sim H'$, if diagonal unitary matrices $D_{\rm r},D_{\rm c}$ and permutation matrices $P_{\rm r},P_{\rm c}$ exist such that
\begin{equation}
H'=D_{\rm r}P_{\rm r}HP_{\rm c}D_{\rm c}.
\label{eq:equivalence}
\end{equation}
The equivalence class of $H$ is written as $[H]=\{H':H'\sim H\}$. 
A matrix is called dephased if every entry of its first row and first column equals $1$. Every equivalence class contains a dephased representative. 
For a fixed row and column ordering, dephasing produces a unique matrix~\cite{TadejZyczkowski2006,Szollosi2012FourParameter}.
We define the order-six class space by
\begin{equation}
\HH_6 =\bigl\{[H]:H\in\T^{6\times6},\ HH^{\dagger}=6I_6\bigr\}.
\label{eq:H6classes}
\end{equation}
\end{definition}

\begin{definition}[Fourier matrix]
\label{def:fouriermatrix}
The Fourier matrix $F_n$ is a complex Hadamard matrix which exists for every order $n$, defined through its entries
\begin{equation}
(F_n)_{jk}=\exp\!\left(\frac{2\pi\I jk}{n}\right),
\qquad 0\leq j,k<n
\label{eq:fouriern}
\end{equation}
where $\I$ denotes the imaginary unit,
$\I^2=-1$. 
\end{definition}

A $k$-parameter family is a set of equivalence classes
\begin{equation}
\mathcal F_k=\{[H(\mathbf z)]:\mathbf z=(z_1,\ldots,z_k)\in\T^k\}
\subseteq\HH_6,
\end{equation}
where the $k$ phases are independent real coordinates.
In dimension six, Di\c{t}\u{a}-type, circulant, self-adjoint, and other structured ansätze form important continuous families~\cite{Dita2004,BeauchampNicoara2008,MatolcsiSzollosi2008,
Szollosi2010Hypocycloids}. 
Karlsson proved that every matrix containing a $2\times2$ Hadamard submatrix belongs to his complete three-parameter
family~\cite{Karlsson2011H2,Karlsson2011ThreeParameter}. 
We first define the Karlsson family intrinsically through reducibility to order two Hadamard matrices $H_2$.
\begin{definition}[$2\times2$ Hadamard submatrix]
A $2\times2$ complex Hadamard submatrix is any choice of two rows and two columns whose intersection is itself a $2\times2$ complex Hadamard matrix. An order-six complex Hadamard matrix is $H_2$-reducible if it is equivalent to a matrix with such a submatrix. Equivalently, allowed row and column operations can expose the
canonical block
\begin{equation}
F_2=\begin{pmatrix}1&1\\1&-1\end{pmatrix}.
\label{eq:F2}
\end{equation}
We define the Karlsson sector by
\begin{equation}
    \KK_6^{(3)}=\{[H]\in\HH_6:H\text{ is }H_2\text{-reducible}\}.
\label{eq:karlsson}
\end{equation}
Karlsson's theorem identifies this set exactly with the equivalence classes generated by his explicit three-real-parameter family
~\cite{Karlsson2011H2,Karlsson2011ThreeParameter}.
\end{definition}

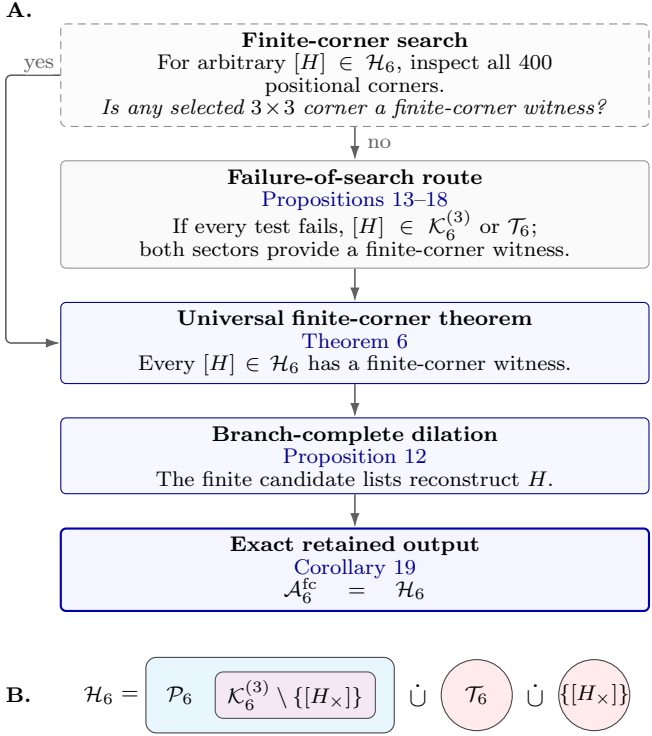
\begin{figure}[htb!]
\centering
\begingroup
\begin{tikzpicture}[
  every node/.style={font=\footnotesize,align=center},
  >={Latex[length=1.9mm,width=1.25mm]},
  flow/.style={->,semithick,draw=black!62,rounded corners=2pt},
  full/.style={rounded corners=2pt,inner xsep=5pt,inner ysep=3pt,
               minimum width=.90\columnwidth,text width=.82\columnwidth},
  decision/.style={full,draw=black!45,densely dashed,fill=black!1},
  neutral/.style={full,draw=black!42,fill=black!2},
  bluebox/.style={full,draw=blue!58!black,fill=blue!3},
  output/.style={full,draw=blue!68!black,thick,fill=blue!4},
  edge/.style={font=\scriptsize,fill=white,inner sep=1pt,text=black!62},
  panel/.style={font=\footnotesize\bfseries,text=black,inner sep=0pt,outer sep=0pt},
  classlabel/.style={font=\footnotesize,align=center,text=black},
  eqsymbol/.style={font=\footnotesize,align=center,text=black,inner sep=0pt,outer sep=0pt},
  disjoint/.style={font=\normalsize,align=center,text=black,inner sep=0pt,outer sep=0pt},
  pregion/.style={draw=black!75,fill=cyan!8,rounded corners=4pt,
                  minimum width=.575\columnwidth,minimum height=10.0mm},
  subclass/.style={classlabel,draw=black!70,fill=violet!9,rounded corners=3pt,
                   minimum height=6.1mm},
  isolated/.style={circle,draw=black!75,fill=red!8,minimum size=9.4mm,
                   inner sep=0pt,outer sep=0pt}
]

\node[decision] (choice) {%
  \textbf{Finite-corner search}\\[-1pt]
  For arbitrary $[H]\!\in\!\HH_6$, inspect all $400$\\[-1pt]
  positional corners.\\[-1pt]
  \emph{Is any selected $3\!\times\!3$ corner a finite-corner witness?}%
};
\coordinate (figleft) at ([xshift=-7.2mm]choice.west);
\node[panel,anchor=south west]
  at ($(figleft |- choice.north)+(0,0.8mm)$) {\textbf{A.}};

\node[neutral,below=4.4mm of choice] (routing) {%
  \textbf{Failure-of-search route}\\[-1pt]
  \hyperref[prop:main-infinite-fiber-trichotomy]{Propositions~\ref*{prop:main-infinite-fiber-trichotomy}}--\hyperref[prop:tao-finite-corner]{\ref*{prop:tao-finite-corner}}\\[-1pt]
  If every test fails, $[H]\!\in\!\KK_6^{(3)}$ or $\mathcal T_6$;\\[-1pt]
  both sectors provide a finite-corner witness.%
};

\node[bluebox,below=4.4mm of routing] (witness) {%
  \textbf{Universal finite-corner theorem}\\[-1pt]
  \hyperref[thm:classification]{Theorem~\ref*{thm:classification}}\\[-1pt]
  Every $[H]\!\in\!\HH_6$ has a finite-corner witness.%
};

\node[bluebox,below=4.4mm of witness] (dilation) {%
  \textbf{Branch-complete dilation}\\[-1pt]
  \hyperref[prop:algorithmvalid]{Proposition~\ref*{prop:algorithmvalid}}\\[-1pt]
  The finite candidate lists reconstruct $H$.%
};

\node[output,below=4.4mm of dilation] (finish) {%
  \textbf{Exact retained output}\\[-1pt]
  \hyperref[cor:atlas-classification]{Corollary~\ref*{cor:atlas-classification}}\\[-1pt]
  $\FCAtlas_6^{\mathrm{fc}}\!=\!\HH_6$%
};

\draw[flow] (choice.south)
  -- node[edge,right,midway,xshift=1.2mm,font=\footnotesize] {no}
  (routing.north);
\draw[flow] (choice.west)
  -- ++(-4.0mm,0)
  node[edge,above,pos=.45,xshift=-1.0mm,font=\footnotesize] {yes}
  -- (figleft) |- (witness.west);
\draw[flow] (routing.south) -- (witness.north);
\draw[flow] (witness.south) -- (dilation.north);
\draw[flow] (dilation.south) -- (finish.north);

\coordinate (brow) at ([yshift=-11.0mm]finish.south);
\node[panel,anchor=west] (panelb)
  at (figleft |- brow) {\textbf{B.}};
\node[eqsymbol,anchor=west] (hclass)
  at ([xshift=3.2mm]finish.west |- brow) {$\HH_6$};
\node[eqsymbol,anchor=west] (equals)
  at ([xshift=0.9mm]hclass.east |- brow) {$=$};
\node[pregion,minimum width=.38\columnwidth,anchor=west] (preg)
  at ([xshift=0.9mm]equals.east |- brow) {};
\coordinate (pregrow) at (preg.center);
\node[subclass,minimum width=.245\columnwidth,anchor=east] (kclass)
  at ([xshift=-2.5mm]preg.east |- pregrow)
  {$\KK_6^{(3)}\setminus\{[H_\times]\}$};
\coordinate (pcenter) at ($($(preg.west)!0.5!(kclass.west)$ |- pregrow)$);
\node[classlabel,anchor=center] at (pcenter) {$\mathcal P_6$};
\node[disjoint,anchor=west] (unionone)
  at ([xshift=1.9mm]preg.east |- brow) {$\mathbin{\dot\cup}$};
\node[isolated,anchor=west] (taoclass)
  at ([xshift=1.9mm]unionone.east |- brow) {};
\node[classlabel] at (taoclass.center) {$\mathcal T_6$};
\node[disjoint,anchor=west] (uniontwo)
  at ([xshift=1.9mm]taoclass.east |- brow) {$\mathbin{\dot\cup}$};
\node[isolated,anchor=west] (hxclass)
  at ([xshift=1.9mm]uniontwo.east |- brow) {};
\node[classlabel] at (hxclass.center) {$\{[H_\times]\}$};
\end{tikzpicture}
\vspace{2pt}
\caption{\textbf{Proof architecture and class-space structure.} \textbf{A.} Summary of the logic which proves the classification theorem (Theorem ~\ref{thm:classification}). Starting from three rows and three columns of $H$ in block form, we search all corners for a finite-corner witness, which is a dephased $3\times3$ corner whose normalized invertible side fibers are finite, nonempty, and contain the adjacent blocks (\hyperref[def:finitecorner]{Definition~\ref*{def:finitecorner}}). If such a corner exists, the remaining block is uniquely forced by Sz\"oll\H{o}si's construction. If no selected corner is a finite-corner witness, the failure-of-search route places $H$ in the Karlsson or Tao sector, where a finite-corner witness always exists. 
\hyperref[def:completed-output]{Definition~\ref*{def:completed-output}} specifies the branch-complete output, which equals $\HH_6$, the full class space of all order-six Hadamard matrices. \textbf{B.} Geometrical structure of the set of all classes of order-six Hadamards $\HH_6$.
This set consists of the product-regular set $\mathcal P_6$ (see Definition ~\ref{def:product-regularity}), which contains the product-regular Karlsson classes $\KK_6^{(3)}\setminus\{[H_\times]\}$, alongside two isolated sets: Tao's isolated class $\mathcal T_6$ and the unique product-exceptional Karlsson class $\{[H_\times]\}$ (\hyperref[thm:global-product-regular-escape]{Theorem~\ref*{thm:global-product-regular-escape}}).}
\label{fig:finite-corner-proof}
\endgroup
\end{figure}


Distinct from it lies the isolated cubic-root matrix $S_6^{(0)}$ associated with Tao's work~\cite{Tao2004,TadejZyczkowski2006,Szollosi2012FourParameter,Tadej2008}. \begin{definition}[Tao's dephased matrix]
Let $\omega=\exp(2\pi\I/3)$. Then $\omega^3=1$, $\omega\ne1$, and $1+\omega+\omega^2=0$. Tao's dephased matrix~\cite{Tao2004} is
\begin{equation}
S_6^{(0)}=
\begin{pmatrix}
1&1&1&1&1&1\\
1&1&\omega&\omega&\omega^2&\omega^2\\
1&\omega&1&\omega^2&\omega^2&\omega\\
1&\omega&\omega^2&1&\omega&\omega^2\\
1&\omega^2&\omega^2&\omega&1&\omega\\
1&\omega^2&\omega&\omega^2&\omega&1
\end{pmatrix},
\label{eq:taomatrix}
\end{equation}
\end{definition}
which is complex Hadamard. We define the one-point Tao sector~\cite{Tao2004} by
\begin{equation}
\mathcal T_6=\{[S_6^{(0)}]\}.
\label{eq:Taosector}
\end{equation}

We can write any complex Hadamard matrix of order six in block form as
\begin{equation}
H=
\begin{pmatrix}
E&B\\
C&D
\end{pmatrix},
\label{eq:blockform}
\end{equation}
where $E,B,C,D\in\T^{3\times3}$. Sz\"oll\H{o}si's dilation construction~\cite{Szollosi2012FourParameter} starts from a dephased $3\times3$ corner $E$, determines the normalized candidate blocks $B$ and $C$ from their fixed-Gram constraints, and then forces $D$ by block orthogonality. We call the complete normalized solution set for $B$ or $C$ a \emph{side fiber}.

Sz\"oll\H{o}si's fixed-corner completeness theorem states that, when the two invertible side fibers of a prescribed corner $E$ are finite and nonempty, pairing all their candidates and testing the forced block $D$ enumerates every Hadamard completion containing $E$. Accordingly, a selected $3\times3$ corner of a given order-six Hadamard matrix $H$ is a \emph{finite-corner witness} if its two invertible side fibers are finite and nonempty and contain the actual adjacent blocks $B$ and $C$. The original matrix $H$ is then recovered among finitely many completions.
The fixed-corner theorem is conditional: some corners have infinite side fibers, and it does not show that every Hadamard matrix possesses a finite-corner witness. 

In the following, we prove that every order-six complex Hadamard matrix has such a witness. We show that an infinite side fiber forces either a Fourier block, a strictly negative real cubic Gram invariant, or the Karlsson sector. Outside the latter case, complementary $3\times3$ blocks reverse the sign of the cubic invariant, producing a corner with two finite side fibers unless all four complementary blocks are Hadamard. This case then reduces to the Karlsson or Tao sectors. Neither sector is excluded from the witness theorem: every Karlsson class admits a finite-corner witness, while the leading $3\times3$ corner of Tao's matrix is one. Thus every order-six complex Hadamard matrix has a finite-corner witness, closing the gap between Sz\"oll\H{o}si's fixed-corner construction and an exhaustive classification.

\section{Finite-corner classification}
\label{sec:mainclassification}

We formulate the classification in terms of finite $3\times3$ corners.
Choose three rows and three columns of an order-six Hadamard matrix, move them to the leading positions (rows and columns $1,...,3$), and dephase. The matrix splits into four $3\times3$ blocks as in Eq.~\eqref{eq:blockform}, and we call its upper-left block $E$ the \emph{corner} for that choice of rows and columns. Orthogonality of the full matrix fixes every inner product among the rows of the block $B$ and among the columns of the block $C$ in terms of $E$ alone. Each of $B$ and $C$ is therefore confined to a set of candidates determined by the corner (see Eq.~\eqref{eq:candidatesets}).
If both candidate sets are finite and nonempty they can be listed exhaustively, the fourth block is then determined, and the matrix we started from is recovered among the finitely many completions. Such a corner is a \emph{finite-corner witness}, defined rigorously in Definition~\ref{def:finitecorner}.

\begin{theorem}[Complete finite-corner classification]
\label{thm:classification}
Every order-six complex Hadamard matrix is equivalent to a dephased matrix
having a finite-corner witness.
\end{theorem}

Figure~\ref{fig:finite-corner-proof}(a) summarizes the proof, which is assembled at the end of this section. Below, we state two structural results from the literature used in this manuscript's proof. We emphasize that the accompanying Lean formalization ~\cite{CardenesWuttig2026HadamardRepository} directly proves the parts of both results needed for the classification, rather than using them as assumptions. We then formulate our branch-complete refinement of Sz\"oll\H{o}si's dilation construction~\cite{Szollosi2012FourParameter}, which retains solutions on branches where the divisions used in a generic implementation would be invalid.

The central step in our proof is then a global routing argument. Starting from an arbitrary order-six Hadamard matrix, we examine all $400$ positional corners. 
If none of them are a finite-corner witness, then at least one side fiber is infinite at every corner. We show that, by analyzing these infinite fibers and their complementary blocks, this forces the matrix into either the Karlsson or Tao sector. We then prove separately that every matrix in the Karlsson and Tao sectors has a finite-corner witness. 
Hence, every order-six complex Hadamard matrix has such a witness. Applying the branch-complete construction at that corner recovers the original matrix, and therefore completes the classification.
The full class space of all Hadamard matrices in sketched in Figure~\ref{fig:finite-corner-proof}(b).

\subsection{Structural results used in the universal proof}
The proof in this manuscript starts from two previously established structural results for order-six Hadamard matrices. 

\begin{proposition}[Published order-six structural inputs]
\label{prop:publishedinputs}
For an order-six complex Hadamard matrix $H$:
\begin{enumerate}
\item $H$ is $H_2$-reducible if and only if it is equivalent to a
member of Karlsson's complete three-real-parameter family
\cite{Karlsson2011H2,Karlsson2011ThreeParameter}. Karlsson's
parametrization also covers the parameter values at which both
M\"obius formulas degenerate: the resulting Hadamard matrices are,
up to equivalence, members of the two-parameter Fourier family
$F_6^{(2)}$ or its transpose
\cite[Secs.~4--5 and Theorem~11]{Karlsson2011ThreeParameter}.
See also
\cite[p.~623, Theorem~2.11]{Szollosi2012FourParameter}.
\item If $H$ is equivalent to a dephased complex Hadamard matrix having both a noninitial row and a noninitial column composed entirely of cubic roots of unity, then either $[H] \in \mathcal T_6$ or
$[H]\in\KK_6^{(3)}$~\cite[p.~624, Lemma~2.14]{Szollosi2012FourParameter}. 
In a dephased matrix, a noninitial row or
column is any other than the first.
\end{enumerate}
\end{proposition}

Our accompanying Lean~4 formalization \cite{CardenesWuttig2026HadamardRepository} does not assume the two published results of Proposition~\ref{prop:publishedinputs}. It proves both within Lean as part of the argument, making the formal verification of the classification self-contained~\cite{CardenesWuttig2026HadamardRepository}.

\begin{lemma}[Singular-corner reduction]
\label{lem:singularcorner}
If a $3\times3$ submatrix of an order-six complex Hadamard matrix is singular, then the matrix contains a $2\times2$ Hadamard submatrix.
\end{lemma}

\begin{proof}[Proof Sketch]
After dephasing, the singularity condition together with unimodularity forces two rows or two columns of the corner to coincide. If two rows coincide on the selected columns, their relative phases there are all $1$. Orthogonality of the full rows then forces their relative phases on the complementary columns to be all $-1$, so choosing one column from each triple yields a $2\times2$ Hadamard submatrix. The coincident-column case follows by transposition. The detailed proof is contained in Supplemental Sec.~\ref{sec:supp-singular-corner}.
\end{proof}

Sz\"oll\H{o}si proves the equivalent statement that a vanishing $3\times3$ minor -- the determinant of a submatrix -- forces
$[H]\in\KK_6^{(3)}$~\cite{Szollosi2012FourParameter}. The two conclusions agree by Proposition~\ref{prop:publishedinputs}(1).
We present a self-contained proof that stops at the $2\times2$ submatrix, so that Karlsson's theorem is not used here and the two structural inputs for this manuscript remain confined to Proposition~\ref{prop:publishedinputs}.

\subsection{Normalized candidate fibers}
Given an arbitrary matrix, we use the $3\times3$ block form of Eq.~\eqref{eq:blockform}, select three rows and three columns, move the resulting $3\times3$ submatrix to the upper-left corner, and dephase the result.
The first row of $E$ and $B$, and the first column of $E$ and $C$, then consist of ones. Block multiplication in $HH^{\dagger}=6I_6$ and
$H^{\dagger}H=6I_6$ gives
\begin{equation}
\begin{aligned}
EE^{\dagger}+BB^{\dagger}&=6I_3,\\
CC^{\dagger}+DD^{\dagger}&=6I_3,\\
E^{\dagger}E+C^{\dagger}C&=6I_3,\\
B^{\dagger}B+D^{\dagger}D&=6I_3,\\
EC^{\dagger}+BD^{\dagger}&=0,\\
E^{\dagger}B+C^{\dagger}D&=0.
\end{aligned}
\label{eq:blockorthogonality}
\end{equation}
The first and third lines of Eq.~\eqref{eq:blockorthogonality} are the Gram conditions that determine the normalized candidates adjacent to $E$.

\begin{definition}
For a fixed dephased corner $E$, its normalized horizontal and vertical candidate sets are
\begin{equation}
\begin{aligned}
\mathcal B_E={}&\{B'\in\T^{3\times3}:\\
&B'B'^{\dagger}=6I_3-EE^{\dagger},\\
&B'_{1j}=1\ \text{for }1\le j\le3\},\\[2pt]
\mathcal C_E={}&\{C'\in\T^{3\times3}:\\
&C'^{\dagger}C'=6I_3-E^{\dagger}E,\\
&C'_{i1}=1\ \text{for }1\le i\le3\}.
\end{aligned}
\label{eq:candidatesets}
\end{equation}
The first condition in each set is a fixed-Gram condition. The row or column of ones correspond to the full matrix $H$ being dephased.
Write
\begin{equation}
\begin{aligned}
\mathcal B_E^\times&=\{B\in\mathcal B_E:\det B\ne0\},\\
\mathcal C_E^\times&=\{C\in\mathcal C_E:\det C\ne0\},
\end{aligned}
\label{eq:invertible-candidates}
\end{equation}
for the invertible candidates.
\end{definition}

\begin{definition}[Finite-dilation corner]
\label{def:finitecorner}
A dephased corner $E$ is a \emph{finite-dilation corner} when
\begin{align}
\mathcal B_E^\times\text{ is nonempty and finite}, \notag \\
\qquad
\mathcal C_E^\times\text{ is nonempty and finite}.
\label{eq:finite-dilation-corner}
\end{align}
A displayed Hadamard matrix has a \emph{finite-corner witness} when its upper-left block is a finite-dilation corner and its actual adjacent blocks belong to $\mathcal B_E^\times$ and $\mathcal C_E^\times$.
\end{definition}

\subsection{The completed finite-dilation procedure}
\label{sec:dilation}
We now introduce the following branch-complete version of Sz\"oll\H{o}si's dilation algorithm~\cite{Szollosi2012FourParameter}.
His architecture is unchanged: choose a dephased corner, solve for the adjacent blocks, and force the fourth block. Our formulation replaces the generic companion-function divisions by the full normalized fixed-Gram equations. 
Keeping these uncancelled retains finite exceptional solutions that may be lost in the generic quotient formulas. Below, we denote our five-step refinement as the ``branch-complete finite-dilation" procedure.
Every dephased corner has the following form
\begin{equation}
E(a,b,c,d)=
\begin{pmatrix}
1&1&1\\
1&a&b\\
1&c&d
\end{pmatrix},
\label{eq:singular-corner-form}
\end{equation}
with four phases $(a,b,c,d)\in\T^4$. After the normalization in Eq.~\eqref{eq:candidatesets}, $B$ has six unknown phase entries below its first row and $C$ has six unknown phase entries to the right of its first column; these phases are constrained unknowns.

\paragraph*{Branch-complete dilation procedure.}
\begin{enumerate}
\item Choose seed phases $(a,b,c,d)\in\T^4$ and form $E(a,b,c,d)$.
\item Form the complete normalized physical solution sets of the two fixed-Gram systems
\begin{equation}
\begin{aligned}
BB^{\dagger}&=6I_3-EE^{\dagger},\\
C^{\dagger}C&=6I_3-E^{\dagger}E.
\end{aligned}
\label{eq:algorithmgrams}
\end{equation}
Equivalently, solve for every member of $\mathcal B_E$ and $\mathcal C_E$.
\item Continue only when $E$ is a finite-dilation corner (Definition \ref{def:finitecorner}), and consider every pair
\hbox{$(B,C)\in\mathcal B_E^\times\times\mathcal C_E^\times$}.
\item For every such pair define the only possible fourth block,
\begin{equation}
D(E;B,C)=-CE^{\dagger}(B^{-1})^{\dagger}.
\label{eq:maincompletionformula}
\end{equation}
\item Retain the resulting block matrix precisely when all nine entries of $D(E;B,C)$ have modulus one, that is, when $|D_{ij}(E;B,C)|=1$ for every
$1\leq i,j\leq3$.
\end{enumerate}

\paragraph*{Solving the normalized fixed-Gram equations} In Step~2, we determine the complete physical solution sets $\mathcal B_E$ and $\mathcal C_E$, without selecting a single generic algebraic branch. For example, write a normalized horizontal candidate as 
\begin{equation} 
B= 
\begin{pmatrix} 
1&1&1\\ x_1&x_2&x_3\\ y_1&y_2&y_3 \end{pmatrix}, 
\label{eq:main-fixed-gram-coordinate-form} 
\end{equation} 
where $x_j,y_j\in\T$.
Equation \ref{eq:algorithmgrams} prescribes the three off-diagonal row inner products of $B$.
From the entries of Eq.~\ref{eq:singular-corner-form}, we define
\begin{equation}
\begin{split}
   s_h:=&-(1+a+b),\\
    t_h:=&-(1+c+d),\\
    r_h:=&-\left(1+\frac ca+\frac db\right). 
\end{split}
\label{eq:trace-data-main}
\end{equation}
Because the seed phases have unit modulus, we can write the complex conjugate as
\begin{equation}
\begin{aligned}
\overline{s_h}&=-(1+a^{-1}+b^{-1}),\\
\overline{t_h}&=-(1+c^{-1}+d^{-1}),\\
\overline{r_h}&=-\left(1+\frac ac+\frac bd\right).
\end{aligned}
\end{equation}
The complete horizontal fixed-Gram system is therefore
\begin{equation}
\begin{aligned}
\sum_{j=1}^3x_j&=s_h,
&\sum_{j=1}^3x_j^{-1}&=\overline{s_h},\\
\sum_{j=1}^3y_j&=t_h,
&\sum_{j=1}^3y_j^{-1}&=\overline{t_h},\\
\sum_{j=1}^3\frac{y_j}{x_j}&=r_h,
&\sum_{j=1}^3\frac{x_j}{y_j}&=\overline{r_h}.
\end{aligned}
\label{eq:main-fixed-gram-sums}
\end{equation}
Multiplying these equations by the appropriate products $x_1x_2x_3$ and $y_1y_2y_3$ removes all displayed denominators. These
products are nonzero everywhere due to the unimodularity of $x_{i}$ and $y_{i}$, so the conversion is reversible and neither introduces nor discards any physical solution.

The vertical candidate system is obtained in the same way for $C$ from Eq.~\ref{eq:algorithmgrams}. If a physical side fiber is finite, every one of its solutions is retained. If it is positive-dimensional, and hence infinite, we do not replace it by a generic finite root list. Instead, the corner is treated as nonfinite and passed to the global routing argument, which will be introduced in Proposition~\ref{prop:main-corner-routing}.

\paragraph*{Relation to Sz\"oll\H{o}si's published construction.}
Our procedure and Sz\"oll\H{o}si's Construction~3.1
\cite{Szollosi2012FourParameter} solve the same fixed-Gram candidate problem. Using the normalized block $B$ from Eq.~\eqref{eq:main-fixed-gram-coordinate-form}, consider the first unknown column $(x_1,y_1)^T$. These phases are denoted by $(e,f)^T$ in Sz\"oll\H{o}si's construction.

Two consequences of the fixed-Gram equations are quadratic in $y_1$:
\begin{equation}
\begin{aligned}
F_1(x_1)+F_2(x_1)y_1+F_3(x_1)y_1^2&=0,\\
G_1(x_1)+G_2(x_1)y_1+G_3(x_1)y_1^2&=0,
\end{aligned}
\label{eq:main-szollosi-parent-quadratics}
\end{equation}
where the coefficients also depend on the corner phases $a,b,c,d$.
Eliminating $y_1^2$ without division gives
\begin{equation}
A(x_1)+B(x_1)y_1=0,
\label{eq:main-szollosi-companion-relation}
\end{equation}
where
\begin{equation}
A=F_3G_1-F_1G_3,
\qquad
B=F_3G_2-F_2G_3.
\end{equation}
When $B(x_1)\neq0$, the companion phase is
\begin{equation}
y_1=-\frac{A(x_1)}{B(x_1)}.
\label{eq:main-szollosi-companion-quotient}
\end{equation}
Imposing $|y_1|=1$ produces Sz\"oll\H{o}si's fundamental polynomial in $x_1$, which has degree six in the generic case. The same calculation is then used to obtain the other two pairs $(x_2,y_2)$ and $(x_3,y_3)$.

The quotient in Eq.~\eqref{eq:main-szollosi-companion-quotient} requires
$B(x_1)\neq0$. If this condition fails, there are two possibilities:
\begin{equation}
\begin{array}{rcl}
B(x_1)=0,\ A(x_1)\neq0
&\Longrightarrow& \text{no solution},\\[2pt]
B(x_1)=A(x_1)=0
&\Longrightarrow& \text{return to parent equations.}
\end{array}
\label{eq:main-szollosi-companion-cases}
\end{equation}
In the second case, the linear relation becomes $0=0$ and does not determine $y_1$. The two parent equations may then have finitely many exceptional solutions or infinitely many solutions. The degree drops and repeated roots must be handled separately.

Our branch-complete procedure retains the original fixed-Gram equations and strictly avoids premature division. Here \emph{division-free} is meant in the following precise sense: the procedure allows reciprocals of nonvanishing phase variables and the matrix inverse $B^{-1}$, computed only after verifying $\det B\neq0$. We do not divide by a parameter-dependent expression without first proving that it is nonzero.

Because of this strict condition, our method agrees exactly with Sz\"oll\H{o}si's dilation method on regular domains \cite{Szollosi2012FourParameter}. Whenever his divisions are valid and the normalized invertible side fibers are finite, both procedures yield the same candidate blocks and the same Hadamard completions.

The advantage of the complete-fiber formulation emerges in the exceptional cases. Where the standard quotient formulas become undefined, our method returns to the uncancelled parent equations and captures every finite physical solution. For a concrete seed in an effectively presented field, these exceptional solutions are obtained as the unit-torus points of the resulting zero-dimensional Laurent system, rather than from undefined quotients. Concretely, this is achieved through the following pipeline:
\begin{itemize}
\item \emph{Formulation.} Each phase is expressed in its real and imaginary parts, and the unit-circle constraints are imposed.
\item \emph{Elimination.} Exact algebraic elimination is applied, such as computing a Gr\"obner basis or a rational-univariate representation.
\item \emph{Isolation and verification.} All real roots are isolated exactly, and every candidate is verified directly against the original system.
\end{itemize}

\begin{definition}[Total finite-corner atlas]
\label{def:completed-output}
For $E=E(a,b,c,d)$, see Eq.~\ref{eq:singular-corner-form}, let $\operatorname{Out}(E)$ be the set of all retained matrices from the branch-complete dilation procedure.
We define
\begin{equation}
\begin{aligned}
\FCAtlas_6^{\mathrm{fc}}=\bigl\{[H]:\;&\exists(a,b,c,d)\in\T^4
\text{ such that}\\[-2pt]
&H\in\operatorname{Out}(E(a,b,c,d))\bigr\}.
\end{aligned}
\label{eq:Gdefinition}
\end{equation}
Thus, $\FCAtlas_6^{\mathrm{fc}}$ is the set of equivalence classes of the total retained output of our completed procedure, including every regular or special retained output arising from an invertible candidate pair in a finite normalized fiber. The existential quantifier means that for each equivalence class, there is at least one four-phase seed that produces a representative of that class. 
\end{definition}

\begin{proposition}[Soundness of the retained output]
\label{prop:algorithmvalid}
Every matrix retained in Definition~\ref{def:completed-output} is a complex Hadamard matrix.
\end{proposition}

\begin{proof}
The candidate equations and Eq.~\eqref{eq:maincompletionformula} imply
\begin{equation}
C^{\dagger}D=-E^{\dagger}B,
\qquad
D^{\dagger}D=6I_3-B^{\dagger}B.
\label{eq:maincompletionidentities}
\end{equation}
Together with Eq.~\eqref{eq:algorithmgrams}, these identities give
\begin{equation}
H(E;B,C)^{\dagger}H(E;B,C)=6I_6.
\label{eq:mainoutputunitarity}
\end{equation}
where 
\begin{equation}
H(E;B,C)=
\begin{pmatrix}
E&B\\
C&D(E;B,C)
\end{pmatrix}.
\label{eq:algorithmoutputmatrix}
\end{equation}
The retention test says that every entry of the remaining block $D$ has modulus one, while this is already true for $E,B,C$. Hence the retained matrix satisfies Definition~\ref{def:hadamard}. The complete derivation of Eq.~\eqref{eq:maincompletionidentities} is given in Supplemental Sec.~\ref{sec:directcompletion}.
\end{proof}

\subsection{Proof of the classification theorem}

Next, we prove the exhaustiveness of the retained output (Theorem~\ref{thm:classification}). Starting from an arbitrary order-six Hadamard matrix $H$, choose three of its six rows and three of its six columns. This gives $\binom{6}{3}^2=400$ positional corners. We prove that at least one is a finite-corner witness to $H$.

We first consider matrices outside the Karlsson sector. By Lemma~\ref{lem:singularcorner} and Proposition~\ref{prop:publishedinputs}(1), every $3\times3$ submatrix of such a matrix is invertible. Consequently, at every possible corner, because $H$ satisfies the block Gram identities, the blocks $B$ and $C$ belong to the corresponding side fibers. Thus both fibers are nonempty at every positional corner. If both fibers are finite, the corner is already a finite-corner witness. The only remaining obstruction is that one or both fibers are infinite.

We resolve this obstruction in three steps. Proposition~\ref{prop:main-infinite-fiber-trichotomy} classifies the possible structures of an infinite fixed-Gram fiber. Proposition~\ref{prop:main-corner-routing} combines that local result with the complementary block identities to find a finite corner, except when all four blocks in a selected partition are order-three Hadamard matrices. Proposition~\ref{prop:main-fourier-block} identifies the matrices
as belonging to the Tao and Karlsson sectors in this remaining case.

For a matrix $X \in \mathbb{T}^{3 \times 3}$, define its row cubic Gram invariant by
\begin{equation}
\tau_{\rm r}(X)
=(XX^\dagger)_{12}(XX^\dagger)_{23}(XX^\dagger)_{31}.
\label{eq:main-row-cubic-invariant}
\end{equation}
Multiplying the columns of $X$ by phases does not change $XX^\dagger$,
so we may normalize its first row to $(1,1,1)$. We define the normalized physical row-Gram fiber of $X$ by
\begin{equation}
\mathcal R_X=
\left\{
Y\in\T^{3\times3}:
Y_{1j}=1\ \text{for }1\leq j\leq3,\ 
YY^\dagger=XX^\dagger
\right\}.
\label{eq:main-normalized-fiber}
\end{equation}

\begin{proposition}[Infinite-fiber trichotomy]
\label{prop:main-infinite-fiber-trichotomy}
Let $X\in\T^{3\times3}$ be invertible. If $\mathcal R_X$ is infinite,
then at least one of
\begin{equation}
\begin{aligned}
&XX^\dagger=3I_3,\\
&\RePart\tau_{\rm r}(X)<0,\\
&X\text{ contains a }2\times2\text{ Hadamard submatrix}
\end{aligned}
\label{eq:main-fixed-gram-trichotomy}
\end{equation}
holds. These alternatives may overlap.
\end{proposition}

\begin{proof}[Proof sketch]
Fixing $XX^\dagger$ fixes the three off-diagonal row inner products. Supplemental Sec.~\ref{sec:local} analyzes the resulting six unit-circle equations by division-free elimination. It proves that, if the normalized physical fiber is infinite, the uncancelled equations must satisfy one of three algebraic conditions.
Translating these conditions back onto $X$ leads to three cases: $(i)$ the dependent case $XX^\dagger=3I_3$, $(ii)$ a nonsingular M\"obius case, possible only when $\operatorname{Re}\tau_{\rm r}(X)<0$, or $(iii)$ a common-zero case in which a normalized row has the form $(1,z,-z)$, exposing a $2\times2$ Hadamard submatrix. Because the elimination retains every vanishing-denominator and endpoint case, these alternatives exhaust the infinite-fiber possibilities.
\end{proof}

\begin{proposition}[Corner routing]
\label{prop:main-corner-routing}
Let $H$ be an order-six Hadamard matrix, written in the block form
of Eq.~\eqref{eq:blockform}, and suppose that
$[H]\notin\KK_6^{(3)}$. Then either all four displayed blocks are
order-three Hadamard matrices, or an allowed row and column permutation,
followed by dephasing, produces a finite-corner witness.
\end{proposition}

\begin{proof}
At the corner $E$, the blocks $B$ and $C$ belong to two normalized side fibers of $E$, and all candidates are invertible because their fixed Gram matrices are positive definite. If both fibers are finite, there is nothing to prove.

By replacing $H$ with $H^\dagger$ if necessary, we may assume that the horizontal side fiber is infinite.  Taking the adjoint exchanges the horizontal and vertical fibers, and the resulting conclusion transfers back by adjunction.

Proposition~\ref{prop:main-infinite-fiber-trichotomy}, applied to the actual adjacent block $B$, gives three alternatives. The $2\times2$ alternative is impossible because such a submatrix would also occur in $H$, placing $[H]$ in $\KK_6^{(3)}$ by Proposition~\ref{prop:publishedinputs}(1).

If $BB^\dagger=3I_3$, the complementary Gram identity
\begin{equation}
BB^\dagger=6I_3-EE^\dagger
\label{eq:main-complementary-row-gram}
\end{equation}
also gives $EE^\dagger=3I_3$. Since $E$ and $B$ are square, both are order-three Hadamard matrices, and the remaining block identities imply that $C$ and $D$ are as well. This gives the first alternative in the proposition.

It remains to consider $\RePart\tau_{\rm r}(B)<0$.
Exchange the two complementary column triples, so that the new block decomposition is
\begin{equation}
    \widetilde H=
    \begin{pmatrix}
    B&E\\
    D&C
    \end{pmatrix}.
\end{equation}
For every off diagonal element, Eq.~\eqref{eq:main-complementary-row-gram} gives $(EE^\dagger)_{ij}=-(BB^\dagger)_{ij}$.
Multiplying the three cyclic off-diagonal entries therefore gives
\begin{equation}
\tau_{\rm r}(E)=-\tau_{\rm r}(B)
\quad \text{and} \quad
\RePart\tau_{\rm r}(E)>0.
\label{eq:main-complementary-cubic-sign}
\end{equation}

It remains to check the vertical fiber at the new corner $B$. Define
\begin{equation}
    \tau_{\rm c}(X)
=(X^\dagger X)_{12}(X^\dagger X)_{23}(X^\dagger X)_{31}.
\end{equation}

The new horizontal side fiber at the corner $B$ contains $E$. If this fiber were infinite, Proposition~\ref{prop:main-infinite-fiber-trichotomy}, applied to $E$, would again give one of three alternatives.
Its order-three-Hadamard alternative would force the all-Hadamard-block case already treated, its negative-invariant alternative would contradict Eq.~\eqref{eq:main-complementary-cubic-sign}, and its $2\times2$ alternative would place $H$ in the Karlsson sector. Hence the new horizontal side fiber is finite.

Supplemental Sec.~\ref{sec:invariant}, Lemma~\ref{lem:rowcolumn}, gives $\RePart\tau_{\rm c}(X)=\RePart\tau_{\rm r}(X)$  for every $3\times3$ unimodular $X$. Thus $\RePart\tau_{\rm c}(B)<0$, whereas the complementary identity $B^\dagger B+D^\dagger D=6I_3$ gives $\RePart\tau_{\rm c}(D)>0$. 
If the new vertical fiber were infinite, Proposition~\ref{prop:main-infinite-fiber-trichotomy}, applied to the normalized adjoint $D^\dagger$, would give the same three excluded alternatives. Hence both fibers at $B$ are finite, so $B$ is a finite-corner witness.
\end{proof}

\begin{proposition}[Fourier-block alternative]
\label{prop:main-fourier-block}
If one block in a block partition of an order-six Hadamard matrix $H$ is
an order-three Hadamard matrix, then
\begin{equation}
[H]\in\KK_6^{(3)}
\qquad\text{or}\qquad
[H]\in\mathcal T_6.
\label{eq:main-fourier-block-alternative}
\end{equation}
\end{proposition}

\begin{proof}[Proof sketch]
Moving the Hadamard block to the upper left and applying the complementary Gram identities shows first that all four blocks are order-three Hadamard matrices. Since every order-three Hadamard matrix is equivalent to the Fourier matrix $F_3$, see Eq.~\ref{eq:fouriern}, block-preserving equivalence reduces the matrix to a Fourier-block normal form. Solving its remaining orthogonality equations gives either a matrix equivalent to Tao's isolated matrix or an $H_2$-reducible matrix. The latter belongs to Karlsson's sector by Proposition~\ref{prop:publishedinputs}(1). The normal form calculation and all endpoint cases are proved in Supplemental Sec.~\ref{sec:fourierblock}.
\end{proof}

\begin{proposition}[Absence of every finite-corner witness forces Karlsson or Tao]
\label{thm:finitecornerintro}
If an order-six complex Hadamard matrix $H$ has no equivalent dephased representative with a finite-corner witness, then
\begin{equation}
[H]\in\KK_6^{(3)}
\qquad\text{or}\qquad
[H]\in\mathcal T_6 \,.
\label{eq:failed-search-routing}
\end{equation}
\end{proposition}

\begin{proof}
Assume $[H]\notin\KK_6^{(3)}$. By
Proposition~\ref{prop:main-corner-routing}, either $H$ has an equivalent representative with a finite-corner witness, contrary to the hypothesis, or some block decomposition of $H$ consists of four order-three Hadamard blocks. In the latter case, Proposition~\ref{prop:main-fourier-block} places
$[H]$ in
$\mathcal T_6\cup\KK_6^{(3)}$. The Karlsson alternative is excluded by assumption, so
$[H]\in\mathcal T_6$. 
Together with the case $[H]\in\KK_6^{(3)}$, this proves the proposition.
\end{proof}

\begin{proposition}[Finite-corner witnesses for the Karlsson family]
\label{prop:karlsson-finite-corner}
Every equivalence class in $\KK_6^{(3)}$ has a finite-corner witness. Consequently,
\begin{equation}
\KK_6^{(3)}\subset\FCAtlas_6^{\mathrm{fc}},
\label{eq:karlsson-finite-corner-main}
\end{equation}
i.e. the Karlsson family is contained in the output of our branch-complete dilation procedure, Eq.~\eqref{eq:Gdefinition}.
\end{proposition}

\begin{proof}[Proof sketch]
We use Karlsson's explicit parametrization of the complete
family~\cite{Karlsson2011H2,Karlsson2011ThreeParameter}. The parametrization contains two phase-determining formulas. When neither formula degenerates, an exact calculation identifies a fixed positional $3\times3$ corner for which both normalized side fibers are finite and contain the actual adjacent blocks. This corner is therefore a finite-corner witness.

If exactly one formula degenerates, an equivalent parametrization of the same matrix makes the corresponding calculation regular and again supplies a finite-corner witness. It remains only to treat simultaneous degeneration of both formulas. This locus is precisely the affine-Fourier subfamily and its transpose. On this boundary, exact algebraic calculations show that at least one of six explicitly specified positional corners is a finite-corner witness. Hence every matrix in Karlsson's family has a finite-corner witness. The six corners and the complete three-case proof are given in Supplemental Sec.~\ref{sec:supp-karlsson-witnesses}~\cite{SupplementalMaterial}. The exact arithmetic identities used there are provided in the accompanying Lean repository~\cite{CardenesWuttig2026HadamardRepository}.
\end{proof}

\begin{proposition}[A finite-corner witness for Tao]
\label{prop:tao-finite-corner}
The leading $3\times3$ corner of the displayed Tao matrix in
Eq.~\eqref{eq:taomatrix} is a finite-corner witness. Consequently,
\begin{equation}
\mathcal T_6\subset\FCAtlas_6^{\mathrm{fc}}.
\label{eq:tao-finite-corner-main}
\end{equation}
\end{proposition}

\begin{proof}
For the leading decomposition of Eq.~\eqref{eq:taomatrix}, use
$(a,b,c,d)=(1,\omega,\omega,1)$, and thus
\begin{equation}
E=\begin{pmatrix}1&1&1\\1&1&\omega\\1&\omega&1\end{pmatrix}, 
\,\,
B=\begin{pmatrix}
1&1&1\\
\omega&\omega^2&\omega^2\\
\omega^2&\omega^2&\omega
\end{pmatrix},
\,\, C^T=B,
\label{eq:tao-seed-main}
\end{equation}
with $\det E=3\omega$ and $\det B=\det C=3$. The local infinite-fiber trichotomy (Proposition~\ref{prop:main-infinite-fiber-trichotomy}) shows that an infinite normalized fiber through this invertible block would force an order-three Hadamard block, a strictly negative cubic invariant, or an $H_2$ submatrix in the ambient Hadamard matrix. Direct multiplication gives
\begin{equation}
(BB^\dagger)_{12}=\omega-1\ne0,
\qquad
(BB^\dagger)_{23}=0.
\label{eq:tao-gram-check}
\end{equation}
The first equality shows that $B$ is not an order-three Hadamard block, and the second makes its row cubic invariant zero. Tao has no $H_2$ submatrix: every $2\times2$ cross ratio is a cubic root of unity, whereas the $H_2$ value is $-1$. Thus the horizontal fiber is finite. For the vertical fiber, transpose its normalized candidates: their fixed column-Gram condition becomes a fixed row-Gram condition, and the actual transposed block is $C^T=B$. The same calculation therefore proves that the vertical fiber is finite. Both fibers are nonempty because they contain the actual blocks $B$ and $C$. Finally Eq.~\eqref{eq:maincompletionformula} reconstructs the lower-right block of Eq.~\eqref{eq:taomatrix} exactly. The algorithm therefore retains the Tao matrix.
\end{proof}

We are now ready to prove the main classification result.

\begin{proof}[Proof of Theorem~\ref{thm:classification}]
Let $[H]\in\HH_6$ be an arbitrary equivalence class. Suppose, for contradiction, that no equivalent representative of $H$ has a finite-corner witness. Proposition~\ref{thm:finitecornerintro} then gives $[H]\in\KK_6^{(3)}$ or $[H]\in\mathcal T_6$. In the first case, Proposition~\ref{prop:karlsson-finite-corner} supplies a finite-corner witness, and in the second, Proposition~\ref{prop:tao-finite-corner} supplies one. Either conclusion contradicts the supposition. Hence every order-six complex Hadamard matrix is equivalent to a dephased matrix having a finite-corner witness.
\end{proof}

Theorem~\ref{thm:classification} has the following algorithmic consequence.

\begin{corollary}[Exact retained output]
\label{cor:atlas-classification}
The branch-complete dilation procedure is sound and exhaustive:
\begin{equation}
\FCAtlas_6^{\mathrm{fc}}=\HH_6.
\label{eq:completeclassification}
\end{equation}
\end{corollary}

\begin{proof}
Proposition~\ref{prop:algorithmvalid} proves soundness: $\FCAtlas_6^{\mathrm{fc}}\subseteq\HH_6$. Conversely, Theorem~\ref{thm:classification} gives every $[H]\in\HH_6$ a finite-corner witness. At that corner, the adjacent blocks $B,C$ occur in the finite candidate lists, and Eq.~\eqref{eq:maincompletionformula} recovers the actual unimodular fourth block. Hence
$[H]\in\FCAtlas_6^{\mathrm{fc}}$, proving
$\HH_6\subseteq\FCAtlas_6^{\mathrm{fc}}$. Thus, $\FCAtlas_6^{\mathrm{fc}}=\HH_6$.
\end{proof}

\section{Four-phase reconstruction}
\label{sec:finite-corner-reconstruction-space}
\label{sec:atlas}

The classification Theorem \ref{thm:classification} shows that every class in $\HH_6$ has an equivalent representative with a finite $3\times3$ corner. To describe the reconstruction locally, we choose such a corner, order its three rows and three columns -- every such choice is called a frame -- move them to the leading positions, and dephase. The resulting corner $E(a,b,c,d)$ has four phases $(a,b,c,d)\in\T^4$, see Eq.~\eqref{eq:singular-corner-form}, which provide the continuous coordinates of the reconstruction. In the following, we show that every class outside that of the isolated Tao matrix and a single explicit Karlsson matrix admits a frame on which a single universal quadratic--cubic reconstruction applies, see Theorem~\ref{thm:global-product-regular-escape}.
The full proof is given in Supplemental Material Sec.~\ref{sec:supp-generic-cover} and sketched in this section. It additionally invokes two published results in References~\cite{BondalZhdanovskiy2016} and~\cite{BasuPollackRoy2006}.

\subsection{The regular product cover}
\label{subsec:quadratic-cubic-regular-cover}
For a Laurent polynomial in the scalar variable $z$,
\begin{equation}
f(z)=\sum_{m\in\mathbb Z}c_mz^m,
\end{equation}
define
\begin{equation}
f^\#(z)=\sum_{m\in\mathbb Z}\overline{c_m}z^{-m},
\qquad
\left(\frac{f}{g}\right)^\#=\frac{f^\#}{g^\#}\quad(g\neq0).
\label{eq:sharp-definition-main}
\end{equation}
For $z\in\T$ we have $f^\#(z)=\overline{f(z)}$ and $z^\#=z^{-1}=\overline z$.

We define $(s_h,t_h,r_h)$ from $E$ as in Eq.~\eqref{eq:trace-data-main}. The corresponding $(s_v,t_v,r_v)$ are obtained from $E^{\mathsf T}$, see Supplemental Eq.~\eqref{eq:supp-vertical-data}. 
For either triple $(s,t,r)$, let $A_{s,t,r}$ and $B_{s,t,r}$ be the uncancelled companion polynomials defined in Supplemental Eqs.~\eqref{eq:supp-A-explicit}--\eqref{eq:supp-B-explicit}. We write $A_h=A_{s_h,t_h,r_h}$ and $B_h=B_{s_h,t_h,r_h}$, with $A_v,B_v$ defined analogously. For horizontal candidates, $x_j$ and $y_j$, reversible elimination on the regular locus gives
\begin{equation}
A_h(x_j)+B_h(x_j)y_j=0,
\label{eq:regular-companion-relation}
\end{equation}
where 
\begin{equation}
    y_j=-\frac{A_h(x_j)}{B_h(x_j)}\quad\text{when }B_h(x_j)\neq0.
\end{equation}
Substituting the companion formula into $y_jy_j^\#=1$ gives
\begin{equation}
A_h(x_j)A_h(x_j)^\#-B_h(x_j)B_h(x_j)^\#=0
\label{eq:horizontal-companion-phase-condition}
\end{equation}
provided that $B_h(x_j)\neq0$. 
Because $x_j\in\T$ and hence $x_j\neq0$, this equation remains zero after multiplication by $x_j^3$, i.e. 
\begin{equation}x_{j}^3\bigl(A_h(x_j)A_h(x_j)^\#-B_h(x_j)B_h(x_j)^\#\bigr) = 0.
\label{eq:horizontal-fundamental-sextic}
\end{equation}
This defines the sextic
\begin{equation}
  \Phi_h(x_{j})  =\sum_{k=0}^{6}c_{k,h}x_j^k,
\end{equation}
for which the three horizontal coordinates $x_1,x_2,x_3$ are roots.
Let
\begin{equation}
u=x_1x_2x_3
\label{eq:u}
\end{equation}
be their product. Combining $\Phi_h(x_j) = 0$ for all $j$ with the fixed-Gram equation gives the denominator-free product equation
\begin{equation}
(1+s_hs_h^\#)c_{6,h}u^2+c_{3,h}u+(1+s_hs_h^\#)c_{0,h}=0,
\end{equation}
see Supplemental Eq.~\eqref{eq:supp-product-quadratic}.
On the seed torus, $1+s_hs_h^\#=1+|s_h|^2>0$. Therefore, when $c_{6,h}\neq0$, division by $(1+s_hs_h^\#)c_{6,h}$ gives the quadratic
\begin{equation}
u^2-U_hu+V_h=0
\label{eq:horizontal-product-quadratic}
\end{equation}
with
\begin{equation}
    U_h=-\frac{c_{3,h}}{(1+s_hs_h^\#)c_{6,h}},
\,\,
V_h=\frac{c_{0,h}}{c_{6,h}}.
\end{equation}
Denote the two roots of the quadratic in Eq.~\ref{eq:horizontal-product-quadratic} by $u_+$ and $u_-$. For either choice $u=u_\pm$, the corresponding coordinates $x_1,x_2,x_3$ are the roots of
\begin{equation}
q_{s_h,u}(x)=x^3-s_hx^2+us_h^\#x-u.
\label{eq:horizontal-coordinate-cubic}
\end{equation}
Supplemental Lemma~\ref{lem:supp-product-factor} proves the exact factorization,
also recorded in Supplemental Eq.~\eqref{eq:supp-sextic-factor}:
\begin{equation}
\Phi_h(x)=c_{6,h}q_{s_h,u_+}(x)q_{s_h,u_-}(x).
\end{equation}

The horizontal reconstruction therefore proceeds by solving the quadratic for $u$, solving the associated cubic for $x_1,x_2,x_3$, and then determining each paired coordinate from
\begin{equation}
y_j=-\frac{A_h(x_j)}{B_h(x_j)}.
\end{equation}
The same procedure can be applied to $E^{\mathsf T}$ to get the relevant coordinates of $C$. 

\begin{definition}[Product-regular lift]
\label{def:product-regularity}
An algebraic lift of a framed seed consists of compatible horizontal and vertical product roots, their coordinate roots, and the paired coordinates determined by the companion relations, all satisfying the parent fixed-Gram equations. For the selected horizontal and vertical product roots $u$ (Eq.~\ref{eq:u}) and $m$, respectively, define
\begin{equation}
\begin{aligned}
q_h(x)&:=x^3-s_hx^2+us_h^\#x-u,\\
q_v(p)&:=p^3-s_vp^2+ms_v^\#p-m,
\end{aligned}
\end{equation}
and set $\delta_h:=r_hr_h^\#-t_ht_h^\#$ and $\delta_v:=r_vr_v^\#-t_vt_v^\#$. Define the following discriminant $\operatorname{Disc}$ and resultant $\operatorname{Res}$ for both polynomials
\begin{equation}
\begin{aligned}
\operatorname{Disc}q_h&:=\operatorname{Disc}_x\bigl(q_h(x)\bigr),&
R_h&:=\operatorname{Res}_x\bigl(q_h(x),B_h(x)\bigr),\\
\operatorname{Disc}q_v&:=\operatorname{Disc}_p\bigl(q_v(p)\bigr),&
R_v&:=\operatorname{Res}_p\bigl(q_v(p),B_v(p)\bigr).
\end{aligned}
\label{eq:product-regularity-auxiliary-guards}
\end{equation}
The lift is \emph{product regular} when the following holds:
\begin{equation}
\begin{aligned}
&\det(E)\det(B)(\det B)^\#\det(C)(\det C)^\#\\
&\qquad{}\times c_{6,h}c_{6,v}\delta_h\delta_v\\
&\qquad{}\times\operatorname{Disc}(q_h)\operatorname{Disc}(q_v)R_hR_v\neq0.
\end{aligned}
\label{eq:product-regularity-guards}
\end{equation}
The star-stable determinant guards ensure that the matrices and formal adjoints required by the reconstruction are invertible. On the physical torus, $(\det B)^\#=\overline{\det B}$ and similarly for $C$, so these guards impose no additional physical restriction. The factors $c_{6,h},c_{6,v},\delta_h,\delta_v$ keep the product and companion equations nondegenerate. The discriminant guards ensure that the coordinate cubics have distinct roots, and the resultant guards ensure that the companion denominators are nonzero at those roots. 
The discriminant of Eq.~\ref{eq:horizontal-product-quadratic} is not included because its vanishing describes the regular ramification studied below. 
If a guard fails and the physical fiber is finite, the uncancelled incidence equations retain all of it. If the physical fiber is positive-dimensional, the global routing argument beginning with Proposition~\ref{prop:main-corner-routing} applies.
\end{definition}

The reciprocal symmetry of $\Phi_h$ imposes the coefficient identities
\begin{equation}
V_hV_h^\#=1,
\qquad
U_h=V_hU_h^\#.
\label{eq:self-inversive-product-data}
\end{equation}
The normalized real discriminants of the horizontal and vertical product quadratics are
\begin{equation}
\begin{aligned}
\omega_{\rm n}&:=U_hU_h^\#-4=|U_h|^2-4,\\
\omega_{{\rm n},v}&:=U_vU_v^\#-4=|U_v|^2-4,
\end{aligned}
\label{eq:normalized-product-discriminants}
\end{equation}
respectively.
On the product-regular locus, $\delta_h=r_hr_h^\#-t_ht_h^\#$ and its vertical analogue $\delta_v$ are real and nonzero. After removing the corresponding nonzero square factors, the two discriminants reduce to the same real function $\Omega$:
\begin{equation}
\Omega=\frac{|c_{6,h}|^2}{\delta_h^2}\,\omega_{\rm n}
=\frac{|c_{6,v}|^2}{\delta_v^2}\,\omega_{{\rm n},v}.
\label{eq:omega-residual-discriminant-bridge}
\end{equation}
On the product-regular locus, both prefactors are strictly positive. Therefore, $\omega_{\rm n}$ and $\omega_{{\rm n},v}$ have the same sign as $\Omega$ and vanish on the same set of seeds. The horizontal product roots are
\begin{equation}
u_\pm=\frac{U_h\pm\sqrt{U_h^2-4V_h}}{2}.
\label{eq:horizontal-product-roots}
\end{equation}
We define their physical separation as
\begin{equation}
    |u_+-u_-|^2=-\omega_{\rm n}\qquad(\omega_{\rm n}\leq0).
\end{equation}
The two horizontal product roots $u_\pm$ are distinct phases for $\omega_{\rm n}<0$, coincide for $\omega_{\rm n}=0$, and are nonunimodular for $\omega_{\rm n}>0$.

Although the horizontal and vertical product quadratics each have two roots, these roots cannot be chosen independently. Requiring every entry of the forced lower-right block to have modulus one pairs each horizontal root $u_\pm$ with a unique vertical root $m_\pm$ according to
\begin{equation}
m_\pm=\frac{U_v-\kappa(2u_\pm-U_h)}{2},
\qquad
\kappa=\frac{\delta_vc_{6,h}}{\delta_hc_{6,v}}.
\label{eq:affine-minus-matching-main}
\end{equation}
After simultaneous permutations of the three paired horizontal coordinates and the three paired vertical coordinates are identified, each compatible pair $(u_\pm,m_\pm)$ defines one \emph{product sheet}. 

\begin{figure}[t!]
\centering
\includegraphics[width=\columnwidth]{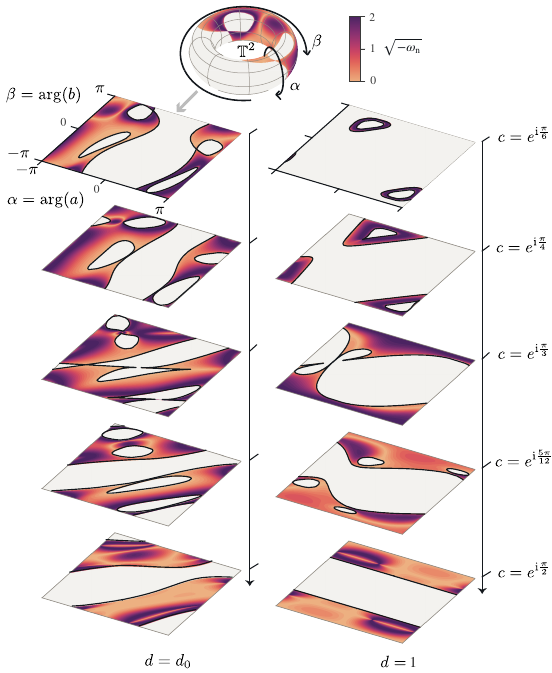}
\caption{\textbf{Geometry of the regular four-phase reconstruction on different two-dimensional slices.} Value of the product-regular witness $\sqrt{-\omega_{\rm n}}$ for varying initial seeds $a,b,c,d$. The two stacks correspond to two different $d$ values. Note that $d_{0} = \frac{1 + {\rm i}\xi_{0}}{1 - {\rm i}\xi_{0}} \approx -0.674 - 0.739 \mathrm{i}$, where $\xi_{0}$ is the real root of Eq.~\eqref{eq:ramification-seed-polynomial}.
The slices within a stack are at the marked values of $c$, which are regularly spaced in phase space. 
The colored region corresponds to $\omega_{\rm n} \leq 0$ where a product regular completion from the initial seed is possible by solving two quadratic and two cubic equations. When $\omega_{\rm n} < 0$, there are two physical product sheets from the two quadratic solutions, separated by $\sqrt{-\omega_{\rm n}}=|u_+-u_-|$. Gray regions denote $\omega_{\rm n}>0$, where no product-regular completion exists in this frame. There are only two classes of order-six complex Hadamard matrices that cannot be constructed in this way: the Tao class and $\{[H_{x}]\}$.}
\label{fig:regular-seed-geometry}
\end{figure}

\begin{theorem}[Regular seed-domain theorem]
\label{thm:regular-seed-domain}
A product-regular lift above a seed $(a,b,c,d)\in\T^4$ is physical (i.e. all reconstructed variables lie on the unit circle, so that the resulting completion is an actual complex Hadamard matrix) if and only if
\begin{equation}
\omega_{\rm n}(a,b,c,d)\leq0.
\label{eq:regular-physical-seed-domain}
\end{equation}
Whenever both matched branches satisfy the guards in
Eq.~\eqref{eq:product-regularity-guards}, $\omega_{\rm n}<0$
gives two distinct physical product sheets, while at $\omega_{\rm n}=0$ they
coalesce into one physical sheet. On each physical sheet,
Eq.~\eqref{eq:affine-minus-matching-main} determines the vertical root from
the horizontal root, and the forced fourth block is unimodular.
\end{theorem}

\begin{proof}[Proof sketch]
Let $u$ be the product root selected by the lift, and let $u'=V_h/u$ be the second root of the product quadratic. If the lift is physical, then $u$ is a product of phases and is therefore itself a phase. Since $V_h$ is also a phase, $u'$ is a phase as well. The sum of the two roots is $U_h=u+u'$, so $|U_h|\leq2$ and hence $\omega_{\rm n}=|U_h|^2-4\leq0$.

For the reverse implication, suppose that $\omega_{\rm n}\leq0$. The self-inversive product quadratic then has phase roots, distinct when $\omega_{\rm n}<0$ and coincident when $\omega_{\rm n}=0$. Choose either root $u$. The corresponding coordinate cubic is an odd-degree self-inversive polynomial, so it has at least one phase root $x$. Since this cubic divides $\Phi_h$, one has $\Phi_h(x)=0$ and therefore $|A_h(x)|=|B_h(x)|$. The resultant guard ensures that $B_h(x)\neq0$, so
\begin{equation}
y=-\frac{A_h(x)}{B_h(x)}
\end{equation}
is also a phase. Thus $(1,x,y)^{\mathsf T}$ is one phase column of $B$.

Supplemental Lemma~\ref{lem:supp-complement-positivity} and the parent determinant identity show that the remaining Gram matrix is positive semidefinite of rank two. Haagerup's two-phase decomposition then supplies the other two phase columns of $B$~\cite{Haagerup1997,Szollosi2012FourParameter}. Applying the same argument to $E^{\mathsf T}$ constructs $C$. Supplemental Proposition~\ref{prop:supp-generic-flatness} proves the affine-matching flatness identities in the product function field, and Supplemental Lemma~\ref{lem:supp-regular-localization} shows that the displayed guards specialize those identities at every product-regular lift. Hence the forced block $D$ is unimodular. The argument includes $\omega_{\rm n}=0$ directly.
\end{proof}

\subsection{Global reach of the regular cover}
\label{subsec:global-reach-regular-cover}

A finite-corner presentation of $[H]$ is an equivalent dephased representative obtained by moving an ordered choice of three rows and three columns to the first three positions, such that the resulting leading $3\times3$ block is a finite-corner witness.
For a product-regular finite-corner presentation,
Theorem~\ref{thm:regular-seed-domain} identifies the physical region.

We define
\begin{equation}
\mathcal P_6:=\{[H]\in\HH_6:\text{$[H]$ admits a product-regular frame}\}.
\label{eq:main-product-regular-class-locus}
\end{equation}
The exceptional Karlsson class appearing below is represented through the matrix
\begin{equation}
H_\times=
\begin{pmatrix}
1&1&1&1&1&1\\
1&-1&1&\I&-1&-\I\\
1&-1&-1&-\I&\I&1\\
1&\I&-\I&-1&-\I&\I\\
1&1&\I&-\I&-1&-1\\
1&-\I&-1&\I&1&-1
\end{pmatrix}.
\label{eq:main-product-exceptional-karlsson-matrix}
\end{equation}

\begin{theorem}[Exact product-regular reach]
\label{thm:atlas-main}
\label{thm:global-product-regular-escape}
The only order-six classes without a product-regular frame are Tao's class
and the single Karlsson class represented by $H_\times$:
\begin{equation}
\HH_6\setminus\mathcal P_6
=\mathcal T_6\,\dot\cup\,\{[H_\times]\}.
\label{eq:global-product-regular-reach}
\end{equation}
Theorem~\ref{thm:regular-seed-domain} characterizes the physical domain of
each such presentation.
\end{theorem}

\begin{proof}[Proof sketch]
Outside the Karlsson and Tao sectors, the counting argument in Supplemental
Sec.~\ref{sec:supp-product-escape} proves product regularity: the assumption
of no regular frame gives at least $100$ positive-dependent obstructions,
whereas the exact incidence bound permits at most $80$.
Supplemental Proposition~\ref{prop:supp-karlsson-product-exceptional-singleton}
shows that every Karlsson class except $[H_\times]$ has a product-regular
frame and that $[H_\times]$ does not, while
Proposition~\ref{prop:supp-tao-product-exceptional} proves that Tao does not.
The exact certificates used in these arguments are listed alongside the
corresponding proofs in the Supplemental Material.
\end{proof}

Theorems~\ref{thm:regular-seed-domain} and~\ref{thm:global-product-regular-escape} reveal the geometry of the regular reconstruction. The four seed phases are the continuous base coordinates, the product quadratic supplies the two sheets, and changing a corner or frame changes only the presentation. Figure~\ref{fig:regular-seed-geometry} shows a two-dimensional slice of $\omega_{\rm n}$ for fixed $c$ and $d$. At product-regular points, $\omega_{\rm n}<0$ gives two physical sheets and the black contour $\omega_{\rm n}=0$ is where they merge. The gray region $\omega_{\rm n}>0$ has no product-regular completion in this frame, but does not exclude a guard-failing completion or a presentation through another frame.

\subsection{Ramification in the class space}
\label{subsec:intrinsic-loci-product-ramification}

Write the seed as $(a,b,c,d)=(e^{\I\alpha},e^{\I\beta},e^{\I\gamma},e^{\I\delta})$ and regard $\omega_{\rm n}$ as a real-valued function of $(\alpha,\beta,\gamma,\delta)$. Its differential is the cotangent vector
\begin{equation}
\mathrm d\omega_{\rm n}=\frac{\partial\omega_{\rm n}}{\partial\alpha}\,\mathrm d\alpha+\frac{\partial\omega_{\rm n}}{\partial\beta}\,\mathrm d\beta +\frac{\partial\omega_{\rm n}}{\partial\gamma}\,\mathrm d\gamma+\frac{\partial\omega_{\rm n}}{\partial\delta}\,\mathrm d\delta.
\label{eq:normalized-discriminant-differential-main}
\end{equation}
Thus $\mathrm d\omega_{\rm n}\neq0$ means that at least one of these four partial derivatives is nonzero. At a product-regular seed satisfying $\omega_{\rm n}=0$ and $\mathrm d\omega_{\rm n}\neq0$, the real implicit-function theorem makes the zero set of $\omega_{\rm n}$ a smooth real three-dimensional hypersurface in the four-dimensional seed torus $\mathbb{T}^4$. Its class-space image is
\begin{equation}
\begin{aligned}
\mathcal R_{6,\mathrm{prod}}=\{[H]\in\HH_6:\ &\text{some representative of $[H]$}\\
&\text{has a product-regular frame}\\
&\text{with $\omega_{\rm n}=0$}\}.
\end{aligned}
\label{eq:product-ramification-locus}
\end{equation}
This existential definition does not exclude Karlsson by convention. Since $\mathcal R_{6,\mathrm{prod}}\subseteq\mathcal P_6$, Theorem~\ref{thm:global-product-regular-escape} shows that it is disjoint from Tao and $[H_\times]$. Its possible intersection with the remaining Karlsson classes is not determined here and remains an open problem.

\begin{figure}[t!]
\centering
\includegraphics[width=\columnwidth]{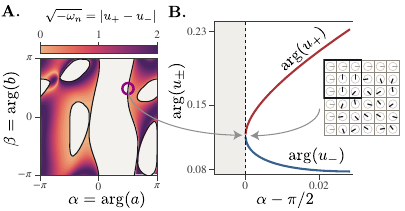}
\caption{\textbf{Ramified Hadamard completion at a branch collision.} \textbf{A.} Value of $\sqrt{-\omega_n} = \vert u_{+} - u_{-} \vert$ for varying ${\rm arg}(a)$ and ${\rm arg}(b)$ at fixed $c = {\rm i}$ and $d = d_{0} = \frac{1 + {\rm i}\xi_{0}}{1 - {\rm i}\xi_{0}}$, where $\xi_{0}$ is the unique real root of Eq.~\eqref{eq:ramification-seed-polynomial}. At the circled seed, the two branches of the quadratic in Eq.~(\ref{eq:horizontal-product-quadratic}) meet and $\omega_{n} = 0$ (black lines).  \textbf{B.} Arguments of the two roots of the relevant quadatics versus $\alpha$. They coalesce at $\alpha = \pi/2$ and $\beta = \pi/2$. The phase portrait of the ramified, product-regular completed Hadamard at this branch collision point is shown. Numerical diagnostics for this representative, reported in Supplemental Table~\ref{tab:supp-ramification-bounds}, distinguish it from the Karlsson, Tao, and four-$3\times3$-circulant-block
constructions.}
\label{fig:ramifiedHadamard}
\end{figure}

\paragraph*{A representative ramification seed.}
Let $\xi_0$ be the unique real root of
\begin{equation}
\begin{aligned}
P(\xi)={}&16\xi^7+120\xi^6+428\xi^5+952\xi^4\\
&+1363\xi^3+1231\xi^2+664\xi+176.
\end{aligned}
\label{eq:ramification-seed-polynomial}
\end{equation}
The Cayley seed
\begin{equation}
(a,b,c,d)=\left(\I,\I,\I,\frac{1+\I\xi_0}{1-\I\xi_0}\right),
\label{eq:ramification-cayley-seed}
\end{equation}
lies on the sheet-collision set $\omega_{\rm n}=0$, which is visualized in Fig.~\ref{fig:ramifiedHadamard}. Our high-precision numerics finds a product-regular completion and separates this seed from the Karlsson, Tao, and four-$3\times3$-circulant-block constructions. The numerical diagnostics are reported in Supplemental Table~\ref{tab:supp-ramification-bounds}.

\section{Proof of Sz\"oll\H{o}si's conjecture}
\label{sec:szollosi-conjecture}
The purpose of this section is to prove Sz\"oll\H{o}si's Conjecture~4.2~\cite{Szollosi2012FourParameter}, clarifying the relationship between $G_6^{(4)}$, the notation he used for the generic family produced by his Construction~3.1, and $\FCAtlas_6^{\mathrm{fc}}$.

\begin{definition}[Sz\"oll\H{o}si's construction output]
\label{def:historical-G-output}
Following Sz\"oll\H{o}si, let $G_6^{(4)}$ denote the non-Karlsson, non-Tao classes returned by Steps~1--8 of his Construction~3.1 \cite{Szollosi2012FourParameter}. The definition includes every unimodular root and every finite common-root branch retained in Case~2 of his detailed construction. It excludes the identically vanishing fundamental-polynomial case for which that construction is declared to fail. Equivalently, by his Theorem~4.1, it is the output obtained from corners having finite nonempty
normalized invertible candidate sets. This is Sz\"oll\H{o}si's construction
output, not the output of our branch-complete procedure by definition. The
exclusion of Karlsson and Tao records his sector convention rather than an
obstruction inherent in fixed-corner dilation.
\end{definition}

At a product-regular corner, all elimination steps are reversible, so
Sz\"oll\H{o}si's Construction~3.1 and our quadratic--cubic
reconstruction produce the same candidate blocks and Hadamard
completions.

\begin{proposition}[Identification of Sz\"oll\H{o}si's sector]
\label{prop:historical-output-identification}
\begin{equation}
G_6^{(4)}
=
\FCAtlas_6^{\mathrm{fc}}\setminus
\bigl(\KK_6^{(3)}\cup\mathcal T_6\bigr).
\label{eq:historical-output-identification}
\end{equation}
\end{proposition}

\begin{proof}
Let $[H]\in G_6^{(4)}$. By
Definition~\ref{def:historical-G-output}, a representative of
$[H]$ is produced by Sz\"oll\H{o}si's Construction~3.1 from finite
normalized candidate sets. These candidates satisfy the parent
fixed-Gram equations, and the retained completion uses the same
forced-block formula and unimodularity tests as our branch-complete
procedure. Hence $[H]\in\FCAtlas_6^{\mathrm{fc}}$. Since
$G_6^{(4)}$ excludes the Karlsson and Tao sectors by definition,
this proves
\[
G_6^{(4)}
\subseteq
\FCAtlas_6^{\mathrm{fc}}\setminus
\bigl(\KK_6^{(3)}\cup\mathcal T_6\bigr).
\]

Conversely, let
\[
[H]\in
\FCAtlas_6^{\mathrm{fc}}\setminus
\bigl(\KK_6^{(3)}\cup\mathcal T_6\bigr).
\]
By the definition of $\FCAtlas_6^{\mathrm{fc}}$, some representative
of $[H]$ is recovered from a finite-corner witness.
Sz\"oll\H{o}si's Theorem~4.1 states that his Construction~3.1
exhausts all Hadamard completions at such a corner, including $H$.
Since $[H]$ is neither Karlsson nor Tao, it follows that
$[H]\in G_6^{(4)}$.
\end{proof}

\begin{corollary}[Sz\"oll\H{o}si's Conjecture~4.2]
\label{cor:szollosi-conjecture}
The three historical sectors form an exhaustive list:
\begin{equation}
\HH_6
=
G_6^{(4)}\cup\KK_6^{(3)}\cup\mathcal T_6,
\label{eq:szollosi-conjecture-list}
\end{equation}
and the three sectors are pairwise disjoint. Equivalently, every
class outside the Karlsson and Tao sectors is recovered by
Sz\"oll\H{o}si's Construction~3.1.
\end{corollary}

\begin{proof}
Corollary~\ref{cor:atlas-classification} gives
\[
\FCAtlas_6^{\mathrm{fc}}=\HH_6.
\]
Substitution into
Proposition~\ref{prop:historical-output-identification} yields
\[
G_6^{(4)}
=
\HH_6\setminus
\bigl(\KK_6^{(3)}\cup\mathcal T_6\bigr),
\]
which proves exhaustiveness and shows that $G_6^{(4)}$ is disjoint
from the other two sectors.

It remains only to separate Tao from Karlsson. Every cross ratio of
Tao's matrix is a cubic root of unity, whereas an $H_2$-reducible
matrix has a cross ratio equal to $-1$. Hence
$\mathcal T_6\cap\KK_6^{(3)}=\varnothing$.
\end{proof}

\section{Outlook}
\label{sec:outlook}
We have completed an explicit finite-corner classification of order-six complex Hadamard matrices. Specifically, we have proved that a representative of every equivalence class can be constructed from our branch-complete dilation procedure starting from four initial seeds in $\mathbb{T}^4$. This dilation procedure provides the missing step from Sz\"oll\H{o}si's Construction~3.1 and, as our procedure also covers the Karlsson and Tao sectors within the same framework, has allowed us to prove a stronger version of his Conjecture 4.2~\cite{Szollosi2012FourParameter}. 

Moreover, outside the Tao sector and the sector represented by a single explicit Karlsson matrix, we show that every equivalence class admits a \textit{product-regular} witness. In this case, the branch-complete dilation procedure simplifies to solving only two quadratic and two cubic equations, whose discriminant provides the geometric criterion for product regularity.

A direct experimental application of this work is multiphoton interference in balanced six-mode networks. After normalization, every order-six Hadamard matrix is a lossless six-port unitary with uniform single-particle probabilities, while its internal phases remain observable in multiparticle correlations. Such transformations are already realized in programmable photonic circuits, where two-photon experiments distinguish networks that are identical at the single-particle level~\cite{Carolan2015Universal,Laing2012BerryPhase,Tichy2010ZeroTransmission,Crespi2015Sylvester,Crespi2016FFT,Shchesnovich2015Partial}. Our results supply a complete, finite-chart design space for these balanced six-port transformations, extending searches for suppression laws and phase-sensitive interference signatures beyond selected families. The same Hadamard ingredient recurs as a coherence-generating gate~\cite{Yao2016MaximalCoherence} and as a building block of dual-unitary circuits, a standard testbed for entanglement growth and thermalization~\cite{Gutkin2020DualUnitary,Claeys2022StateDesigns}. In each setting our theorem classifies the Hadamard ingredient, while device-specific constraints remain explicit conditions on this complete space.

The classification does not settle the problem of mutually unbiased bases in $\mathbb C^6$, but it removes the need to search over an unclassified set of Hadamard matrices. After fixing one basis as standard, every transition matrix to a second unbiased basis satisfies $[\sqrt6\,U]\in\FCAtlas_6^{\mathrm{fc}}$.
With three or more bases, each basis is shared by several transition matrices at once, so the relabelings and phase choices used to reach standard form cannot be applied independently to each matrix: fixing one shared basis constrains every matrix built from it. The remaining task is to impose these joint compatibility conditions directly on the classified space constructed here~\cite{WoottersFields1989,BengtssonEtAl2007,HorodeckiRudnickiZyczkowski2022,McNultyWeigert2026,MatolcsiEtAl2026MUBTriplets}. 
This also gives a controlled domain for exact or certified searches in quantum-state determination, cryptography, and related measurement-design
problems~\cite{WoottersFields1989,Cerf2002,Englert2001,Aravind2003}. An important structural question is whether analogous finite-corner arguments extend to complex Hadamards in higher orders.

\section*{Acknowledgements}
The authors thank Andy Millis and Uli Schollwöck for helpful discussions on this manuscript.
We thank Gavin Crooks for suggesting that the structural inputs be derived internally in the Lean formalization.
M.C.W. acknowledges the hospitality of the Center for Computational Quantum Physics at the Flatiron Institute. The Flatiron Institute is a division of the Simons Foundation.

OpenAI's ChatGPT Sol 5.6 Pro and Codex 5.6 Sol and Claude's Opus 5.0 were used interactively in the development of this work. The LLMs assisted in proposing and testing proof strategies, searching for gaps and possible counterexamples, checking intermediate algebra, drafting and debugging symbolic and formal-verification code, and revising the manuscript. 
The authors specified the definitions and target statements, supplied and checked the relevant literature, selected the proof strategies, reviewed proposed arguments and counterexamples, and decided which derivations and computations to retain.
Several model-generated proposals were found to be incomplete or incorrect and were either discarded or repaired. 

The authors independently checked all retained LLM-assisted material and verified the computer-assisted claims described in the Supplemental Material. The authors take full responsibility for the mathematical content and presentation of this work.

\section*{Data availability}
No external datasets were used. The accompanying repository at \url{github.com/mateocardeneswuttig/all_hadamard_matrices_in_dimension_six} contains the Lean~4 formalization and its pinned dependencies~\cite{CardenesWuttig2026HadamardRepository}, together with
the exact and interval certificates supporting the computer-assisted proofs for the results in
in Sec.~IV~\cite{CardenesWuttig2026HadamardRepository}. The public entry point \path{Hadamard6/PaperTheorem.lean} verifies Theorem~\ref{thm:classification} and Corollary~\ref{cor:atlas-classification}, including the two structural results stated in Proposition~\ref{prop:publishedinputs}. The \path{certificates/} directory contains the certificates, and \path{certificates/PAPER_CLAIM_AUDIT.md} maps each computer-assisted claim to its verifier.

\bibliography{ref}
\fi

\clearpage
\onecolumngrid

\setcounter{section}{0}
\setcounter{subsection}{0}
\renewcommand{\thesection}{\Roman{section}}

\setcounter{equation}{0}
\numberwithin{equation}{section}
\renewcommand{\theequation}{S.\arabic{section}.\arabic{equation}}

\setcounter{theorem}{0}
\renewcommand{\thetheorem}{S.\arabic{theorem}}

\makeatletter
\renewcommand{\theHsection}{supp.\arabic{section}}
\renewcommand{\theHequation}{supp.\arabic{section}.\arabic{equation}}
\providecommand{\theHtheorem}{}
\renewcommand{\theHtheorem}{supp.\arabic{theorem}}
\makeatother

\makeatletter
\newcommand{\LLM@suppsectionnumber}{}
\newcommand{\LLM@suppsubsectionnumber}{}
\newcommand{\LLM@sectionnumberline}[1]{%
  \xdef\LLM@suppsectionnumber{#1}%
  \hb@xt@\@tempdima{\bfseries #1\hfil}}
\newcommand{\LLM@subsectionnumberline}[1]{%
  \xdef\LLM@suppsubsectionnumber{\LLM@suppsectionnumber.#1}%
  \hb@xt@\@tempdima{\LLM@suppsubsectionnumber\hfil}}
\newcommand{\LLM@subsubsectionnumberline}[1]{%
  \hb@xt@\@tempdima{\LLM@suppsubsectionnumber.#1\hfil}}
\newcommand{\supplementaltableofcontents}{%
  \begin{center}\large\bfseries{Supplemental contents}\end{center}\vspace{0.25em}%
  \begingroup
  \setcounter{tocdepth}{3}%
  \setlength{\parskip}{0pt}%
  \renewcommand*\l@section[2]{%
    \addpenalty{-\@highpenalty}\addvspace{0.45em}%
    \begingroup\bfseries
    \let\numberline\LLM@sectionnumberline
    \@dottedtocline{1}{0em}{3.2em}{##1}{##2}%
    \endgroup}%
  \renewcommand*\l@subsection[2]{%
    \begingroup
    \let\numberline\LLM@subsectionnumberline
    \@dottedtocline{2}{1.5em}{4.2em}{##1}{##2}%
    \endgroup}%
  \renewcommand*\l@subsubsection[2]{%
    \begingroup\itshape
    \let\numberline\LLM@subsubsectionnumberline
    \@dottedtocline{3}{3.5em}{5.0em}{##1}{##2}%
    \endgroup}%
  \@starttoc{stoc}%
  \endgroup}
\let\LLM@addcontentsline\addcontentsline
\renewcommand{\addcontentsline}[3]{%
  \def\LLM@target{#1}%
  \def\LLM@toc{toc}%
  \ifx\LLM@target\LLM@toc
    \LLM@addcontentsline{stoc}{#2}{#3}%
  \else
    \LLM@addcontentsline{#1}{#2}{#3}%
  \fi}
\makeatother
\begin{center}
\large\bfseries Supplementary Material
\end{center}
\newcommand{\subproofheading}[1]{%
  \par\medskip\noindent\textbf{#1}\par\smallskip}

{The supplementary material is divided into two main sections. In section~\ref{sec:supp-classification-proof}, we prove the finite-corner classification using the two published structural results isolated in Proposition~\ref{prop:publishedinputs}. Section~\ref{sec:supp-reconstruction-geometry} supplies the detailed proofs for the four-phase reconstruction.}

The Lean formalization we have also built~\cite{CardenesWuttig2026HadamardRepository} verifies Theorem~\ref{thm:classification} and
Corollary~\ref{cor:atlas-classification}, including the two structural
reductions stated in Proposition~\ref{prop:publishedinputs}. The relevant files are described in the first row of Table~\ref{tab:supp-certificate-index}.

The exact and interval certificates listed in the subsequent rows in Table~\ref{tab:supp-certificate-index}
support the separate reconstruction geometry of Sec.~IV. Each row gives
a concrete repository entry point. The complete claim-to-file map, pinned
dependencies, and reproduction commands are in
\path{certificates/PAPER_CLAIM_AUDIT.md}~\cite{CardenesWuttig2026HadamardRepository}.
Finally, the comparison
with Sz\"oll\H{o}si's Construction~3.1 is proved in Sec.~V of the main text
from Theorem~\ref{thm:classification} and Sz\"oll\H{o}si's published
fixed-corner completeness theorem.

\begin{table}[!ht]
\caption{Formal-verification and certificate entry points.}
\label{tab:supp-certificate-index}
\centering
\footnotesize
\renewcommand{\arraystretch}{1.12}
\begin{tabular}{|p{0.22\textwidth}|p{0.35\textwidth}|p{0.35\textwidth}|}
\hline
Result & Repository entry point & What is verified \\
\hline
Classification and Karlsson coverage
& \path{Hadamard6/PaperTheorem.lean}\newline
  \path{Hadamard6/KarlssonContainment.lean}
& Theorem~\ref{thm:classification}, Corollary~\ref{cor:atlas-classification},
  Proposition~\ref{prop:karlsson-finite-corner}, and the two structural
  reductions of Proposition~\ref{prop:publishedinputs}. \\
\hline
Generic four-phase reconstruction
& \path{certificates/generic_cover/product_cubic_reduction.py}\newline
  \path{certificates/generic_cover/positivity_lemma_reduction.py}\newline
  \path{certificates/generic_cover/exact_lower_block_specialization.py}
& Theorem~\ref{thm:regular-seed-domain}: the quadratic--cubic identities,
  complement positivity, and lower-block sheet selection. \\
\hline
Product-exceptional sectors
& \path{certificates/product_exceptional/tao_product_exceptional_check.py}\newline
  \path{certificates/product_exceptional/karlsson_product_regular_coverage/karlsson_product_exceptional_theorem_check.py}
& Propositions~\ref{prop:supp-karlsson-product-exceptional-singleton}
  and~\ref{prop:supp-tao-product-exceptional}: exact all-frame enumerations
  and the Karlsson-sector refinement. \\
\hline
Global product reach
& \path{certificates/product_escape/global_escape_incidence_check.py}\newline
  \path{certificates/product_escape/dependent_block_threshold_check.py}
& Theorem~\ref{thm:global-product-regular-escape}: the exact incidence,
  endpoint, orbit, and threshold calculations used in the reach proof. \\
\hline
Ramification representative
& \path{certificates/ramification/ramification_seed_certificate.py}
& Figure~\ref{fig:ramifiedHadamard} and
  Table~\ref{tab:supp-ramification-bounds}: exact algebraic seed data and
  rigorous Arb enclosures. \\
\hline
\end{tabular}
\end{table}

\setcounter{tocdepth}{3}
\supplementaltableofcontents
\clearpage

\section{{Proof of Theorem~\ref{thm:classification}}}
\label{sec:supp-classification-proof}

{This section contains the arguments behind the classification proof in Sec.~III of the main text. The manuscript proof uses exactly the two published structural results stated in Proposition~\ref{prop:publishedinputs}: Karlsson's classification of the $H_2$-reducible sector and Sz\"oll\H{o}si's cubic-root row-and-column criterion~\cite{Karlsson2011H2,Karlsson2011ThreeParameter,Szollosi2012FourParameter}.}

\subproofheading{{Theorem~\ref{thm:classification} (Complete finite-corner classification; restated)}}
{\emph{Every order-six complex Hadamard matrix is equivalent to a dephased matrix having a finite-corner witness.}}

{This is the principal theorem proved in this section.}

\subproofheading{{Proposition~\ref{prop:publishedinputs} (Published order-six structural inputs; restated)}}
{\emph{For an order-six complex Hadamard matrix $H$:}}
\begin{enumerate}
\item {$H$ is $H_2$-reducible if and only if it is equivalent to a member of Karlsson's complete three-real-parameter family \cite{Karlsson2011H2,Karlsson2011ThreeParameter}. 
Karlsson's parametrization also covers the parameter values at which both
M\"obius formulas degenerate. The resulting Hadamard matrices are, up to equivalence, members of the two-parameter Fourier family $F_6^{(2)}$ or its transpose \cite[Secs.~4--5 and Theorem~11]{Karlsson2011ThreeParameter}, see also \cite[p.~623, Theorem~2.11]{Szollosi2012FourParameter}.}
\item {If $H$ is equivalent to a dephased complex Hadamard matrix having both a noninitial row and a noninitial column composed entirely of cubic roots of unity, then either $[H] \in \mathcal T_6$ or $[H]\in\KK_6^{(3)}$~\cite[p.~624, Lemma~2.14]{Szollosi2012FourParameter}.}
\end{enumerate}
{These are the two published structural results used in the manuscript proof~\cite{Karlsson2011H2,Karlsson2011ThreeParameter,Szollosi2012FourParameter}. The Lean formalization we have provided~\cite{CardenesWuttig2026HadamardRepository} does not assume them. Instead, it proves the required reductions internally: the simultaneous cubic-Fourier branch leads to the Tao or $H_2$-reducible sectors, and every $H_2$-reducible matrix is reduced to regular Karlsson coordinates or an affine-Fourier seam, with a finite-corner witness in every resulting case.}

\subsection{{Proof of Lemma~\ref{lem:singularcorner}}}
\label{sec:supp-singular-corner}

\subproofheading{{Lemma~\ref{lem:singularcorner} (Singular-corner reduction; restated)}}
{\emph{If a $3\times3$ submatrix of an order-six complex Hadamard matrix is singular, then the matrix contains a $2\times2$ Hadamard submatrix.}}

\begin{proof}
We prove first that singularity forces two rows or two columns of the chosen corner to coincide, and then use full-matrix orthogonality to obtain a $2\times2$ Hadamard submatrix. Move the chosen rows and columns to the first three positions and dephase. The $3\times3$ submatrix is then
\begin{equation}
E(a,b,c,d)=\begin{pmatrix}1&1&1\\1&a&b\\1&c&d\end{pmatrix},
\label{eq:supp-singular-corner-form}
\end{equation}
where $a,b,c,d\in\C$ are the four unfixed phase entries of the dephased submatrix and $|a|=|b|=|c|=|d|=1$.  Its determinant equation is
\begin{equation}
ad-bc-d+b+c-a=0.
\label{eq:singular-corner-determinant}
\end{equation}

We first prove that two rows or two columns of $E$ coincide.  If $a=1$, Eq.~\eqref{eq:singular-corner-determinant} becomes $(b-1)(c-1)=0$; hence either the first two rows or the first two columns coincide. If $a\ne1$ and $b=1$, the same equation becomes $(a-1)(d-1)=0$, so $d=1$ and the first and third columns coincide. The case $a\ne1$ and $c=1$ is the transpose: again $d=1$, and the first and third rows coincide. If $a,b,c\ne1$ but $d=1$, the determinant equation gives $(b-1)(c-1)=0$, a contradiction.

It remains to treat $a,b,c,d\ne1$. Solving Eq.~\eqref{eq:singular-corner-determinant} for $d$ gives
\begin{equation}
(a-1)d=a+bc-b-c.
\label{eq:singular-corner-solved}
\end{equation}
All four parameters have modulus one. Taking absolute values in Eq.~\eqref{eq:singular-corner-solved}, expanding, and clearing the nonzero unit factors gives
\begin{equation}
(a-b)(a-c)(b-1)(c-1)=0.
\label{eq:singular-corner-factor}
\end{equation}

The last two factors are nonzero. Thus $a=b$ or $a=c$. In the first case Eq.~\eqref{eq:singular-corner-determinant} reduces to $(c-d)(1-a)=0$, so $c=d$ and the last two columns coincide. In the second case it reduces to $(b-d)(1-a)=0$, so $b=d$ and the last two rows coincide. This completes the local case split.

Suppose, for example, that rows $r$ and $s$ coincide on the selected three columns. Each of their three relative phases there is $1$. Orthogonality of the two full rows says that the complementary relative phases $z_1,z_2,z_3\in\T$ satisfy $z_1+z_2+z_3=-3$. The triangle inequality gives
\begin{equation}
3=|z_1+z_2+z_3|\leq |z_1|+|z_2|+|z_3|=3.
\end{equation}
Equality forces $z_1,z_2,z_3$ to have a common argument, and their sum is negative real; hence $z_1=z_2=z_3=-1$. Thus the relative phase is simultaneously $1$ on each selected column and $-1$ on each complementary column. Selecting any one column of each kind produces a $2\times2$ submatrix with orthogonal rows and unimodular entries, hence a $2\times2$ complex Hadamard matrix. The coincident-column case is identical after transposition.
\end{proof}

\subsubsection{{Proof of Lemma~\ref{lem:rowcolumn}}}
\label{sec:invariant}

{The cubic invariant is the sign marker used in the corner-routing argument. Complementary blocks reverse its sign, while the identity proved below allows the same argument to be applied after exchanging rows and columns. This is the step that lets Proposition~\ref{prop:main-corner-routing} make both side fibers finite at one corner.}

For any $X\in\T^{3\times3}$, define its row and column cubic invariants by
\begin{equation}
\begin{aligned}
\tau_{\rm r}(X)&=(XX^{\dagger})_{12}(XX^{\dagger})_{23}
(XX^{\dagger})_{31},\\
\tau_{\rm c}(X)&=(X^{\dagger}X)_{12}(X^{\dagger}X)_{23}
(X^{\dagger}X)_{31}.
\end{aligned}
\label{eq:cubicinvariants}
\end{equation}
Row or column phasing leaves the corresponding cyclic product unchanged. An odd permutation complex conjugates it, and an even permutation preserves it. Therefore its real part is invariant under every monomial row or column operation.

\begin{lemma}[Row--column invariant identity]
\label{lem:rowcolumn}
For every $X\in\T^{3\times3}$,
\begin{equation}
\RePart\tau_{\rm r}(X)=\RePart\tau_{\rm c}(X).
\label{eq:rowcolidentity}
\end{equation}
\end{lemma}

\begin{proof}
Set $A=XX^{\dagger}$ and $C=X^{\dagger}X$.  Both are Hermitian, both have diagonal $(3,3,3)$, and they have the same eigenvalues because $X$ is square. For a Hermitian matrix $M$ with this diagonal, define
\begin{equation}
\sigma(M)=\sum_{1\leq i<j\leq3}|M_{ij}|^2.
\label{eq:sigmam}
\end{equation}
Direct multiplication gives
\begin{equation}
\tr(M^2)=27+2\sigma(M).
\label{eq:tracem2}
\end{equation}
Equality of the spectra of $A$ and $C$ implies equality of their traces of squares, and hence
\begin{equation}
\sigma(A)=\sigma(C).
\label{eq:sigmaequal}
\end{equation}

Expanding a $3\times3$ determinant gives
\begin{equation}
\det M=27-3\sigma(M)
+2\RePart\!\left(M_{12}M_{23}M_{31}\right).
\label{eq:detcubic}
\end{equation}
The equal spectra also give $\det A=\det C$.  Subtracting the two instances of Eq.~\eqref{eq:detcubic} and using Eq.~\eqref{eq:sigmaequal} proves Eq.~\eqref{eq:rowcolidentity}.
\end{proof}

We shall use a normalized form for a general $3\times3$ unimodular matrix:
\begin{equation}
X=\begin{pmatrix}
1&1&1\\
x_1&x_2&x_3\\
y_1&y_2&y_3
\end{pmatrix},
\qquad |x_j|=|y_j|=1.
\label{eq:normalizedx}
\end{equation}
Define the complex numbers
\begin{equation}
S=\sum_{j=1}^{3}x_j,\qquad T=\sum_{j=1}^{3}y_j,\qquad
R=\sum_{j=1}^{3}\overline{x_j}y_j\in\C.
\label{eq:STR}
\end{equation}
Then
\begin{equation}
XX^{\dagger}=
\begin{pmatrix}
3&\overline S&\overline T\\
S&3&\overline R\\
T&R&3
\end{pmatrix},
\qquad
\tau_{\rm r}(X)=\overline S\,\overline R\,T.
\label{eq:gramSTR}
\end{equation}

{We use the notation $\mathcal R_X$ for the normalized physical row-Gram fiber already defined in Eq.~\eqref{eq:main-normalized-fiber} of the main text.}

\subsection{{Proof of Proposition~\ref{prop:main-infinite-fiber-trichotomy}}}
\label{sec:local}

{We now analyze when $\mathcal R_X$ can be infinite. The elimination is kept division-free so that common roots and vanishing denominators remain in the proof.}

\subproofheading{{Proposition~\ref{prop:main-infinite-fiber-trichotomy} (Infinite-fiber trichotomy; restated)}}
{\emph{Let $X\in\T^{3\times3}$ be invertible. If its normalized physical row-Gram fiber is infinite, then at least one of the following holds:}}
\begin{equation}
\begin{aligned}
&XX^{\dagger}=3I_3;\qquad\text{or}\\
&\RePart\tau_{\rm r}(X)<0;\qquad\text{or}\\
&X\text{ contains a }2\times2\text{ complex Hadamard submatrix}.
\end{aligned}
\label{eq:trichotomy}
\end{equation}
{\emph{These alternatives may overlap.}}

\begin{proof}
{Every $Y\in\mathcal R_X$ satisfies $YY^{\dagger}=XX^{\dagger}$. Since $X$ is invertible, this common Gram matrix is positive definite. Therefore every $Y$ in the fiber is invertible.}

Phasing or interchanging the last two rows maps the whole normalized fiber bijectively to another normalized fixed-Gram fiber.  These operations preserve invertibility, $XX^{\dagger}=3I_3$, the real part of the cubic invariant, and the existence of a $2\times2$ Hadamard submatrix.  We may therefore perform them without changing the conclusion.

If only finitely many second rows and finitely many third rows occurred in the fiber, then the fiber itself would be finite.  After interchanging the last two rows if necessary, the set of second-row triples is infinite.  If each of the three coordinate projections of this set were finite, their Cartesian product would be finite.  Hence some fixed coordinate has infinitely many distinct values.  Denote that coordinate and the entry below it by $x$ and $y$.

\subproofheading{Two residual pairs and polynomial elimination}

We reduce infinitude of the fiber to a set of one-variable polynomial identities. Isolate one coordinate from each noninitial row and denote the remaining entries by two residual pairs. After one common permutation of the three columns, write the two noninitial rows and the residual quantities consistently as
\begin{equation}
\begin{gathered}
(x_1,x_2,x_3)=(x,u_1,u_2),\qquad
(y_1,y_2,y_3)=(y,v_1,v_2),\\
\Sigma=x-S,\qquad \Delta=y-T,\qquad
\Psi=\frac{y}{x}-R.
\end{gathered}
\label{eq:SigmaDeltaPsi}
\end{equation}
Equations~\eqref{eq:STR} and \eqref{eq:SigmaDeltaPsi} give
\begin{equation}
u_1+u_2=-\Sigma,
\qquad v_1+v_2=-\Delta,
\qquad
\overline{u_1}v_1+\overline{u_2}v_2=-\Psi.
\label{eq:residualsums}
\end{equation}

Set $\rho=\overline{u_1}u_2$ and $\sigma=\overline{v_1}v_2$.  Because all four entries are unimodular,
\begin{equation}
\begin{aligned}
|\Sigma|^2&=2+2\RePart\rho,\\
|\Delta|^2&=2+2\RePart\sigma,\\
|\Psi|^2&=2+2\RePart(\rho\overline\sigma).
\end{aligned}
\label{eq:residualmoduli}
\end{equation}
The phases of $u_1$ and $v_1$ cancel in the product, and direct expansion gives
\begin{equation}
\Sigma\overline\Delta\Psi
=-(1+\rho)(1+\overline\sigma)(1+\overline\rho\sigma)
=-2-2\RePart(\rho+\sigma+\rho\overline\sigma).
\label{eq:residualproduct}
\end{equation}
Combining Eqs.~\eqref{eq:residualmoduli} and \eqref{eq:residualproduct}, we obtain the exact identity
\begin{equation}
\mathcal H=\overline{\mathcal H},
\qquad
\mathcal H=4-|\Sigma|^2-|\Delta|^2-|\Psi|^2,
\label{eq:residualhaagerup}
\end{equation}
where
\begin{equation}
\mathcal H=(x-S)(y^{-1}-\overline T)(y/x-R).
\label{eq:Hcal}
\end{equation}
The identity in Eq.~\eqref{eq:residualhaagerup} is commonly called Haagerup's trick \MainCiteHaagerupSzollosi. Here, we derive it directly from the two residual coordinate pairs.

On the unit circle, complex conjugation replaces $x$ by $x^{-1}$ and $y$ by $y^{-1}$.  Multiplying the two equations in Eq.~\eqref{eq:residualhaagerup} by the nonzero monomial $xy$ gives two polynomial equations
\begin{equation}
\Phi_{\rm H}(x,y)=(x-S)(1-\overline T y)(y-Rx)
-(1-\overline Sx)(y-T)(x-\overline R y)=0,
\label{eq:Phi}
\end{equation}
and
\begin{equation}
\begin{aligned}
\Gamma(x,y)={}&(x-S)(1-\overline T y)(y-Rx)-4xy\\
&+y(x-S)(1-\overline Sx)
+x(y-T)(1-\overline T y)\\
&+(y-Rx)(x-\overline R y)=0.
\end{aligned}
\label{eq:Gamma}
\end{equation}

Both polynomials are quadratic in $y$.  Write
\begin{equation}
\Phi_{\rm H}=F_1+F_2y+F_3y^2,
\qquad
\Gamma=G_1+G_2y+G_3y^2.
\label{eq:FGcoeffdefinition}
\end{equation}
Their coefficients are
\begin{equation}
\begin{aligned}
F_1={}&-Rx(x-S)+Tx(1-\overline Sx),\\
F_2={}&(x-S)(1+\overline T Rx)
       -(1-\overline Sx)(x+\overline RT),\\
F_3={}&-\overline T(x-S)+\overline R(1-\overline Sx),
\end{aligned}
\label{eq:Fcoeffs}
\end{equation}
and
\begin{equation}
\begin{aligned}
G_1={}&-Rx(x-S)-xT-Rx^2,\\
G_2={}&(x-S)(1+\overline T Rx)-4x
 +(x-S)(1-\overline Sx)\\
&+x(1+|T|^2)+x(1+|R|^2),\\
G_3={}&-\overline T(x-S)-x\overline T-\overline R.
\end{aligned}
\label{eq:Gcoeffs}
\end{equation}

Eliminate $y^2$ by forming $F_3\Gamma-G_3\Phi_{\rm H}$.  Every coordinate pair in every fiber member therefore satisfies
\begin{equation}
\mathcal A(x)+\mathcal B(x)y=0,
\label{eq:linearEliminant}
\end{equation}
where
\begin{equation}
\mathcal A=F_3G_1-F_1G_3,
\qquad
\mathcal B=F_3G_2-F_2G_3.
\label{eq:ABeliminants}
\end{equation}
Substitution of Eqs.~\eqref{eq:Fcoeffs} and \eqref{eq:Gcoeffs}, followed by elementary collection of powers of $x$, yields
\begin{equation}
\mathcal A(x)=x\bigl(|R|^2-|T|^2\bigr)\kappa_S(x),
\qquad \deg\mathcal B\leq3.
\label{eq:Afactor}
\end{equation}
Here, $\deg\mathcal B\leq3$ means that $\mathcal B(x)$ has degree at most three. The factor $\kappa_S$ is
\begin{equation}
\kappa_S(x)=2\overline Sx^2-(|S|^2+3)x+2S.
\label{eq:kappaS}
\end{equation}

If $\mathcal B$ were the zero polynomial, the selected infinite set of $x$-values in Eq.~\eqref{eq:linearEliminant} would force $\mathcal A\equiv0$.  Otherwise discard the finitely many roots of $\mathcal B$ and write
\begin{equation}
y=-\frac{\mathcal A(x)}{\mathcal B(x)}.
\label{eq:yquotient}
\end{equation}
Because $|y|=1$, the Laurent polynomial $|\mathcal A(x)|^2-|\mathcal B(x)|^2$ vanishes at infinitely many unit-circle points and hence vanishes identically.  In this second branch $\mathcal A$ cannot be the zero polynomial, since Eq.~\eqref{eq:linearEliminant} would then give $y=0$ away from finitely many roots.  Thus the two exhaustive branches are the dependent branch
\begin{equation}
\mathcal A\equiv\mathcal B\equiv0,
\label{eq:dependent-branch}
\end{equation}
and the nondependent branch
\begin{equation}
\mathcal A\not\equiv0\ \text{and}\ \mathcal B\not\equiv0.
\label{eq:nondependent-branch}
\end{equation}

\subproofheading{The dependent branch}
There are two cases according to whether both eliminant coefficients vanish identically. We first treat the dependent branch, in which both of its coefficient polynomials vanish identically. This is the branch in which ordinary division would lose information.

Assume first that $\mathcal A\equiv\mathcal B\equiv0$.  The coefficient of $x$ in $\kappa_S$ is $-(|S|^2+3)$, so $\kappa_S$ is never the zero polynomial.  Equation~\eqref{eq:Afactor} therefore implies
\begin{equation}
|R|=|T|.
\label{eq:RTequalmod}
\end{equation}

Suppose $T\neq0$.  Phase the second and third rows so that
\begin{equation}
S=s\geq0,
\qquad T=t>0,
\qquad R=t\zeta,
\qquad |\zeta|=1.
\label{eq:dependentnormalization}
\end{equation}
The transformed second-row projection is still infinite.  Its transformed $\mathcal A$ vanishes by Eq.~\eqref{eq:RTequalmod}. Hence, Eq.~\eqref{eq:linearEliminant} at infinitely many $x$-values forces its transformed $\mathcal B$ to vanish as well.

Write $\mathcal B(x)=\sum_{k=0}^{3}\beta_kx^k$.  Substitution of Eq.~\eqref{eq:dependentnormalization} in Eqs.~\eqref{eq:Fcoeffs}--\eqref{eq:ABeliminants} gives
\begin{equation}
\begin{aligned}
\beta_0&=-t\Pi/\zeta^2,&
\beta_1&=st\Pi/\zeta^2,\\
\beta_2&=-st\Theta/\zeta,&
\beta_3&=t\Theta/\zeta,
\end{aligned}
\label{eq:dependentbetas}
\end{equation}
where
\begin{equation}
\begin{aligned}
\Pi&=s^2\zeta^2-st^2\zeta+3s\zeta+t^2,\\
\Theta&=s^2-st^2\zeta+3s\zeta+t^2\zeta^2.
\end{aligned}
\label{eq:PiTheta}
\end{equation}
Because $t>0$ and $\mathcal B\equiv0$, both $\Pi$ and $\Theta$ vanish. Subtracting them gives
\begin{equation}
(s^2-t^2)(\zeta^2-1)=0.
\label{eq:dependentfactor}
\end{equation}

There are three cases.  If $\zeta=1$, the equation $\Pi=0$ becomes
\begin{equation}
t^2(s-1)=s(s+3).
\label{eq:zetaone}
\end{equation}
It implies $s>1$.  Since a sum of three unit numbers has modulus at most three, $s,t\leq3$.  Rearranging Eq.~\eqref{eq:zetaone} gives
\begin{equation}
t^2-9=\frac{(s-3)^2}{s-1}.
\label{eq:zetaonebound}
\end{equation}
The left side is nonpositive and the right side is nonnegative, so both are zero.  Hence $s=t=3$.  Equality in the triangle inequality makes both noninitial rows constant after phasing, contradicting invertibility.

If $\zeta=-1$, Eq.~\eqref{eq:PiTheta} excludes $s=0$, and Eq.~\eqref{eq:gramSTR} gives
\begin{equation}
\RePart\tau_{\rm r}(X)=-st^2<0.
\label{eq:zetaminustau}
\end{equation}

If $s=t$ and $\zeta^2\neq1$, the equation $\Pi=0$ reduces to
\begin{equation}
s\zeta^2+(3-s^2)\zeta+s=0,
\qquad
\RePart\zeta=\frac{s^2-3}{2s}.
\label{eq:stzeta}
\end{equation}
The Gram determinant is
\begin{equation}
\det(XX^{\dagger})=(s^2-3)(s^2-9).
\label{eq:dependentdet}
\end{equation}
Invertibility makes this determinant positive.  Since $s\leq3$, we obtain $s<\sqrt3$.  The bound $|\RePart\zeta|\leq1$ in Eq.~\eqref{eq:stzeta} gives $s\geq1$.  Therefore
\begin{equation}
\RePart\tau_{\rm r}(X)
=s^3\RePart\zeta
=\frac{s^2(s^2-3)}{2}<0.
\label{eq:dependentnegative}
\end{equation}

It remains to consider $T=R=0$.  A triple of unit numbers with sum zero is, up to a common phase and permutation,
\begin{equation}
(y_1,y_2,y_3)=\lambda(1,\omega,\omega^2).
\label{eq:unitzerosum}
\end{equation}
The condition $R=0$ places the second row in the orthogonal complement of this Fourier row, so
\begin{equation}
(x_1,x_2,x_3)
=\alpha(1,1,1)+\beta(1,\omega^2,\omega).
\label{eq:xFourierdecomp}
\end{equation}
The three equalities $|x_j|=1$ imply that
\begin{equation}
\RePart(\alpha\overline\beta)
=\RePart(\omega\alpha\overline\beta)
=\RePart(\omega^2\alpha\overline\beta).
\label{eq:threeRe}
\end{equation}
Their sum is zero, so they all vanish.  The first two already force $\alpha\overline\beta=0$.  If $\beta=0$, the first two rows are proportional and $X$ is singular.  If $\alpha=0$, the three rows are Fourier rows and
\begin{equation}
XX^{\dagger}=3I_3.
\label{eq:dependentFourier}
\end{equation}

\subproofheading{The nondependent branch}
Assume now that $\mathcal A\not\equiv0$. Phase row two so that $S=s$ is real with $0\leq s\leq3$, and define
\begin{equation}
\delta=|R|^2-|T|^2\neq0,
\qquad
\kappa_s(x)=2sx^2-(s^2+3)x+2s.
\label{eq:deltakappas}
\end{equation}
Write
\begin{equation}
\mathcal B(x)=\beta_0+\beta_1x+\beta_2x^2+\beta_3x^3.
\label{eq:Bbeta}
\end{equation}
The modulus identity following Eq.~\eqref{eq:yquotient} and Eq.~\eqref{eq:Afactor} gives, for $|x|=1$,
\begin{equation}
|\mathcal B(x)|^2=|\delta|^2|\kappa_s(x)|^2.
\label{eq:modulusidentity}
\end{equation}
The right side has Fourier degree at most two.  The coefficient of the third Fourier mode on the left is $\beta_3\overline{\beta_0}$.  Hence
\begin{equation}
\beta_3\overline{\beta_0}=0.
\label{eq:endcoefficient}
\end{equation}

If $\beta_0=0$, write $\mathcal B=xQ$; if $\beta_3=0$, write $\mathcal B=Q$.  Choose either representation if both conditions hold.  In both cases $Q\neq0$, $\deg Q\leq2$, and $|Q|=|\mathcal B|$ on the unit circle.  Define the reversed-conjugate polynomial
\begin{equation}
Q^{\#}(x)=x^2\overline{Q(1/\overline x)}.
\label{eq:Qsharp}
\end{equation}
On $|x|=1$, $Q^{\#}(x)=x^2\overline{Q(x)}$.  The real palindromic polynomial $\kappa_s$ satisfies $\kappa_s^{\#}=\kappa_s$. Multiplying Eq.~\eqref{eq:modulusidentity} by $x^2$ and using polynomial identity gives
\begin{equation}
Q(x)Q^{\#}(x)=|\delta|^2\kappa_s(x)^2.
\label{eq:QQsharp}
\end{equation}

We now split the nondependent branch according to the real number $s=|S|\in[0,3]$ into: $(a)$ $s=0$, $(b)$ $0<s<1$, and $(c)$ $1\leq s\leq3$.

\paragraph{\texorpdfstring{Case $s=0$}{Case s=0}.}

Now $\kappa_0(x)=-3x$, and Eq.~\eqref{eq:modulusidentity} says that $\mathcal B$ has nonzero constant modulus on the unit circle. Such a polynomial is a monomial. If its smallest and largest nonzero exponents were different, their product would give a nonzero highest Fourier mode in the squared modulus.  Thus
\begin{equation}
\mathcal B(x)=b x^m,
\qquad b\neq0,
\qquad m\in\{0,1,2,3\}.
\label{eq:Bmonomial}
\end{equation}
Equation~\eqref{eq:linearEliminant} has no exceptional unit-circle root in this case and gives, for every coordinate,
\begin{equation}
y_j=\lambda x_j^k,
\qquad |\lambda|=1,
\qquad k=2-m\in\{-1,0,1,2\}.
\label{eq:powerrelation}
\end{equation}
The cases $k=0$ and $k=1$ make the third row proportional to the first or second row, contradicting invertibility.  Since $S=0$, the second row is a phased and permuted copy of $(1,\omega,\omega^2)$.  For $k=-1$ or $k=2$, the third row is the other nonconstant Fourier row.  Hence Eq.~\eqref{eq:dependentFourier} holds.

\paragraph{\texorpdfstring{Case $0<s<1$}{Case 0<s<1}.}

Write $r_-<r_+$ for the two roots of the quadratic $\kappa_s$. The discriminant and product of the roots of $\kappa_s$ are
\begin{equation}
\operatorname{disc}(\kappa_s)=(s^2-1)(s^2-9)>0,
\qquad r_-r_+=1.
\label{eq:kappadiscriminant}
\end{equation}
Thus $\kappa_s$ has two distinct positive reciprocal roots $r_-<r_+$.  The constant and leading coefficients on the right side of Eq.~\eqref{eq:QQsharp} are nonzero, so $Q$ has degree two and nonzero constant term.  Unique factorization permits exactly the following allocations:
\begin{equation}
Q\propto\kappa_s,
\qquad
Q\propto(x-r_-)^2,
\qquad
Q\propto(x-r_+)^2.
\label{eq:rootallocations}
\end{equation}

An allocation records how the two copies of each root on the right of Eq.~\eqref{eq:QQsharp} are distributed between $Q$ and $Q^\#$, up to a nonzero scalar factor.

If $Q\propto\kappa_s$, Eq.~\eqref{eq:linearEliminant} makes the third row proportional to the first when $\mathcal B=xQ$, or to the second when $\mathcal B=Q$.  The roots $r_\pm$ are off the unit circle, so no coordinate is exceptional.  Invertibility excludes this allocation.

For a squared-root allocation, let $r$ be the doubled root and define
\begin{equation}
m_r(z)=r\frac{z-r^{-1}}{z-r}.
\label{eq:mobius}
\end{equation}
For $|z|=1$ and real $r$, $|m_r(z)|=1$.  Equation \eqref{eq:linearEliminant} then gives one of the two relations
\begin{equation}
y_j=\lambda m_r(x_j),
\qquad\text{or}\qquad
y_j=\lambda x_jm_r(x_j),
\label{eq:twoMobiusAllocations}
\end{equation}
where $\lambda$ is independent of $j$ and $|\lambda|=1$.

The critical-root equation and root locations are
\begin{equation}
2sr^2-(s^2+3)r+2s=0,
\label{eq:criticalroot}
\end{equation}
\begin{equation}
r_-<s<1<1/s<r_+.
\label{eq:rootlocations}
\end{equation}
Indeed, $\kappa_s(0)=2s>0$ and $\kappa_s(s)=s(s^2-1)<0$, so $r_-\in(0,s)$, and reciprocality gives $r_+>1/s$. In particular, all denominators below are nonzero.

Let $u=x_1x_2x_3$.  Since the $x_j$ are unimodular and their sum is the real number $s$,
\begin{equation}
\sum_{i<j}x_ix_j=u\sum_j\overline{x_j}=us.
\label{eq:secondsymmetric}
\end{equation}
Thus the $x_j$ are the roots of
\begin{equation}
p_u(z)=z^3-sz^2+usz-u.
\label{eq:pu}
\end{equation}
Because $r$ is off the unit circle, $p_u(r)\neq0$.  The logarithmic derivative and Eq.~\eqref{eq:criticalroot} give
\begin{equation}
\sum_{j=1}^3\frac{1}{x_j-r}
=-\frac{p_u'(r)}{p_u(r)}
=\frac{s}{1-sr}.
\label{eq:logderivative}
\end{equation}

Using Eq.~\eqref{eq:logderivative} in the elementary partial-fraction expansions of Eq.~\eqref{eq:mobius}, we obtain
\begin{equation}
\sum_jm_r(x_j)=s,
\label{eq:mobiusSum0}
\end{equation}
\begin{equation}
\sum_j\frac{m_r(x_j)}{x_j}
=\frac{s(r-s)}{1-sr},
\label{eq:mobiusSumMinus}
\end{equation}
and
\begin{equation}
\sum_jx_jm_r(x_j)
=\frac{s(r^{-1}-s)}{1-s/r}.
\label{eq:mobiusSumPlus}
\end{equation}

For either $r=r_-$ or $r=r_+$, the numerator and denominator in each of Eqs.~\eqref{eq:mobiusSumMinus} and \eqref{eq:mobiusSumPlus} have opposite signs by Eq.~\eqref{eq:rootlocations}.  Both displayed quotients are therefore strictly negative.

In the first relation of Eq.~\eqref{eq:twoMobiusAllocations}, $T=\lambda s$ and $R=-\lambda q$ for some $q>0$.  In the second, $R=\lambda s$ and $T=-\lambda q'$ for some $q'>0$.  In either case, Eq.~\eqref{eq:gramSTR} yields
\begin{equation}
\RePart\tau_{\rm r}(X)<0.
\label{eq:mobiusnegative}
\end{equation}

\paragraph{\texorpdfstring{Case $1\leq s\leq3$ and the common-root repair}{Case 1<=s<=3 and the common-root repair}.}

Equation~\eqref{eq:kappadiscriminant} now shows that the roots of $\kappa_s$ lie on the unit circle; at $s=1$ and $s=3$ they merge into the double root $1$. As before, Eq.~\eqref{eq:QQsharp} forces $Q$ to have degree two and nonzero constant term.  A unit-circle root is fixed by the $\#$ operation, with its multiplicity unchanged.  Comparing root multiplicities on both sides of Eq.~\eqref{eq:QQsharp}, including the double-root endpoints, gives
\begin{equation}
Q=\gamma\kappa_s,
\qquad \gamma\neq0.
\label{eq:QgammaKappa}
\end{equation}

Let
\begin{equation}
Z_s=\{z\in\T:\kappa_s(z)=0\},
\qquad \lambda=-\delta/\gamma.
\label{eq:Zs}
\end{equation}
Thus $Z_s$ is the exceptional-root set: the one or two unit-circle zeros of the critical quadratic $\kappa_s$, at which the eliminant cannot be cancelled. If $\mathcal B=xQ$, Eq.~\eqref{eq:linearEliminant} becomes
\begin{equation}
x_j\kappa_s(x_j)(\delta+\gamma y_j)=0;
\label{eq:commonrootcase1}
\end{equation}
if $\mathcal B=Q$, it becomes
\begin{equation}
\kappa_s(x_j)(\delta x_j+\gamma y_j)=0.
\label{eq:commonrootcase2}
\end{equation}
Outside $Z_s$, these equations say respectively $y_j=\lambda$ and $y_j=\lambda x_j$.  A fiber member with no $x_j\in Z_s$ would therefore have two proportional rows.  Every fiber member is invertible, so every second-row triple in the fiber contains at least one member of $Z_s$.

The set $Z_s$ has at most two elements.  Cover the infinite set of second-row triples by the finitely many cells obtained by fixing a position $j$ and a root $z\in Z_s$ with $x_j=z$.  At least one cell contains infinitely many distinct triples.  In that cell the other two entries, say $a,b$, obey
\begin{equation}
a+b=s-z=:c_0.
\label{eq:abfixedsum}
\end{equation}
If $c_0\neq0$, unimodularity gives
\begin{equation}
ab\,\overline{c_0}=ab(\overline a+\overline b)=a+b=c_0,
\qquad
ab=\frac{c_0}{\overline{c_0}}.
\label{eq:abproduct}
\end{equation}
Thus $a,b$ are the two roots, in either order, of one fixed quadratic,
\begin{equation}
\xi^2-c_0\xi+\frac{c_0}{\overline{c_0}}=0.
\label{eq:abquadratic}
\end{equation}
There would be only finitely many triples in that cell, a contradiction. Consequently $c_0=0$, so $s=z$.  Because $s$ is real and at least one while $|z|=1$,
\begin{equation}
s=z=1,
\qquad Z_1=\{1\}.
\label{eq:szone}
\end{equation}

Every fiber member, including the original $X$, has an $x$-coordinate equal to $1$.  After a column permutation, the condition $x_1+x_2+x_3=S=1$ gives
\begin{equation}
(x_1,x_2,x_3)=(1,u,-u),
\qquad |u|=1.
\label{eq:oneuminusu}
\end{equation}
On the last two columns, the first two rows contain
\begin{equation}
\begin{pmatrix}1&1\\u&-u\end{pmatrix},
\label{eq:H2block}
\end{equation}
which is a $2\times2$ complex Hadamard matrix.  

{All branches have now been exhausted, proving Proposition~\ref{prop:main-infinite-fiber-trichotomy}.}
\end{proof}

\subsection{{Proof of Proposition~\ref{prop:main-corner-routing}}}
\label{sec:blockswap}

{Proposition~\ref{prop:main-infinite-fiber-trichotomy} controls one side fiber at a time. We now combine it with the complementary block Gram identities. The aim is to move to a corner where the horizontal and vertical fibers are finite simultaneously.} Throughout this section assume
\begin{equation}
[H]\notin\KK_6^{(3)}.
\label{eq:outsideK}
\end{equation}
Lemma~\ref{lem:singularcorner} and Proposition~\ref{prop:publishedinputs}(1) then imply that every $3\times3$ submatrix of $H$ is invertible.  In particular, all four displayed blocks $E,B,C,D$ are invertible.  Moreover, a candidate block with the same Gram matrix as one of these blocks is also invertible, because that common Gram matrix is positive definite.

Relabeling a fixed corner and dephasing the full matrix act on its candidate sets by permutations and diagonal phase multiplications.  These operations give bijections between the old and new normalized candidate sets.  They therefore preserve whether a candidate set is finite.  We use such relabelings below, but we never assert that unrelated corners have candidate sets of the same size.

\subproofheading{{Proposition~\ref{prop:main-corner-routing} (Corner routing; restated)}}
{\emph{Let $H$ be an order-six Hadamard matrix, written in the block form of Eq.~\eqref{eq:blockform}, and suppose that $[H]\notin\KK_6^{(3)}$. Then either all four displayed blocks are order-three Hadamard matrices, or an allowed row and column permutation, followed by dephasing, produces a finite-corner witness.}}

\begin{proof}
If both invertible-candidate sets at $E$ are nonempty and finite, the assertion already holds.  Under the standing assumption, the actual adjacent blocks are invertible and belong to these sets, so nonemptiness is automatic. Moreover, the fixed positive-definite Gram matrices imply that the complete candidate sets equal their invertible subsets. Suppose first that the horizontal set $\mathcal B_E$ is infinite.  Since
\begin{equation}
BB^{\dagger}=6I_3-EE^{\dagger},
\label{eq:Bgramcomplement}
\end{equation}
the uniquely first-row-normalized copy of the actual block $B$ lies in an infinite normalized fixed-Gram fiber.  Proposition~\ref{prop:main-infinite-fiber-trichotomy} applies.

The $2\times2$ alternative in Proposition~\ref{prop:main-infinite-fiber-trichotomy} is impossible: the same submatrix would occur inside $H$, contrary to Eq.~\eqref{eq:outsideK} and Proposition~\ref{prop:publishedinputs}(1).  If
\begin{equation}
BB^{\dagger}=3I_3,
\label{eq:Borthogonalbranch}
\end{equation}
then Eq.~\eqref{eq:Bgramcomplement} also gives $EE^{\dagger}=3I_3$.  Since $E$ and $B$ are square, both are order-three complex Hadamard matrices, and the block equations then show that $C$ and $D$ are as well. This is the first alternative of the proposition. In the remaining branch,
\begin{equation}
\RePart\tau_{\rm r}(B)<0.
\label{eq:taurbnegative}
\end{equation}

Interchange the two column triples.  Before the subsequent dephasing, the new block display is
\begin{equation}
\widetilde H=\begin{pmatrix}B&E\\D&C\end{pmatrix}.
\label{eq:columntripleswap}
\end{equation}
For each off-diagonal position $i\neq j$, Eq.~\eqref{eq:Bgramcomplement} gives $(EE^{\dagger})_{ij}=-(BB^{\dagger})_{ij}$.  Multiplying the three cyclic off-diagonal entries therefore yields
\begin{equation}
\tau_{\rm r}(E)=-\tau_{\rm r}(B),
\qquad
\RePart\tau_{\rm r}(E)>0.
\label{eq:taurepositive}
\end{equation}

The horizontal candidate set at the new corner $B$ has $E$ as a member. If that set were infinite, Proposition~\ref{prop:main-infinite-fiber-trichotomy} applied to $E$ would give one of three contradictions: $EE^{\dagger}=3I_3$ would force Eq.~\eqref{eq:Borthogonalbranch}; a negative real invariant would contradict Eq.~\eqref{eq:taurepositive}; and a $2\times2$ Hadamard submatrix would contradict Eq.~\eqref{eq:outsideK}.  Hence this horizontal set is finite.

It remains to prove finiteness of the vertical candidate set at the same new corner $B$.  Lemma~\ref{lem:rowcolumn} and Eq.~\eqref{eq:taurbnegative} give
\begin{equation}
\RePart\tau_{\rm c}(B)<0.
\label{eq:taucbnegative}
\end{equation}
The lower block $D$ satisfies $B^{\dagger}B+D^{\dagger}D=6I_3$.  The same three-minus-sign calculation therefore gives
\begin{equation}
\RePart\tau_{\rm c}(D)>0.
\label{eq:taucdpositive}
\end{equation}

If the new vertical candidate set were infinite, apply Proposition~\ref{prop:main-infinite-fiber-trichotomy} to the first-row-normalized adjoint $D^{\dagger}$. Its orthogonal alternative would give $D^{\dagger}D=3I_3$ and then $B^{\dagger}B=3I_3$, contrary to the branch under consideration.  Its negative-invariant alternative contradicts Eq.~\eqref{eq:taucdpositive}, because $\tau_{\rm r}(D^{\dagger})=\tau_{\rm c}(D)$.  Its $2\times2$ alternative again puts such a submatrix in $H$.  Thus the vertical set is finite too.

The only remaining possibility is that the original horizontal set is finite and the original vertical set is infinite.  Apply the argument just given to $H^{\dagger}$, which exchanges horizontal and vertical roles.  After taking adjoints back, it either shows that all four original blocks are order-three Hadamard matrices or supplies a corner of $H$ with both sets finite. These are exactly the two conclusions of the proposition, so the proof is complete.
\end{proof}

\subsection{{Proof of Proposition~\ref{prop:main-fourier-block}}}
\label{sec:fourierblock}

It remains to treat the case in which one $3\times3$ block is itself Hadamard. We show that this places the full matrix in the Karlsson or Tao sector.
Every order-three complex Hadamard matrix is equivalent to
\begin{equation}
F=F_3=\begin{pmatrix}
1&1&1\\
1&\omega&\omega^2\\
1&\omega^2&\omega
\end{pmatrix}.
\label{eq:F3}
\end{equation}
Indeed, after dephasing, each noninitial row contains three unit numbers with sum zero.  Three unit numbers sum to zero only when, after a common phase and a permutation, they are $1,\omega,\omega^2$, and orthogonality of the two noninitial rows fixes their relative order. This proves the assertion directly.

\subproofheading{{Proposition~\ref{prop:main-fourier-block} (Fourier-block alternative; restated)}}
{\emph{If one block in a block partition of an order-six complex Hadamard matrix is an order-three complex Hadamard matrix, then}}
\begin{equation}
{H\sim S_6^{(0)}
\qquad\text{or}\qquad
[H]\in\KK_6^{(3)}.}
\label{eq:fourierblockalternative}
\end{equation}

\begin{proof}
Moving the Hadamard block to the upper left and using the block equations shows that all four blocks are order-three complex Hadamard matrices.  By block-preserving equivalence operations and the uniqueness of $F_3$ just proved, we may write
\begin{equation}
E=F,\qquad B=D_vF,\qquad C=FD_w,
\label{eq:fouriernormalform}
\end{equation}
where
\begin{equation}
D_v=\operatorname{diag}(1,p,q),\qquad
D_w=\operatorname{diag}(1,r,s),
\label{eq:DvDw}
\end{equation}
with $p,q,r,s\in\T$.

The three normalizations can be imposed simultaneously.  First use operations on the top three rows and left three columns to set $E=F$. Operations on the right three columns then normalize the first row and column ordering of $B$. The order-three uniqueness proved above leaves only one phase on each row, giving $B=D_vF$. Any required permutation of the top three rows can be compensated by a monomial operation on the left columns that restores $E=F$.  Finally, operations on the bottom three rows do not change $E$ or $B$ and reduce $C$ to $FD_w$.  Every permutation of three Fourier columns is an affine permutation of their labels and can therefore be absorbed into a monomial row operation.  Common scalar phases are absorbed into $D_v$ and $D_w$, whose first diagonal entries can consequently be fixed to one.

The off-diagonal column-orthogonality equation $E^{\dagger}B+C^{\dagger}D=0$ uniquely determines
\begin{equation}
D=-\frac{1}{3}FD_wF^{\dagger}D_vF.
\label{eq:Dfourierformula}
\end{equation}
Thus it remains only to determine when every entry on the right has modulus one.

Write
\begin{equation}
v=(1,p,q),\qquad w=(1,r,s).
\label{eq:vwhat}
\end{equation}
All subscripts in this calculation are interpreted modulo three because the coordinates form a three-cycle.  The following equation defines the unnormalized three-point Fourier transform:
\begin{equation}
\widehat v_k=\sum_{\ell=0}^{2}v_{\ell}\omega^{k\ell}.
\label{eq:vhat-transform}
\end{equation}
Define
\begin{equation}
A_k=\widehat v_k\overline{\widehat v_{k-1}},
\qquad
\alpha_m=w_m\overline{w_{m+1}}
=(\overline r,r\overline s,s)_m.
\label{eq:Aalpha}
\end{equation}
For a fixed column $k$ in Eq.~\eqref{eq:Dfourierformula}, direct matrix multiplication shows that the three entries before the factor $-1/3$ are the unnormalized Fourier transform of
\begin{equation}
z_m^{(k)}=w_m\widehat v_{k-m}.
\label{eq:zk}
\end{equation}
{For each fixed $k$, apply the same three-point Fourier transform to the triple $z^{(k)}$}
\begin{equation}
{\widehat z_j^{(k)}=\sum_{m=0}^{2}z_m^{(k)}\omega^{jm}.}
\label{eq:zhat-transform}
\end{equation}

For any triple $z$, the three values $|\widehat z_j|^2$ are equal exactly when both nonzero cyclic autocorrelations vanish.  In the present case, unnormalized Parseval identities give
\begin{equation}
\sum_{j=0}^{2}|\widehat z_j^{(k)}|^2
=3\sum_{m=0}^{2}|z_m^{(k)}|^2
=3\sum_{m=0}^{2}|\widehat v_m|^2=27.
\label{eq:parseval27}
\end{equation}
Consequently equal Fourier magnitudes mean that each magnitude is $3$, which is exactly what is required after multiplication by $1/3$ in Eq.~\eqref{eq:Dfourierformula}.

The shift-one autocorrelation is
\begin{equation}
\sum_{m=0}^{2}z_m^{(k)}\overline{z_{m+1}^{(k)}}
=\sum_{m=0}^{2}\alpha_m A_{k-m}.
\label{eq:autoconvolution}
\end{equation}
Indeed, substituting $z_m^{(k)}=w_m\widehat v_{k-m}$ into the left side produces $w_m\overline{w_{m+1}}=\alpha_m$ and $\widehat v_{k-m}\overline{\widehat v_{k-m-1}}=A_{k-m}$. The shift-two autocorrelation is the conjugate of a cyclic reindexing of the same equations.  Hence all entries of $D$ are unimodular exactly when the cyclic convolution, defined by $(\alpha*A)_k=\sum_{m=0}^{2}\alpha_mA_{k-m}$, satisfies
\begin{equation}
\alpha*A=0.
\label{eq:alphaconvolutionA}
\end{equation}

Taking the three-point Fourier transform converts convolution into coordinatewise multiplication:
\begin{equation}
\widehat\alpha_j\widehat A_j=0,\qquad j=0,1,2.
\label{eq:fourierproduct}
\end{equation}
Furthermore,
\begin{equation}
\widehat A_0=\sum_{k=0}^{2}A_k
=3\sum_{m=0}^{2}|v_m|^2\omega^m
=3(1+\omega+\omega^2)=0.
\label{eq:Ahat0}
\end{equation}
Only the two equations
\begin{equation}
\widehat\alpha_1\widehat A_1=0,
\qquad
\widehat\alpha_2\widehat A_2=0
\label{eq:twoproducts}
\end{equation}
remain.

If $\widehat A_1=\widehat A_2=0$, then Eq.~\eqref{eq:Ahat0} implies $A=0$. The three cyclic products in Eq.~\eqref{eq:Aalpha} then allow at most one $\widehat v_k$ to be nonzero.  Inverting the Fourier transform shows that $v$ is a Fourier row.  A permutation of the right column triple makes $B=F$, and Eq.~\eqref{eq:Dfourierformula} then gives $D=-C$.  The resulting $2\times2$ Hadamard submatrices place $[H]$ in $\KK_6^{(3)}$.

If $\widehat\alpha_1=\widehat\alpha_2=0$, then $\alpha$ is constant.  Since $\alpha_0\alpha_1\alpha_2=1$, that constant is a cubic root of unity.  It follows that $w$ is a Fourier row.  A bottom-row permutation gives $C=F$ and $D=-B$, again placing $[H]$ in $\KK_6^{(3)}$.

The only other allocations of the zero factors in Eq.~\eqref{eq:twoproducts} are the two mixed cases.  They are interchanged by complex conjugation and index reversal, so take
\begin{equation}
\widehat A_1=0,\qquad \widehat\alpha_2=0.
\label{eq:mixedallocation}
\end{equation}
Up to order, the two nonconstant Fourier coefficients of $A$ are
\begin{equation}
3\left(q+\frac{\omega}{p}+\frac{\omega^2p}{q}\right),
\qquad
3\left(p+\frac{\omega q}{p}+\frac{\omega^2}{q}\right).
\label{eq:AhatExplicit}
\end{equation}
To derive these coefficients, insert $A_k=\widehat v_k\overline{\widehat v_{k-1}}$ into $\widehat A_j=\sum_kA_k\omega^{jk}$, expand both Fourier sums, and use $\sum_k\omega^{k(a-b+j)}=3$ when $b=a+j\pmod3$, and zero otherwise. Each parenthesis is a sum of three unit numbers whose product is one. Such a sum vanishes only when the summands are a rotated copy of $\{1,\omega,\omega^2\}$. The product condition makes the rotation a cubic root.  Thus the vanishing coefficient in Eq.~\eqref{eq:mixedallocation} forces
\begin{equation}
p,q\in\{1,\omega,\omega^2\}.
\label{eq:pqcubic}
\end{equation}

The two nonconstant Fourier coefficients of $\alpha$ are
\begin{equation}
\overline r+\omega r\overline s+\omega^2s,
\qquad
\overline r+\omega^2r\overline s+\omega s.
\label{eq:alphahatExplicit}
\end{equation}
The same three-unit-number argument gives
\begin{equation}
r,s\in\{1,\omega,\omega^2\}.
\label{eq:rscubic}
\end{equation}
The other mixed allocation gives the identical conclusion.

Equations~\eqref{eq:pqcubic} and \eqref{eq:rscubic} show that the dephased $H$ has both a noninitial row and a noninitial column made entirely of cubic roots of unity.  {Proposition~\ref{prop:publishedinputs}(2) therefore gives the conclusion of Proposition~\ref{prop:main-fourier-block}.} The two pure allocations and two mixed allocations exhaust Eq.~\eqref{eq:twoproducts}, so the proposition is proved.
\end{proof}

\subsection{{Proof of Proposition~\ref{thm:finitecornerintro}}}
\label{sec:directcompletion}
\label{sec:completion}

{Proposition~\ref{prop:main-infinite-fiber-trichotomy} controls infinite side fibers, Proposition~\ref{prop:main-corner-routing} selects a finite corner outside the Fourier-block case, and Proposition~\ref{prop:main-fourier-block} identifies the remaining case as Karlsson or Tao. We now combine these results.}

\subproofheading{{Proposition~\ref{thm:finitecornerintro} (Absence of every finite-corner witness forces Karlsson or Tao; restated)}}
{\emph{If an order-six complex Hadamard matrix $H$ has no equivalent dephased representative with a finite-corner witness, then}}
\begin{equation}
{[H]\in\KK_6^{(3)}\qquad\text{or}\qquad[H]\in\mathcal T_6.}
\label{eq:supp-failed-search-routing}
\end{equation}

\begin{proof}[Proof of Proposition~\ref{thm:finitecornerintro}]
{We prove the contrapositive. Let $H$ be an order-six complex Hadamard matrix satisfying $H\not\sim S_6^{(0)}$ and $[H]\notin\KK_6^{(3)}$. Choose any block partition. By Proposition~\ref{prop:main-corner-routing}, either a block can be moved to the upper left and dephased so that both normalized invertible-candidate sets are nonempty and finite, or all four blocks are order-three complex Hadamard matrices. The second possibility is excluded by Proposition~\ref{prop:main-fourier-block} under the standing assumptions. In the first possibility, all $3\times3$ submatrices are invertible by Lemma~\ref{lem:singularcorner} and Proposition~\ref{prop:publishedinputs}(1). In particular, the two complementary blocks are invertible. The selected representative therefore has a finite-corner witness. Thus every matrix that is neither Karlsson nor Tao has such a witness, which is equivalent to the conclusion of Proposition~\ref{thm:finitecornerintro}.}
\end{proof}

\subsubsection{{Proof of Proposition~\ref{prop:directcompletion}}}

\begin{proposition}[Direct finite completion]
\label{prop:directcompletion}
Fix a dephased corner $E$ for which $\mathcal B_E$ and $\mathcal C_E$ are finite and each contains an invertible candidate.  Because every member of a candidate set has the same Gram matrix, every candidate is then invertible. Every normalized Hadamard completion with invertible complementary blocks occurs in the finite list
\begin{equation}
(B,C)\in\mathcal B_E\times\mathcal C_E,
\qquad
D=-CE^{\dagger}(B^{-1})^{\dagger}.
\label{eq:completionformula}
\end{equation}
Conversely, a candidate pair in Eq.~\eqref{eq:completionformula} produces a complex Hadamard matrix exactly when every entry of the resulting $D$ has modulus one.
\end{proposition}

This is Sz\"oll\H{o}si's fixed-corner embedding criterion, including the forced block $D=-CE^{\dagger}(B^{-1})^{\dagger}$ \MainCiteSzollosi. We restate and reprove it here because the completed procedure of Definition~\ref{def:completed-output} requires both directions for the complete normalized candidate sets rather than for a selected generic branch.

\begin{proof}
Any completion satisfies the off-diagonal block equation $EC^{\dagger}+BD^{\dagger}=0$.  Since $B$ is invertible, solving this equation for $D$ gives Eq.~\eqref{eq:completionformula}; hence the formula is necessary and unique.

Conversely, choose $B\in\mathcal B_E$ and $C\in\mathcal C_E$ invertible and define $D$ by Eq.~\eqref{eq:completionformula}.  The candidate equations give
\begin{equation}
\begin{aligned}
C^{\dagger}D
&=-C^{\dagger}CE^{\dagger}(B^{-1})^{\dagger}\\
&=-(6I_3-E^{\dagger}E)E^{\dagger}(B^{-1})^{\dagger}\\
&=-E^{\dagger}(6I_3-EE^{\dagger})(B^{-1})^{\dagger}
=-E^{\dagger}B.
\end{aligned}
\label{eq:crossblockcompletion}
\end{equation}

Set $X=EE^{\dagger}$.  Since $BB^{\dagger}=6I_3-X$, direct expansion gives
\begin{equation}
\begin{aligned}
E(6I_3-E^{\dagger}E)E^{\dagger}
&=6X-X^2,\\
B(6I_3-B^{\dagger}B)B^{\dagger}
&=6(6I_3-X)-(6I_3-X)^2\\
&=6X-X^2.
\end{aligned}
\label{eq:completionidentity}
\end{equation}
Therefore
\begin{equation}
\begin{aligned}
D^{\dagger}D
&=B^{-1}E(6I_3-E^{\dagger}E)E^{\dagger}(B^{-1})^{\dagger}\\
&=B^{-1}B(6I_3-B^{\dagger}B)B^{\dagger}(B^{-1})^{\dagger}\\
&=6I_3-B^{\dagger}B.
\end{aligned}
\label{eq:Dcolumnnorm}
\end{equation}

Equations~\eqref{eq:candidatesets}, \eqref{eq:crossblockcompletion}, and \eqref{eq:Dcolumnnorm}, together with the adjoint of Eq.~\eqref{eq:crossblockcompletion}, give
\begin{equation}
\begin{pmatrix}E&B\\C&D\end{pmatrix}^{\!\dagger}
\begin{pmatrix}E&B\\C&D\end{pmatrix}=6I_6.
\label{eq:completedunitarity}
\end{equation}
The entries of $E,B,C$ already have modulus one.  If the entries of $D$ do too, Eq.~\eqref{eq:completedunitarity} is exactly the complex Hadamard condition.  If any entry of $D$ fails to have modulus one, the block matrix fails the entrywise requirement in Definition~\ref{def:hadamard}.  This proves both directions.

\end{proof}

\subsection{{Proof of Proposition~\ref{prop:algorithmvalid}}}
\subproofheading{{Proposition~\ref{prop:algorithmvalid} (Soundness of the retained output; restated)}}
{\emph{Every matrix retained by the branch-complete finite-dilation procedure is a complex Hadamard matrix.}}

\begin{proof}
{A retained pair consists of normalized candidates $B\in\mathcal B_E$ and $C\in\mathcal C_E$ with the forced block $D=-CE^\dagger(B^{-1})^\dagger$. Proposition~\ref{prop:directcompletion} shows that the resulting block matrix is unitary up to the factor $6$. The retention test requires every entry of $D$ to have modulus one, while the entries of $E,B,C$ are phases by construction. The retained matrix is therefore complex Hadamard.}
\end{proof}

{Proposition~\ref{prop:directcompletion} therefore gives a finite reconstruction procedure from any witnessing corner: solve the two candidate systems, construct $D$ from Eq.~\eqref{eq:completionformula}, and retain the pairs for which $D$ is unimodular. This is the output rule in Definition~\ref{def:completed-output}.}

\subproofheading{{Oriented leading coefficients certify finiteness}}
\begin{lemma}[Oriented leading coefficients certify finiteness]
\label{lem:supp-oriented-leading-finiteness}
Fix a normalized physical side fiber and form its fundamental sextic for each of the two orders of the noninitial rows. If both sextics have nonzero leading coefficient, then the side fiber is finite. Consequently, if the two horizontal and two vertical oriented leading coefficients are nonzero and the actual adjacent blocks are invertible, the selected corner is a finite-corner witness.
\end{lemma}

\begin{proof}
{An infinite fiber has an infinite ordered coordinate projection, forcing the cleared Haagerup identity in Eqs.~\eqref{eq:Phi}--\eqref{eq:Gamma}, and hence its corresponding oriented sextic leading coefficient, to vanish.  Applying this to both row orders and to the transpose contradicts the hypotheses.}
\end{proof}

\subsection{{Proof of Proposition~\ref{prop:karlsson-finite-corner}}}
\label{sec:supp-karlsson-witnesses}

The preceding argument supplies finite-corner witnesses outside the named exceptional sectors. To complete the classification proof, we now establish the corresponding witness theorem for every Karlsson matrix.

The simultaneous-degeneracy boundary is the affine Fourier family
\begin{equation}
F_6^{(2)}(z_1,z_2)=
\begin{pmatrix}
1&1&1&1&1&1\\
1&-1&z_1&-z_1&z_2&-z_2\\
1&1&\omega&\omega&\omega^2&\omega^2\\
1&-1&\omega z_1&-\omega z_1&\omega^2z_2&-\omega^2z_2\\
1&1&\omega^2&\omega^2&\omega&\omega\\
1&-1&\omega^2z_1&-\omega^2z_1&\omega z_2&-\omega z_2
\end{pmatrix},
\label{eq:supp-fourier-family}
\end{equation}
together with its transpose. It is $H_2$-reducible.

\subproofheading{{Proposition~\ref{prop:karlsson-finite-corner} (Finite-corner witnesses for the Karlsson family; restated; computer-assisted)}}
{\emph{Every equivalence class in $\KK_6^{(3)}$ has a finite-corner witness. Consequently,}}
\begin{equation}
{\KK_6^{(3)}\subset\FCAtlas_6^{\mathrm{fc}}.}
\label{eq:supp-full-karlsson-containment}
\end{equation}

{The resultant eliminations, affine-Fourier boundary calculation, and Bernstein-positivity argument used below are checked exactly in Lean; see the classification and Karlsson-coverage certificate group in Table~\ref{tab:supp-certificate-index}.}

\begin{proof}
Starting from Karlsson's complete block parametrization \cite{Karlsson2011H2,Karlsson2011ThreeParameter}, our division-free seam analysis uses the following canonical nondegenerate chart:
\begin{equation}
u=\cos\theta,\qquad v=\sin\theta,\qquad p=e^{\I\phi},
\qquad
0\leq\theta,\phi<\pi.
\label{eq:supp-karlsson-canonical}
\end{equation}
Thus $u^2+v^2=1$, $v\geq0$, and $\operatorname{Im}p\geq0$. The two M\"obius maps have distinct single-degeneracy curves.  Their simultaneous-degeneracy locus gives the affine-Fourier family in Eq.~\eqref{eq:supp-fourier-family} and its transpose.  The single-degeneracy curves will instead be covered below by reversing the Karlsson half-angle orientation.

We first cover the simultaneous-degeneracy boundary. In zero-based indices, use the six row/column corners
\begin{equation}
\begin{gathered}
(012;024),\ (013;024),\ (025;024),\\
(012;025),\ (013;025),\ (025;025).
\end{gathered}
\label{eq:supp-fourier-six-corners}
\end{equation}
For each corner let $W_\nu(z_1,z_2)$ be the product of its three block determinants and four oriented leading fundamental coefficients. Exact elimination produces all $245$ pairwise resultants among their irreducible factors. Reciprocal unit-circle reduction leaves $25$ irreducible first-phase conditions. Quotient-ring tests eliminate every compatible second phase, and the $18$ residual second-phase branches are excluded by exact reciprocal unit-circle resultants. Hence
\begin{equation}
\{(z_1,z_2)\in\T^2:W_1=\cdots=W_6=0\}=\varnothing.
\label{eq:supp-fourier-no-common-zero}
\end{equation}
At least one of the six corners therefore has nonzero determinant and oriented-leading-coefficient factors at every affine-Fourier point. Lemma~\ref{lem:supp-oriented-leading-finiteness} makes it a finite-corner witness. Transposition gives the same result for the transposed family.

It remains to cover the nondegenerate chart.  Use the single mixed/mixed corner
\begin{equation}
I=J=\{1,3,4\}.
\label{eq:supp-karlsson-fixed-corner}
\end{equation}
After cross-ratio dephasing, its eleven sufficient witness factors are three block determinants and the two elementary factors in each of four oriented leading coefficients.  Modulo $u^2+v^2=1$, the determinants are unit monomials times
\begin{equation}
3+\I\sqrt3(u+pv),\qquad
3+\I\sqrt3(pv-u),\qquad
p+\frac{\I}{\sqrt3}(pu-v),
\label{eq:supp-karlsson-determinants}
\end{equation}
up to nonzero constants and unit monomials. They cannot vanish because $|u|+|v|\leq\sqrt2<\sqrt3$. Four further factors are the repeated squares
\begin{equation}
-3\left(pv-u+\frac{\I}{\sqrt3}\right)^2,\qquad
-3\left(pu-v+\frac{\I p}{\sqrt3}\right)^2.
\label{eq:supp-karlsson-squares}
\end{equation}
Either vanishing equation would give $v\operatorname{Im}p=-1/\sqrt3$, impossible on the canonical domain.

The remaining four factors are quadratic in the relevant unit phase.  Put
\begin{equation}
t=\tan(\theta/2),\qquad
M_+(t,p)=(1+t^2)(p^2-1)
 +\I\sqrt3(1-t^2)(p^2+1).
\label{eq:supp-karlsson-Mplus}
\end{equation}
Define the factor in the opposite orientation by
\begin{equation}
M_-(t,p)=(1+t^2)(p^2-1)
 -\I\sqrt3(1-t^2)(p^2+1).
\label{eq:supp-karlsson-Mminus}
\end{equation}
A unit root would be common to the quadratic and its reciprocal-conjugate.  Their four exact resultants are, up to nonzero constants and unit monomials,
\begin{equation}
\begin{split}
&t^2M_+(t,p)^2R(t,p),\quad t^2M_+(t,p)^2R(t,p),\\
&t^2p^4M_+(t,p)^2R(t,p),\quad
t^2p^4M_+(t,p)^2R(t,p).
\end{split}
\label{eq:supp-karlsson-resultants}
\end{equation}
These factors arise by eliminating the unit-phase variable between each hard quadratic and its reciprocal-conjugate. The explicit factors $t$ and $M_+$ record the known boundary degeneracies. After removing them, all four orientations leave the same residual polynomial $R(t,p)$. The factor $t=0$ is on the already covered Fourier boundary.  The equation $M_+=0$ by itself is not a Fourier seam: it is one of the two singly degenerate M\"obius curves.  For $t>0$, however, the reciprocal half-angle representative is completely explicit.  To fix conventions, put
\begin{equation}
\begin{gathered}
 u_t=\frac{1-t^2}{1+t^2},\qquad v_t=\frac{2t}{1+t^2},\qquad
 F_2=\begin{pmatrix}1&1\\1&-1\end{pmatrix},\\
 \Lambda(t,p)=\begin{pmatrix}u_t&pv_t\\v_t/p&-u_t\end{pmatrix},\qquad
 A(t,p)=F_2\left(-\frac12I_2+\frac{\I\sqrt3}{2}\Lambda(t,p)\right),
 \qquad B(t,p)=-F_2-A(t,p),\\
 Z_R(z)=\begin{pmatrix}1&1\\z&-z\end{pmatrix},\qquad
 Z_L(z)=\begin{pmatrix}1&z\\1&-z\end{pmatrix}.
\end{gathered}
\label{eq:supp-karlsson-block-data}
\end{equation}
In the standard $2\times2$-block order, the raw Karlsson matrix used here is
\begin{equation}
 H_{\mathrm K}(t,p,z_1,z_2,z_3,z_4)=
 \begin{pmatrix}
 F_2&Z_R(z_1)&Z_R(z_2)\\
 Z_L(z_3)&\frac12Z_L(z_3)A(t,p)Z_R(z_1)
             &\frac12Z_L(z_3)B(t,p)Z_R(z_2)\\
 Z_L(z_4)&\frac12Z_L(z_4)B(t,p)Z_R(z_1)
             &\frac12Z_L(z_4)A(t,p)Z_R(z_2)
 \end{pmatrix}.
\label{eq:supp-karlsson-raw-matrix}
\end{equation}
Here $(z_1,z_2,z_3,z_4)\in\T^4$ satisfy Karlsson's M\"obius relations
\begin{equation}
\begin{aligned}
z_3^2&=\mathcal M_A(z_1^2)=\mathcal M_B(z_2^2),\\
z_4^2&=\mathcal M_B(z_1^2)=\mathcal M_A(z_2^2),
\end{aligned}
\qquad
\mathcal M_X(\zeta)=
\frac{X_{12}^{\,2}\zeta-X_{11}^{\,2}}
{\overline{X_{11}}^{\,2}\zeta-\overline{X_{12}}^{\,2}}
\quad(X=A,B).
\label{eq:supp-karlsson-mobius-domain}
\end{equation}
When a displayed M\"obius denominator degenerates, these relations are read after cross multiplication and the resulting point is handled by the seam alternatives already separated above \cite{Karlsson2011H2,Karlsson2011ThreeParameter}. Let $\widehat H_{\mathrm K}$ be the same matrix with rows and columns both put in the mixed order $(1,3,4,2,5,6)$.  Let $\tau=(1\ 4)(2\ 5)(3\ 6)$, which exchanges its two three-column halves, and set
\begin{equation}
 r=(1,z_3^{-1},-z_3^{-1},-1,z_4^{-1},-z_4^{-1}).
\label{eq:supp-karlsson-reciprocal-phases}
\end{equation}
Then, entry by entry,
\begin{equation}
 \widehat H_{\mathrm K}(1/t,p,-z_2,-z_1,z_3^{-1},z_4^{-1})_{ij}
 =r_i\widehat H_{\mathrm K}(t,p,z_1,z_2,z_3,z_4)_{i,\tau(j)}.
\label{eq:supp-karlsson-reciprocal-entrywise}
\end{equation}
Indeed, with $S=\begin{psmallmatrix}0&1\\1&0\end{psmallmatrix}$ and $D=\begin{psmallmatrix}1&0\\0&-1\end{psmallmatrix}$, direct substitution in Eq.~\eqref{eq:supp-karlsson-block-data} gives
\begin{equation}
 A(1/t,p)=S B(t,p)D,\qquad B(1/t,p)=S A(t,p)D.
\label{eq:supp-karlsson-reciprocal-cores}
\end{equation}
Equations~\eqref{eq:supp-karlsson-raw-matrix} and \eqref{eq:supp-karlsson-reciprocal-cores} then give Eq.~\eqref{eq:supp-karlsson-reciprocal-entrywise}, without using a M\"obius denominator.  Since $|z_3|=|z_4|=1$, every component of $r$ is a phase, and the fixed mixed-order permutation transports this identity back to standard order.  Thus
\begin{equation}
M_+(1/t,p)=t^{-2}M_-(t,p),
\qquad
H_{\mathrm K}(t,p,z_1,z_2,z_3,z_4)\sim
H_{\mathrm K}(1/t,p,-z_2,-z_1,z_3^{-1},z_4^{-1}).
\label{eq:supp-karlsson-reciprocal-orientation}
\end{equation}
Moreover,
\begin{equation}
M_++M_-=2(1+t^2)(p^2-1),
\qquad
M_+-M_-=2\I\sqrt3(1-t^2)(p^2+1).
\label{eq:supp-karlsson-common-degeneracy}
\end{equation}
On the canonical domain $t>0$, $|p|=1$, $\operatorname{Im}p\ge0$, and $p\ne-1$, the equations $M_+=M_-=0$ therefore force $p=1$ and $t=1$.  This is the doubly degenerate affine-Fourier point already covered above.  Away from that point, if $M_+(t,p)=0$, then $M_-(t,p)\ne0$, so the equivalent reciprocal representative has $M_+\ne0$.  Thus one of the two orientations always reaches the certified fixed-corner chart.  The same residual $R$ occurs in all four resultants.

Set $p=(1+\I\rho)/(1-\I\rho)$, with $\rho\geq0$.  After a positive denominator is cleared, $R$ is the negative of a real polynomial $S(t,\rho)$.  Compactify the quadrant by
\begin{equation}
t=\frac{x}{1-x},\qquad \rho=\frac{y}{1-y},
\qquad (x,y)\in[0,1]^2.
\label{eq:supp-karlsson-compactification}
\end{equation}
{The Bernstein-positivity step is exact: all subdivision boxes and coefficients are rational, and Lean checks the resulting positivity conclusion.}
Write the compactified polynomial in the tensor Bernstein basis $\sum_{i=0}^{16}\sum_{j=0}^{8}b_{ij} \binom{16}{i}x^i(1-x)^{16-i} \binom{8}{j}y^j(1-y)^{8-j}$, where the numbers $b_{ij}$ are its Bernstein coefficients. Exact dyadic de Casteljau subdivision~\cite{Farouki2012} reexpresses these coefficients on each half-box and terminates in ten rational boxes, on each of which every Bernstein coefficient is strictly positive.  Since the Bernstein basis functions are nonnegative and sum to one, this proves $S(t,\rho)>0$ throughout the closed quadrant, so $R\neq0$. All eleven witness factors in the fixed corner are nonzero in whichever of the two equivalent orientations has $M_+\ne0$. Lemma~\ref{lem:supp-oriented-leading-finiteness} makes both complete side fibers finite, and forced completion recovers that representative, hence also its original Karlsson equivalence class. Together with the doubly-degenerate Fourier-boundary argument, this proves $\KK_6^{(3)}\subset\FCAtlas_6^{\mathrm{fc}}$. Tao belongs to $\FCAtlas_6^{\mathrm{fc}}$ but not to $\KK_6^{(3)}$: every cross ratio of Tao's matrix is a cubic root of unity, whereas an $H_2$-reducible matrix has a cross ratio equal to $-1$, which is not a cubic root of unity. Hence the inclusion is proper.
\end{proof}

\subproofheading{{Proof of Theorem~\ref{thm:classification}}}
\begin{proof}[Proof of Theorem~\ref{thm:classification}]
{Proposition~\ref{thm:finitecornerintro} gives a finite-corner witness when the matrix is neither Karlsson nor Tao, Proposition~\ref{prop:tao-finite-corner} gives one for Tao, and Proposition~\ref{prop:karlsson-finite-corner} gives one for every Karlsson matrix. Hence every order-six Hadamard matrix is equivalent to a dephased matrix having a finite-corner witness.}
\end{proof}

\subproofheading{{Corollary~\ref{cor:atlas-classification} (Exact retained output; restated)}}
{\emph{The branch-complete dilation procedure is sound and exhaustive:}}
\begin{equation}
{\FCAtlas_6^{\mathrm{fc}}=\HH_6.}
\label{eq:supp-complete-classification}
\end{equation}

\begin{proof}
{Theorem~\ref{thm:classification} supplies a finite-corner witness for every class. At that corner the actual adjacent blocks occur in the complete finite candidate lists, and Proposition~\ref{prop:directcompletion} recovers the actual fourth block. Thus $\HH_6\subseteq\FCAtlas_6^{\mathrm{fc}}$. Proposition~\ref{prop:algorithmvalid} gives the reverse inclusion. No regular-sheet or closure assumption is used.}
\end{proof}

\section{{Proofs of Theorems~\ref{thm:regular-seed-domain} and~\ref{thm:global-product-regular-escape}}}
\label{sec:supp-reconstruction-geometry}
\label{sec:supp-atlas-proofs}

{This section proves the results stated in Sec.~IV of the main text. We first construct the division-free fixed-Gram incidence correspondence, then derive the quadratic--cubic reconstruction, determine its physical seed domain, prove the global product-regular reach theorem, and record the ramification geometry.}

\def\AtlasEmbedded{1}

\subproofheading{{The global fixed-Gram incidence correspondence}}
\label{sec:supp-global-incidence}

{These equations are the uncancelled parent system behind the companion polynomials and product quadratic used in Sec.~IV of the main text. They are also needed when a product-regular guard vanishes because the rational companion formula would then lose common numerator--denominator zeros. The system below therefore supplies both the regular derivation and the exceptional-fiber interpretation.}

For a triple $z=(z_1,z_2,z_3)\in(\C^*)^3$ of nonzero complex numbers, define the elementary symmetric functions
\begin{equation}
e_1(z)=z_1+z_2+z_3,\qquad
e_2(z)=z_1z_2+z_1z_3+z_2z_3,\qquad
e_3(z)=z_1z_2z_3.
\label{eq:supp-elementary}
\end{equation}

Let $^{\#}$ be the involution of the Laurent polynomial ring that replaces each torus variable by its inverse and conjugates scalar coefficients.  On the torus, $f^\#=\overline f$. Extend it to the localized Laurent field by
\begin{equation}
\left(\frac fg\right)^\#=\frac{f^\#}{g^\#}
\qquad(g\ne0).
\label{eq:supp-sharp-localization}
\end{equation}
For the paired variables $x,y$, the physical torus is $\T^6=\{(x,y)\in\C^6:|x_j|=|y_j|=1\text{ for }j=1,2,3\}$.

Fix trace data $(s,s^\#,t,t^\#,r,r^\#)$.
For paired triples $x,y\in(\C^*)^3$, impose
\begin{equation}
\begin{aligned}
e_1(x)&=s,& e_2(x)&=s^\#e_3(x),\\
e_1(y)&=t,& e_2(y)&=t^\#e_3(y),\\
\sum_jy_j\!\prod_{k\ne j}x_k&=re_3(x),&
\sum_jx_j\!\prod_{k\ne j}y_k&=r^\#e_3(y).
\end{aligned}
\label{eq:supp-six-incidence}
\end{equation}

The complex algebraic-torus fiber is, precisely,
\begin{equation}
\mathfrak F^{\C}_{s,t,r}
=\{(x,y)\in(\C^*)^6:\text{the six equations in
Eq.~\eqref{eq:supp-six-incidence} hold}\}.
\label{eq:supp-complex-fiber}
\end{equation}
Define its physical fixed locus by
\begin{equation}
\mathfrak F^{\mathrm{phys}}_{s,t,r}
=\mathfrak F^{\C}_{s,t,r}\cap\T^6.
\label{eq:supp-physical-fiber}
\end{equation}

Associate to either kind of point the matrix
\begin{equation}
X(x,y)=\begin{pmatrix}1&1&1\\x_1&x_2&x_3\\y_1&y_2&y_3\end{pmatrix}.
\label{eq:supp-X}
\end{equation}

\subproofheading{{Exact fixed-Gram presentation}}

\begin{lemma}[Exact fixed-Gram presentation]
\label{lem:supp-incidence-gram}
On the physical torus, Eq.~\eqref{eq:supp-six-incidence} holds if and only if
\begin{equation}
X(x,y)X(x,y)^\dagger=
\begin{pmatrix}3&s^\#&t^\#\\s&3&r^\#\\t&r&3\end{pmatrix}.
\label{eq:supp-Xgram}
\end{equation}
Simultaneously permuting the three pairs $(x_j,y_j)$ defines an $S_3$-action on the ordered solution fiber and only permutes the columns of $X$. The ordered fiber is identified with the normalized candidate set; its quotient parametrizes candidates only up to column relabeling.
\end{lemma}

\begin{proof}
The first off-diagonal row product is $\sum_jx_j=s$, and its conjugate is the reciprocal sum $\sum_jx_j^{-1}=e_2(x)/e_3(x)=s^\#$.  The same argument gives the two relations for $y$.  Dividing the fifth equation of Eq.~\eqref{eq:supp-six-incidence} by $e_3(x)$ gives $\sum_jy_j/x_j=r$, and dividing the sixth by $e_3(y)$ gives $\sum_jx_j/y_j=r^\#$. These are exactly the six off-diagonal entries of Eq.~\eqref{eq:supp-Xgram}. The diagonal entries equal three because all three entries of each row are phases.  The converse follows by reversing the same divisions, which are legitimate because all coordinates are nonzero.  The final assertion is immediate from Eq.~\eqref{eq:supp-X}.
\end{proof}

\subproofheading{{Physical soundness and completeness for one corner}}

For the four-phase corner $E(a,b,c,d)$, the main article gives all six trace data in Eq.~\eqref{eq:trace-data-main}.  Its horizontal data are
\begin{equation}
s_h=-(1+a+b),\qquad t_h=-(1+c+d),\qquad
r_h=-\left(1+\frac ca+\frac db\right),
\label{eq:supp-horizontal-data}
\end{equation}
with their $^{\#}$-images, meaning $a,b,c,d$ inverted inside these same expressions, e.g.\ $s_h^\#=-(1+a^{-1}+b^{-1})$, not the numerical reciprocal $1/s_h$.  The vertical data are obtained from $E^T$:
\begin{equation}
s_v=-(1+a+c),\qquad t_v=-(1+b+d),\qquad
r_v=-\left(1+\frac ba+\frac dc\right).
\label{eq:supp-vertical-data}
\end{equation}
By Lemma~\ref{lem:supp-incidence-gram}, the physical fibers $\mathfrak F^{\mathrm{phys}}_{s_h,t_h,r_h}$ and $\mathfrak F^{\mathrm{phys}}_{s_v,t_v,r_v}$ are exactly the complete normalized candidate sets $\mathcal B_E$ and $\mathcal C_E$, after transposing the vertical candidate.  No corresponding identification is claimed for the unrestricted complex fibers.

Introduce independent paired triples $x,y,p,q$ and form
\begin{equation}
B=X(x,y),\qquad C^T=X(p,q).
\label{eq:supp-BC}
\end{equation}
Here all four triples belong to $(\C^*)^3$ on the algebraic chart and to $\T^3$ on its physical locus.

For a matrix $M$ with entries in this localized Laurent field, $M^\#$ means sharp applied entrywise, $(M^\#)_{ij}=(M_{ij})^\#$. Define the formal algebraic star operation by
\begin{equation}
M^\star=(M^\#)^T,
\label{eq:supp-formal-star}
\end{equation}
which is an algebraic surrogate for conjugate transpose. Sharp inverts the variables and conjugates scalar coefficients in each entry, then $T$ transposes the matrix. On the physical torus, sharp equals the ordinary adjoint.

Consider the star-stable complex open set
\begin{equation}
\det B(\det B)^\#\det C(\det C)^\#\ne0.
\label{eq:supp-complex-open}
\end{equation}
This is the complement of the zero set of the displayed Laurent polynomial, where all required inverses and their sharp transforms exist.

We define the formal completion
\begin{equation}
D^{\C}=-CE^\star(B^{-1})^\star.
\label{eq:supp-D}
\end{equation}
The word \emph{formal} means that this lower-right block is computed in the localized Laurent coordinate ring before any point is required to lie on the physical torus.

We then impose
\begin{equation}
D^{\C}_{ij}(D^{\C}_{ij})^\#=1,\qquad 1\le i,j\le3.
\label{eq:supp-nine-flatness}
\end{equation}

On the unrestricted complex algebraic torus, Eq.~\eqref{eq:supp-nine-flatness} is only a formal sharp-stable equation. It does not assert the pointwise analytic condition $|D^{\C}_{ij}|=1$, because $^\#$ is not complex conjugation there. The two statements become equivalent only after restriction to the physical compact torus, where $^\#$ evaluates as ordinary complex conjugation.

After multiplication by nonzero torus monomials and the displayed determinant factors, these are nine ordinary polynomial equations. Together with the twelve side equations they define a formal twenty-one-equation complex corner correspondence. On the physical torus, $M^\star=M^\dagger$, so the completion becomes
\begin{equation}
D^{\mathrm{phys}}=-CE^\dagger(B^{-1})^\dagger,
\label{eq:supp-physical-D}
\end{equation}

and Eq.~\eqref{eq:supp-nine-flatness} says exactly that every entry of $D^{\mathrm{phys}}$ has modulus one.

\begin{proposition}[Physical soundness and completeness for one corner]
\label{prop:supp-one-corner}
Let $\mathcal Y_E\subset\T^{16}$ be the set of $\eta=(a,b,c,d,x,y,p,q)$ such that: (i) $B=X(x,y)$ and $C^T=X(p,q)$ satisfy the incidence equations Eq.~\eqref{eq:supp-six-incidence} (and its transpose) for the corner $E(a,b,c,d)$; (ii) $B,C$ are invertible and Eq.~\eqref{eq:supp-physical-D} is entrywise unimodular; (iii) both complete physical side fibers at $E$ are nonempty and finite.  {After clearing denominators, (i) and the flatness part of (ii) define a locally closed Laurent-algebraic correspondence on $\det B\,\det C\ne0$.} Condition (iii) is not an algebraic equation, being the intrinsic finite-corner selector of Definition~\ref{def:finitecorner}.  The matrices reconstructed from $\mathcal Y_E$ are exactly the outputs of our completed finite-dilation procedure at the corner $E$.
\end{proposition}

\begin{proof}
The two physical incidence systems are exactly the complete physical candidate sets by Lemma~\ref{lem:supp-incidence-gram}.  The intrinsic finite-fiber condition is therefore exactly the guard in Definition~\ref{def:finitecorner}.  The determinant conditions select exactly the permitted candidate pairs. Proposition~\ref{prop:directcompletion} then states that Eq.~\eqref{eq:supp-physical-D} is the unique possible completion and that Eq.~\eqref{eq:supp-nine-flatness} is equivalent to the Hadamard entrywise condition.  These are precisely the five steps of our completed procedure.
\end{proof}

\subsection{{Proof of Theorem~\ref{thm:regular-seed-domain}}}
\label{sec:supp-generic-cover}

Set $u=e_3(x)$.  The first two equations of Eq.~\eqref{eq:supp-six-incidence} are equivalent to saying that the $x_j$ are the roots of the self-inversive cubic
\begin{equation}
q_{s,u}(z)=z^3-sz^2+us^\#z-u.
\label{eq:supp-product-cubic}
\end{equation}
At a physical torus point, where sharp evaluates as complex conjugation and $u^\#=\overline u$, self-inversive means that the reciprocal-conjugate root set equals the original root set when $uu^\#=1$.

On the locus where the companion elimination is defined, the uncancelled linear relation is
\begin{equation}
A(z)+B(z)y=0,
\label{eq:supp-companion-linear}
\end{equation}
where
\begin{equation}
A(z)=z(rr^\#-tt^\#)
\bigl(-ss^\#z+2s+2s^\#z^2-3z\bigr)
\label{eq:supp-A-explicit}
\end{equation}
and $B(z)=b_0+b_1z+b_2z^2+b_3z^3$ with
\begin{equation}
\begin{aligned}
b_0={}&-(r^\#)^2t+r^\#stt^\#-3r^\#s-s^2t^\#,
\\
b_1={}&r(r^\#)^2-rr^\#st^\#+(r^\#)^2s^\#t
-r^\#ss^\#tt^\#+3r^\#ss^\#-r^\#tt^\#
+s^2s^\#t^\#+st(t^\#)^2,
\\
b_2={}&-r(r^\#)^2s^\#+rr^\#ss^\#t^\#+rr^\#t^\#
-rs(t^\#)^2-r^\#s(s^\#)^2+r^\#s^\#tt^\#
-3ss^\#t^\#-t(t^\#)^2,
\\
b_3={}&-rr^\#s^\#t^\#+r(t^\#)^2+r^\#(s^\#)^2+3s^\#t^\#.
\label{eq:supp-B-explicit}
\end{aligned}
\end{equation}
Define the reciprocal fundamental sextic by
\begin{equation}
\Phi_{\rm fund}(z)=z^3\bigl(A(z)A(z)^\#-B(z)B(z)^\#\bigr)
=\sum_{k=0}^{6}c_kz^k.
\label{eq:supp-fundamental-sextic}
\end{equation}
Thus $c_k$ is the coefficient of $z^k$ in the displayed polynomial.  The leading coefficient is
\begin{equation}
\begin{aligned}
c_6=-(&r^2t^\#-rs^\#tt^\#+3rs^\#+(s^\#)^2t)\\
&\mathrel{\phantom{-}}\cdot
(rr^\#s^\#t^\#-r(t^\#)^2-r^\#(s^\#)^2-3s^\#t^\#).
\end{aligned}
\label{eq:supp-c6-explicit}
\end{equation}
Exact expansion also gives
\begin{equation}
\begin{gathered}
c_5=-2sc_6,\qquad c_1=-2s^\#c_0,\\
s^\#c_3+(1+ss^\#)(c_4-s^2c_6)=0,\\
s^\#c_2-s(c_4-s^2c_6)-(s^\#)^3c_0=0.
\end{gathered}
\label{eq:supp-coefficient-identities}
\end{equation}
These identities give
\begin{equation}
\Phi_{\rm fund}(z)=c_6q_{s,u}(z)q_{s,v}(z),
\label{eq:supp-sextic-factor}
\end{equation}
where $u,v$ are the two roots of
\begin{equation}
(1+ss^\#)c_6U^2+c_3U+(1+ss^\#)c_0=0.
\label{eq:supp-product-quadratic}
\end{equation}

The quadratic--cubic cover is used only when $c_6\neq0$.  If $c_6=0$, Eq.~\eqref{eq:supp-product-quadratic} drops degree and neither the two-root description nor the quotients in Eq.~\eqref{eq:supp-uv} below are invoked. {Such a point remains in the division-free incidence correspondence of Sec.~\ref{sec:supp-global-incidence}. A finite physical fiber is retained directly, while a positive-dimensional physical fiber is handled by Proposition~\ref{prop:main-infinite-fiber-trichotomy} and the corner-routing argument.}

Operationally the reconstruction proceeds in the opposite order from this derivation: first solve Eq.~\eqref{eq:supp-product-quadratic} for the product $u$, then solve the cubic $q_{s,u}$ for $x_1,x_2,x_3$, and finally recover $y_j=-A(x_j)/B(x_j)$. The same three steps applied to $E^T$ give the vertical candidate, and the affine sheet matching below pairs the two product roots consistently.

This final quotient is used only when $B(x_j)\neq0$.  If $B(x_j)=0$, no value is assigned by the rational companion formula.  Instead one returns to the uncancelled parent equations, equivalently to $\mathfrak F^{\C}_{s,t,r}$ in Eq.~\eqref{eq:supp-complex-fiber}: $A(x_j)\neq0$ gives no solution, whereas $A(x_j)=B(x_j)=0$ is retained and decided by the complete division-free fiber analysis.  Thus no $0/0$ branch is discarded.

\begin{lemma}[Product-factor identities; computer-assisted]
\label{lem:supp-product-factor}
Equations~\eqref{eq:supp-sextic-factor} and \eqref{eq:supp-product-quadratic} are polynomial identities on the regular fixed-Gram locus.  In particular,
\begin{equation}
u+v=-\frac{c_3}{(1+ss^\#)c_6},\qquad
uv=\frac{c_0}{c_6}.
\label{eq:supp-uv}
\end{equation}
\end{lemma}

\begin{proof}
The coefficient identities in this proof are checked by exact symbolic arithmetic; see the generic four-phase reconstruction certificate group in Table~\ref{tab:supp-certificate-index}.
Expand $c_6q_{s,u}q_{s,v}$.  Its leading coefficient is $c_6$, and reciprocity fixes the coefficients opposite one another.  Matching the constant term gives $uv=c_0/c_6$, while matching the middle coefficient gives the displayed expression for $u+v$. Substitution into the remaining four coefficients gives zero identically after the fixed-Gram trace relations are used. Conversely, Vieta's formulas for Eq.~\eqref{eq:supp-product-quadratic} give Eq.~\eqref{eq:supp-uv}, completing the factorization.
\end{proof}

\begin{lemma}[{Sufficiency of the regular companion chart; computer-assisted}]
\label{lem:supp-companion-sufficiency}
{Let $Q$ be the product quadratic.  If $c_6R\ne0$, $Q(u)=0$, $q_{s,u}$ is simple, $R=\operatorname{Res}_z(q_{s,u},B)$, and $y_j=-A(x_j)/B(x_j)$ at its roots, then $\sum_jy_j=t$, $\sum_jy_j/x_j=r$, and $y_jy_j^\#=1$.  Hence Eq.~\eqref{eq:supp-six-incidence} holds and the eliminant introduces no extraneous regular point.}
\end{lemma}

\begin{proof}
The quotient-algebra reductions below are checked exactly; see the generic four-phase reconstruction certificate group in Table~\ref{tab:supp-certificate-index}.
{Work in the ordered splitting algebra of $q_{s,u}$, localized at its Vandermonde and $\prod_jB(x_j)$.  Sharp reciprocity extends there by $u^\#=u^{-1}$ and $x_j^\#=x_j^{-1}$.  If $M_B$ is multiplication by $B$ in $K[z]/(q_{s,u})$, then $\det M_B$ equals $R$ up to a unit and exact reduction gives}
\begin{equation}
{\det(M_B)\left(\sum_jy_j-t\right)=Q(u)P_t(u),
\qquad
u\det(M_B)\left(\sum_j\frac{y_j}{x_j}-r\right)=Q(u)P_r(u).}
\label{eq:supp-companion-remainders}
\end{equation}
{Since $u\det M_B\ne0$, these give the two traces. Eq.~\eqref{eq:supp-sextic-factor} gives $y_jy_j^\#=1$. Vieta and sharp then give the other four parent equations.}
\end{proof}

\subproofheading{{Definition~\ref{def:product-regularity} (Product-regular lift; restated)}}
{Here ``lift'' means the affine-minus matched horizontal--vertical lift in Eq.~\eqref{eq:supp-minus-matching}, with common root labelling modulo simultaneous $S_3$.} The lift is \emph{product regular} when the product quadratics have their expected degree, the coordinate cubics have simple roots, all companion denominators are nonzero, and
\begin{equation}
{\begin{gathered}
\det(E)\det(B)(\det B)^\#\det(C)(\det C)^\#c_{6,h}c_{6,v}
\delta_h\delta_v
{}\operatorname{Disc}(q_h)\operatorname{Disc}(q_v)R_hR_v\ne0,
\\ \delta_\nu=r_\nu r_\nu^\#-t_\nu t_\nu^\#\quad(\nu=h,v).
\end{gathered}}
\label{eq:supp-product-regularity-guards}
\end{equation}
{Here $R_\nu=\operatorname{Res}(q_\nu,B_\nu)$ after clearing denominators.} The vertical quantities are obtained from $E^{\mathsf T}$. These guards select the quadratic--cubic formulas treated below. The global incidence system remains valid when one of them vanishes.

The following result identifies which part of this algebraic double cover is physical. It is the positivity step used in Theorem~\ref{thm:regular-seed-domain}.

\subproofheading{{Complement positivity}}

\begin{lemma}[Complement positivity; computer-assisted]
\label{lem:supp-complement-positivity}
Let $E$ be an invertible $3\times3$ phase matrix, set $G=6I_3-EE^\dagger$, and suppose the fixed-Gram product construction for $G$ satisfies the regularity guards.  If its normalized product discriminant satisfies $\omega_{\rm n}(G)\leq0$, then $G>0$.
\end{lemma}

\begin{proof}
The factorization and sign identities below are checked by exact symbolic arithmetic; see the generic four-phase reconstruction certificate group in Table~\ref{tab:supp-certificate-index}.
Gauge the off-diagonal entries of $G$ as $s=S$, $t=T$, and $r=R\zeta$, where $S,T,R\geq0$ and $|\zeta|=1$.  Put
\begin{equation}
\begin{gathered}
X=S^2,\quad Y=T^2,\quad Z=R^2,\quad
J=STR(\zeta+\zeta^{-1}),\\
p=X+Y+Z,\quad e_2=XY+XZ+YZ,\quad e_3=XYZ.
\end{gathered}
\label{eq:supp-positivity-invariants}
\end{equation}
Exact reduction of the product discriminant gives
\begin{equation}
\omega_{\rm n}(G)=
\frac{(R^2-T^2)^2\{L^2-\det(G)Q(J)\}}
{|P_1|^2|P_2|^2},
\label{eq:supp-positivity-factorization}
\end{equation}
where $P_1,P_2$ are the two nonzero companion factors,
\begin{equation}
\begin{aligned}
L={}&2R^4-R^2S^2T^2+5R^2S^2+5R^2T^2-45R^2\\
&+2S^4+5S^2T^2-45S^2+2T^4-45T^2+243,
\end{aligned}
\label{eq:supp-positivity-L}
\end{equation}
and
\begin{equation}
\begin{aligned}
Q(J)={}&4J^2-(p^2-14p+81)J-p^3+45p^2-567p+2187\\
&+(30-2p)e_2+(2p-30)e_3.
\end{aligned}
\label{eq:supp-positivity-Q}
\end{equation}
The same exact expansion yields
\begin{equation}
\begin{gathered}
\det G=27-3p+J,\qquad
\det(EE^\dagger)=27-3p-J,\\
Q(J)-Q(-J)=-2J\bigl((p-7)^2+32\bigr).
\end{gathered}
\label{eq:supp-positivity-sign-identities}
\end{equation}

Apply Eq.~\eqref{eq:supp-positivity-factorization} first to $EE^\dagger$.  The matrix $E$ is a physical point of its own fixed-Gram fiber, so one product root is a phase, and self-inversiveness makes the other a phase as well. Its normalized discriminant is therefore nonpositive, which gives $Q(-J)\geq0$ whenever the auxiliary elimination factors are nonzero. Those factors are nonzero Laurent polynomials in the entries of a phase matrix. Since the phase torus is Zariski dense, their common nonvanishing set is dense, and intersecting with the open set of invertible matrices preserves density. Every invertible phase matrix is therefore a limit of such points. Because $Q(-J)$ is polynomial in the Gram invariants, continuity gives $Q(-J)\geq0$ also when an auxiliary factor vanishes.

The parent fixed-Gram equations give the formal identity
\begin{equation}
BB^\star=G.
\label{eq:supp-positivity-formal-gram}
\end{equation}
Indeed, their six incidence equations are precisely the off-diagonal entries of this identity, while the diagonal entries equal three because each phase variable multiplied by its sharp is one. Consequently, the star-stable determinant guard gives
\begin{equation}
\det G=\det(B)(\det B)^\#\ne0.
\label{eq:supp-positivity-determinant-guard}
\end{equation}

If $J\geq0$, Eq.~\eqref{eq:supp-positivity-sign-identities} and $\det E\ne0$ give $\det G>0$. If $J<0$, the last identity gives $Q(J)>0$. The $\delta$- and companion-factor guards make the remaining prefactor in Eq.~\eqref{eq:supp-positivity-factorization} strictly positive. Hence $\omega_{\rm n}(G)\leq0$ implies
\begin{equation}
L^2-\det(G)Q(J)\leq0,
\end{equation}
so $\det G\geq L^2/Q(J)\geq0$. Equation~\eqref{eq:supp-positivity-determinant-guard} excludes equality, and therefore $\det G>0$. The diagonal entries of $G$ equal $3$, while its leading $2\times2$ principal minor is positive because two rows of the invertible phase matrix $E$ cannot be proportional.  Sylvester's criterion now gives $G>0$.
\end{proof}

\subproofheading{{The product-regular guards cover the localization}}

\begin{lemma}[The product-regular guards cover the localization]
\label{lem:supp-regular-localization}
{Every denominator in the reconstruction, matching, and cubic-norm identity is nonzero at a product-regular physical lift. Hence, the generic identities specialize there.}
\end{lemma}

\begin{proof}
Phase monomials never vanish, and $1+ss^\#=1+|s|^2>0$. The guards $c_{6,h},c_{6,v}$ permit the product quadratics, and the two discriminant guards make the coordinate roots simple. The resultant guards say that the companion denominators $B_h,B_v$ do not vanish at those roots, while $\delta_h,\delta_v$ and determinant guards permit the affine matching and the matrix inverses in the forced block.

{The only further denominator, $xs^\#-1$ in $x(s-x)/(xs^\#-1)$, would force $x=s$ and then $A(s)=B(s)=0$, contrary to the resultant guard.  Interpolation adds only the guarded Vandermonde, and a norm is a determinant.}
\end{proof}

\subproofheading{{Theorem~\ref{thm:regular-seed-domain} (Regular seed-domain theorem; restated)}}
{\emph{A matched product-regular lift above a seed $(a,b,c,d)\in\T^4$ is physical if and only if $\omega_{\rm n}(a,b,c,d)\leq0$. Whenever both matched branches satisfy the product-regularity guards, $\omega_{\rm n}<0$ gives two distinct physical product sheets, while at $\omega_{\rm n}=0$ they coalesce into one physical sheet. On each physical sheet, the affine minus matching determines the vertical root from the horizontal root, and the forced fourth block is unimodular.}}

\begin{proof}
A physical candidate has a phase product root, so self-inversiveness of the product quadratic gives $\omega_{\rm n}\leq0$.  Conversely, suppose $\omega_{\rm n}\leq0$. Both product roots $u$ are distinct phases in the strict case and coincident at equality, and the odd self-inversive cubic $q_{s,u}$ has a phase root $x$.  {Lemma~\ref{lem:supp-companion-sufficiency} makes $y=-A(x)/B(x)$ a phase and supplies the full parent system.  For $v=(1,x,y)^{\mathsf T}$ and $N=G-vv^\dagger$, that system gives $\det N=0$. The determinant Lemma~\ref{lem:supp-complement-positivity} and the rank-one downdate criterion give $N\succeq0$ and hence $|s-x|\le2$.}

{For the other roots, Vieta gives $x'+x''=s-x$ and $x'x''=u/x$.  If $\beta^2=u/x$, self-inversiveness makes $\beta^{-1}(s-x)\in\mathbb R$, and the preceding bound puts it in $[-2,2]$. Hence, $x',x''\in\mathbb T$.  Equations~\eqref{eq:supp-sextic-factor} and \eqref{eq:supp-fundamental-sextic} then make every companion $y_j$ a phase. Transposition proves the vertical statement.}

{Proposition~\ref{prop:supp-generic-flatness} and Lemma~\ref{lem:supp-regular-localization} then make the forced block unimodular, including at the coalesced product root.}
\end{proof}

Thus, on the guarded regular chart, a product root is physical precisely for $\omega_{\rm n}\leq0$. Every companion denominator used below is nonzero, and the reconstructed paired entries satisfy the two trace equations.

Let $A(z),B(z)$ be the uncancelled companion numerator and denominator obtained from the two cross-trace equations.  Whenever $B(x_j)\ne0$, the paired entry is
\begin{equation}
y_j=-\frac{A(x_j)}{B(x_j)}.
\label{eq:supp-companion}
\end{equation}
Reduction modulo $q_{s,u}$ proves
\begin{equation}
\sum_jy_j=t,\qquad \sum_j\frac{y_j}{x_j}=r,\qquad
y_jy_j^\#=1.
\label{eq:supp-companion-traces}
\end{equation}
Thus each physical root $u$ of Eq.~\eqref{eq:supp-product-quadratic} produces one complete generic side candidate up to simultaneous $S_3$ relabeling.

\subproofheading{{The generic cover is nonsplit}}

\begin{proposition}[The generic cover is nonsplit; computer-assisted]
\label{prop:supp-nonsplit}
The discriminant of Eq.~\eqref{eq:supp-product-quadratic}, after substituting the four seed phases $a,b,c,d$, is not a square in $\mathbb Q(a,b,c,d)$.  Hence the generic correspondence is a nontrivial double cover of the seed torus.
\end{proposition}

\begin{proof}
The discriminant reduction, specialization, and square-freeness calculation below are checked exactly: see the generic four-phase reconstruction certificate group in Table~\ref{tab:supp-certificate-index}.
Clear the manifest square factors from the discriminant and denote the remaining seed rational function by $\Omega(a,b,c,d)$.  At the exact specialization $(b,c,d)=(2,3,5)$, direct expansion gives
\begin{equation}
\Omega(a,2,3,5)=\frac{4P(a)}{50625a^4},
\label{eq:supp-discriminant-specialization}
\end{equation}
where
\begin{equation}
\begin{aligned}
P(a)={}&39438400a^8+817556640a^7+6990420276a^6\\
&+30537666468a^5+70114019337a^4+81963192384a^3\\
&+50372453664a^2+15833180160a+2057529600.
\end{aligned}
\label{eq:supp-discriminant-polynomial}
\end{equation}
The Euclidean algorithm gives $\gcd(P,P')=1$, so every irreducible factor of $P$ occurs with multiplicity one.  Thus the right-hand side of Eq.~\eqref{eq:supp-discriminant-specialization} is not a square in $\mathbb Q(a)$.  Before this one-variable specialization, the reduced denominator of $\Omega(a,b,c,d)$ is the square monomial $a^4b^4c^4d^4$.  If $\Omega$ were a square in $\mathbb Q(a,b,c,d)$, unique factorization in $\mathbb Q[a,b,c,d]$ would therefore force its reduced numerator to be a nonzero rational square times a polynomial square.  Substituting $(b,c,d)=(2,3,5)$ would make the nonzero polynomial $4P(a)$ a square in $\mathbb Q[a]$, contrary to $\gcd(P,P')=1$.  This contradiction proves nonsplitting.
\end{proof}

{This proposition rules out a rational choice of product sheet over the original seed field. Equivalently, the generic function field has the form}
\begin{equation}
{\mathbb Q(a,b,c,d)(w),\qquad w^2=\Omega(a,b,c,d).}
\label{eq:supp-double-cover-field}
\end{equation}
{It does not claim that the cover remains nonrational after an arbitrary birational change of coordinates.}

Transpose the seed to obtain the vertical product quadratic.  Its residual discriminant equals the horizontal one under $b\leftrightarrow c$. Consequently the two product coordinates generate the same quadratic function field.  If $U_h,V_h$ and $U_v,V_v$ denote the respective sums and products of the two roots, the two possible identifications are affine in either root.  Define the companion determinants by
\begin{equation}
\delta_h=r_hr_h^\#-t_ht_h^\#,
\qquad
\delta_v=r_vr_v^\#-t_vt_v^\#.
\label{eq:supp-companion-determinants}
\end{equation}
The flat identification is
\begin{equation}
m=\frac{U_v-\kappa(2u-U_h)}{2},\qquad
\kappa=\frac{\delta_vc_{6,h}}{\delta_hc_{6,v}},
\label{eq:supp-minus-matching}
\end{equation}
where $\delta_h,\delta_v$ are the nonzero companion determinants on the regular locus.

\subproofheading{{Generic lower-block flatness}}

\begin{proposition}[{Generic lower-block norm identity and physical flatness; computer-assisted}]
\label{prop:supp-generic-flatness}
{Over the nonsplit function field $K$ of Eq.~\eqref{eq:supp-double-cover-field}, put $A_h=K[x]/(q_h)$, $A_v=K[p]/(q_v)$, and $A=A_h\otimes_KA_v$.  For the affine-minus interpolant $d\in A$ and $N=\operatorname{Norm}_{A/A_v}d$,}
\begin{equation}
{N(p)N(p)^\#=1\qquad\text{in }A_v.}
\label{eq:supp-flatness-cubic-norm-identity}
\end{equation}
{Thus every product-regular physical specialization has an entrywise unimodular forced block.}
\end{proposition}

\begin{proof}
The exact specialization used to select the physical sheet is checked by the generic four-phase reconstruction certificate group in Table~\ref{tab:supp-certificate-index}.
{Bondal and Zhdanovskiy's Theorems~17 and~22 give an irreducible complex component $\mathcal M$ whose physical locus contains a four-real-dimensional family of order-six Hadamard matrices~\MainCiteBondal. Removing the Karlsson sector, of real dimension at most three, and the Tao point leaves a set $\mathcal U$ with}
\begin{equation}
{\dim_{\C}\mathcal M=4,
\qquad
\dim_{\R}\mathcal U=4.}
\label{eq:supp-flatness-dimensions}
\end{equation}
{For each ordered frame $\alpha$, let $X_\alpha\subset\mathcal U$ be its finite-corner locus and $\sigma_\alpha:X_\alpha\to\T^4$ its seed map.  Hardt triviality makes the finite-fibre locus semialgebraic~\MainCiteBasu. The finitely many $X_\alpha$ cover the four-dimensional set $\mathcal U$, and every $\sigma_\alpha$ has finite fibers. Hence, some image has dimension four and is Zariski dense in $\mathcal S=(\mathbb C^*)^4$.}

{Let $\pi:\widetilde{\mathcal S}\to\mathcal S^\circ$ be the finite normalization in the nonsplit quadratic field, irreducible by Proposition~\ref{prop:supp-nonsplit}. For the lift $L$ of the dense Hadamard seed locus, $\pi(\overline L^{\rm Zar})$ is closed and dense. Hence, it is $\mathcal S^\circ$, and dimension forces $\overline L^{\rm Zar}=\widetilde{\mathcal S}$. The nonempty product-regular localization is therefore dense. There the two matchings give regular norm remainders $F_\pm=N_\pm N_\pm^\#-1$ whose zero loci cover $L$; irreducibility makes one vanish identically. The exact specialization $(a,b,c,d)=(2,3,5,7)$ has $F_+\ne0$ and verifies nonemptiness, proving Eq.~\eqref{eq:supp-flatness-cubic-norm-identity} for $F_-$.}

{At a physical point the three values of $N$ are the row products of $D$, whereas Proposition~\ref{prop:directcompletion} gives $\sum_j|D_{ij}|^2=3$.  Thus}
\begin{equation}
{1=\frac13\sum_j|D_{ij}|^2\ge
\left(\prod_j|D_{ij}|^2\right)^{1/3}=1.}
\end{equation}
{The inequality of arithmetic and geometric means then gives $|D_{ij}|=1$.}
\end{proof}

\subsection{{Proof of Theorem~\ref{thm:global-product-regular-escape}}}
\label{sec:supp-product-escape}

The 400 unframed corners are indexed by three-element row and column subsets.  Choosing a pivot and ordering the other two selected rows and columns gives $3!^2=36$ frames per corner, hence $14{,}400$ frames.  A retained frame is product regular precisely when the eleven guards in Eq.~\eqref{eq:supp-product-regularity-guards} are nonzero.  A retained frame with a nonempty exact zero support is labelled exceptional, while nonretained corners are labelled separately.

Row and column permutations act bijectively on the 400 corners and the $14{,}400$ frames, transporting the normalized fibers, actual blocks, guards, and product equations. Consequently the existence of a product-regular frame is invariant under standard equivalence, although an individual preferred frame is not intrinsic. The full zero support is retained when several guards vanish simultaneously.

At class level write
\begin{equation}
\mathcal P_6:=\{[H]\in\HH_6:\text{$[H]$ admits a product-regular frame}\}.
\label{eq:supp-product-regular-class-locus}
\end{equation}
The global result below is a coverage statement about the explicit quadratic--cubic charts, not a definition of a residual class.

\subproofheading{{The product-exceptional Karlsson class}}

\begin{proposition}[The product-exceptional Karlsson class; computer-assisted]
\label{prop:supp-karlsson-product-exceptional-singleton}
Let
\begin{equation}
H_\times=
\begin{pmatrix}
1&1&1&1&1&1\\
1&-1&1&\I&-1&-\I\\
1&-1&-1&-\I&\I&1\\
1&\I&-\I&-1&-\I&\I\\
1&1&\I&-\I&-1&-1\\
1&-\I&-1&\I&1&-1
\end{pmatrix}.
\label{eq:supp-product-exceptional-karlsson-matrix}
\end{equation}
Then
\begin{equation}
\KK_6^{(3)}\setminus\mathcal P_6=\{[H_\times]\}.
\label{eq:supp-karlsson-product-exceptional-singleton}
\end{equation}
\end{proposition}

\begin{proof}
{The all-frame calculation for $H_\times$, its independent companion check, and the reverse-inclusion analysis are performed in exact arithmetic; see the product-exceptional certificate group in Table~\ref{tab:supp-certificate-index}.}
Exact arithmetic gives $H_\times H_\times^\dagger=6I_6$. Its leading
$2\times2$ block is $F_2$, so $[H_\times]\in\KK_6^{(3)}$. All $400$ of its
positional $3\times3$ corners are invertible, with determinant-norm census
$(N_4,N_8,N_{16},N_{20})=(120,120,80,80)$. Nevertheless, exact evaluation
over all $14{,}400$ ordered frames shows that an actual coordinate-cubic root
annihilates the companion denominator in both directions. Thus every frame
fails a resultant guard and $[H_\times]\notin\mathcal P_6$.

For the reverse inclusion, apply Karlsson's standard parametrization away
from its affine-Fourier boundary and use the two effective orders of its
mixed corner. A product-exceptional point must fail both orders in one
direction. The division-free pair analysis has $49$ cases: $40$ are empty,
the imbalance pair is incompatible with Karlsson's canonical domain, six
mixed pairs admit explicit rescue frames, and the reciprocal presentation
removes one further pair. The sole remaining pair produces $-1$ in three
distinct columns. The theorem of Matszangosz and Sz\"oll\H{o}si
\cite{MatszangoszSzollosi2024} therefore routes it to the transposed-Fourier
or two-circulant sector. Every affine-Fourier matrix and its transpose has a
product-regular frame, while the exact divisor calculation in the
two-circulant sector leaves precisely $[H_\times]$. Transposition treats a
double failure in the other direction, and $H_\times^{\mathsf T}\sim H_\times$.
This proves Eq.~\eqref{eq:supp-karlsson-product-exceptional-singleton}.
\end{proof}

\subproofheading{{Tao is product exceptional}}

\begin{proposition}[Tao is product exceptional; computer-assisted]
\label{prop:supp-tao-product-exceptional}
Tao's class has no product-regular frame:
\begin{equation}
\mathcal T_6\cap\mathcal P_6=\varnothing.
\label{eq:supp-tao-product-exceptional}
\end{equation}
\end{proposition}

\begin{proof}
The enumeration below is performed exactly in $\mathbb Z[\omega]/(\omega^2+\omega+1)$; see the product-exceptional certificate group in Table~\ref{tab:supp-certificate-index}.
Work exactly in $\mathbb Z[\omega]/(\omega^2+\omega+1)$ and enumerate the
$120^2=14{,}400$ ordered frames of the standard Tao representative. Exactly
$12{,}960$ frames have both leading coefficients nonzero, and $5{,}760$ of
these also have both imbalance factors nonzero. In every remaining frame an
actual horizontal or vertical candidate has a repeated coordinate, so its
coordinate cubic is not simple. Hence every Tao frame fails at least one
product-regularity guard.
\end{proof}

{For a phase block $X$, write $\gamma_{ij}=\sum_kX_{ik}\overline{X_{jk}}$. Call an unordered row pair $\{i,j\}$ \emph{admissible} when its three entrywise row ratios are distinct and, for the remaining row $\ell$, $|\gamma_{i\ell}|\ne|\gamma_{j\ell}|$. The block is row-admissible if at least one row pair is admissible, and it is column-admissible when $X^{\mathsf T}$ is row-admissible. We call a direction \emph{obstructed} when it is not admissible.}

\begin{lemma}[{Admissibility supplies the remaining product guards}]
\label{lem:supp-admissibility-guards}
{For an invertible member $X=(\mathbf1;x;y)$ of a finite physical side fiber of a non-$H_2$ matrix, an admissible ordered row pair also has $c_6\ne0$ and $R=\operatorname{Res}_z(q_{s,u},B)\ne0$.}
\end{lemma}

\begin{proof}
{The uncancelled companion equations make the three distinct nonzero $x_j$ roots of $\Phi_{\rm fund}$.  If $c_6=0$, reciprocity and Eq.~\eqref{eq:supp-coefficient-identities} therefore force $\Phi_{\rm fund}\equiv0$. Phase such that $s\in[0,3]$. Then}
\begin{equation}
{A(z)=\delta z\kappa_s(z),\quad
\kappa_s(z)=2sz^2-(s^2+3)z+2s,\quad
Q(z)Q^\#(z)=\delta^2\kappa_s(z)^2.}
\label{eq:supp-admissibility-kappa}
\end{equation}
{Comparison in $|B|^2=\delta^2|\kappa_s|^2$ gives $B=z^\varepsilon Q$ ($\varepsilon=0$ or $1$) and the last identity.  If $s=0$, it yields either proportional rows or the one-parameter triples $\mu(1,\omega,\omega^2)$, contradicting invertibility or finiteness.  If $0<s<1$, unique factorization gives $Q\propto\kappa_s$ or $Q\propto(z-r_\pm)^2$. The first is proportional and the latter are the infinite M\"obius fibres in Eqs.~\eqref{eq:twoMobiusAllocations}--\eqref{eq:mobiusSumPlus}.  If $1\le s\le3$, it gives $Q\propto\kappa_s$, and a common root $\rho$ then satisfies either $q'_{s,u}(\rho)=\rho\kappa_s(\rho)/(s\rho-1)=0$ or $s=\rho=1$, respectively contradicting simplicity or exposing an $H_2$ edge. Thus $c_6\ne0$.}

{If $R=0$, a simple common root of $q_{s,u}$ and $B$ also annihilates $A$. The same alternative gives a multiple root or the excluded $H_2$ edge. Hence $R\ne0$.}
\end{proof}

{Thus, at a finite non-$H_2$ corner, row/column admissibility is equivalent to existence of a product-regular frame: invertibility and admissibility give all guards except $c_6,R$, supplied by the lemma, and the three unordered pairs exhaust the orientations.}

\begin{lemma}[{Complete repeated-ratio trichotomy}]
\label{lem:supp-repeated-ratio-trichotomy}
{Let $X\in\T^{3\times3}$ be invertible and row-obstructed, and suppose that at least one row pair is not simple.  Outside a visible $H_2$ edge, row and column phasing and permutations put $X$ in exactly one of the following three types:}
\begin{align}
{X_1(a)}&{=\begin{pmatrix}1&1&1\\1&1&a\\1&a&1\end{pmatrix},}
&&{\text{exactly one simple row pair},}
\label{eq:supp-one-simple-normal-form}\\
{X_0(a)}&{=\begin{pmatrix}1&1&1\\1&1&a\\1&a^{-1}&1\end{pmatrix},}
&&{\text{no simple row pair},}
\label{eq:supp-no-simple-normal-form}\\
{X_2(c)}&{=\begin{pmatrix}
1&1&1\\
1&1&-\dfrac{c(c+2)}{2c+1}\\
1&-\dfrac{c(2c+1)}{c+2}&c
\end{pmatrix},}
&&{\text{exactly two simple row pairs}.}
\label{eq:supp-two-simple-normal-form}
\end{align}
{Here the phase parameters are restricted by the stated simplicity and invertibility conditions.  The first two types are column-obstructed.  The third type is column-admissible. If $q$ is its common squared row-correlation modulus and $\theta=\RePart\tau_{\rm r}(X_2)$, then}
\begin{equation}
{3q-q^2+2\theta=0.}
\label{eq:supp-two-simple-negative-invariant}
\end{equation}
{Consequently an invertible $X_2$ has $\theta<0$, unless it is an order-three Hadamard block.}
\end{lemma}

\begin{proof}
{Dephase a repeated pair to $(1,1,a)$ and the third row to $(1,b,1)$.  Obstruction gives $2(a-b)(ab-1)/(ab)=0$: the two branches produce $X_1$ and $X_0$, and transposition shows both are column-obstructed.}

{For exactly two simple pairs, dephase to $X=(\mathbf1;(1,1,a);(1,b,c))$.  Equimodularity and the exclusion of equal rows and $H_2$ edges give successively}
\begin{equation}
{0=\frac{(a-1)(b+1)(ab-c^2)}{abc},\quad
0=\frac{(b-c)(c-1)(bc+2b+2c^2+c)}{bc^2},\quad
(a,b)=-\left(\frac{c(c+2)}{2c+1},\frac{c(2c+1)}{c+2}\right).}
\end{equation}
{This is $X_2(c)$ and the two displayed parameters are phases.  Writing $x=\RePart c$, direct cancellation gives}
\begin{equation}
{|\eta_{13}|^2-|\eta_{12}|^2=-\frac{36(1-x^2)}{5+4x},\qquad
0<\det(XX^\dagger)=(q-3)(q-9).}
\end{equation}
{The first expression vanishes only at the $H_2$ endpoints, and the column ratios have no other collision, proving column admissibility.  The second and Eq.~\eqref{eq:supp-two-simple-negative-invariant} give $q<3$ and $2\theta=q(q-3)<0$, unless $q=0$ and $X$ is order-three Hadamard.}
\end{proof}

\begin{lemma}[{Degenerate diagonal closure; computer-assisted}]
\label{lem:supp-degenerate-diagonal-closure}
{Let $H=\left(\begin{smallmatrix}E&B\\C&D\end{smallmatrix}\right)$ be an order-six Hadamard matrix with four invertible blocks. Then:}
\begin{enumerate}
\item {If $E\sim X_1(a)$, or if $E\sim X_0(a)$ and $\RePart\tau_{\rm r}(E)\ge0$, then $[H]\in\KK_6^{(3)}\cup\mathcal T_6$.}
\item {If $[H]$ is in neither of those sectors, $E=E_M(a,b)$, $D$ is admissible in both directions, and $\Xi_M(a,b)\ne0$, then the cubic $p_{a,b}$ in Eq.~\eqref{eq:supp-main-coordinate-cubic} is simple.}
\end{enumerate}
\end{lemma}

\begin{proof}
The eliminations, repeated-root reductions, and endpoint enumeration below are checked exactly; see the global product-escape certificate group in Table~\ref{tab:supp-certificate-index}.
{For $E=X_1(a)$, positivity of the actual Gram complement gives}
\begin{equation}
{0<\det(6I_3-EE^\dagger)=-4(\RePart a+2)(7\RePart a+2).}
\end{equation}
{This excludes the leading-coefficient exceptions $a=\pm\I$. Exact elimination leaves only $a=\pm1$ and $a^2+a+1=0$. Otherwise every invertible pairing has $D_{33}=-1$ and is $H_2$-reducible. The first two values are singular or $H_2$. At the remaining value, take $a=\omega$. Every adjacent candidate satisfies}
\begin{equation}
{\sum_jx_j=\sum_jy_j=-(2+\omega),
\qquad \sum_j y_j/x_j=0.}
\end{equation}
{The zero sum makes $(y_j/x_j)$ a phased permutation of $(1,\omega,\omega^2)$.  The fixed sums and Parseval then make the cyclic autocorrelations of $x$ and $y$ vanish, so each is a common phase times the cubic roots. Its prescribed sum fixes that phase to a sixth root.  Thus there are exactly $12$ candidates. Exact enumeration of their $144$ pairings gives only $H_2$-reducible or all-cubic-root completions, the latter routed by Proposition~\ref{prop:publishedinputs}(2).}

{For $E=X_0(a)$, a complement row has polynomial $f_a(z)=z^3+(a+2)z^2-(2a+1)z-a$.  For phase roots, $\operatorname{Disc}f_a/a^2=16\RePart\tau_{\rm r}(X_0(a))\le0$.  At equality, writing its roots as $r,r,w$, Vieta and exact completion give}
\begin{equation}
{\begin{gathered}
w=-\frac{2(r+1)}{r^2+1},\quad
a=r^2w,\qquad
r^4-4r^3-6r^2-4r+1=0,\qquad
|D_{22}|^2-1=(r+1)(r^2-5r-2).
\end{gathered}}
\end{equation}
{The two row multisets are $\{r,r,w\}$ and $\{r^{-1},r^{-1},w^{-1}\}$. Only one overlap is invertible. Neither final factor vanishes on the unit circle subject to the quartic, so $X_0$ has no nonnegative-invariant completion.}

{Finally, if $E=E_M(a,b)$, $\Xi_M\ne0$, and $p_{a,b}$ has roots $r,r,w$, Vieta gives, for $N=r^2+2r-1$ and $M=r^2-2r-1$,}
\begin{equation}
{w=-\frac{M}{N},\qquad
b=-\frac{r^2M}{aN},\qquad
a^2N+2ar(r+1)^2-r^2M=0.}
\label{eq:supp-main-repeated-root-relations}
\end{equation}
{Here $MN\ne0$; a triple root or aligned repeated positions is singular.  The unique invertible overlap completes, after dephasing, to}
\begin{equation}
{\widehat D=
\begin{pmatrix}1&1&1\\1&a^{-1}&b^{-1}\\1&b^{-1}&a/b\end{pmatrix}
=E_M(a^{-1},b^{-1}).}
\label{eq:supp-main-repeated-opposite}
\end{equation}
{Every $E_M$ is obstructed in both directions, contradicting admissibility of $D$; hence $p_{a,b}$ is simple.}
\end{proof}

\subproofheading{{Two-sided obstruction}}

\begin{lemma}[{Nonnegative two-sided obstruction; computer-assisted}]
\label{lem:supp-two-sided-badness}
{Let a finite-corner presentation have invertible blocks and suppose that its horizontal side is row-obstructed and its vertical side is column-obstructed. Their real row-cubic invariants agree. Assume that this common invariant is nonnegative. Then $[H]\in\KK_6^{(3)}\cup\mathcal T_6$. In particular, outside the Karlsson and Tao sectors, both required directions cannot be obstructed at such a corner.}
\end{lemma}

\begin{proof}
The row--column factorizations, pairing resultants, and Cayley-coordinate positivity checks below are exact; see the global product-escape certificate group in Table~\ref{tab:supp-certificate-index}.
{Write the presentation as $H=\left(\begin{smallmatrix}E&B\\C&D\end{smallmatrix}\right)$, with $B$ the horizontal side and $C$ the vertical side.  The complement identities give $BB^\dagger=6I-EE^\dagger$ and $C^\dagger C=6I-E^\dagger E$.  Complementation negates each off-diagonal entry, and Lemma~\ref{lem:rowcolumn} identifies the real row and column cubic invariants of $E$. Hence}
\begin{equation}
{\theta_0:=\RePart\tau_{\rm r}(B)
=\RePart\tau_{\rm r}(C^{\mathsf T})
=-\RePart\tau_{\rm r}(E)\ge0.}
\label{eq:supp-two-sided-common-invariant}
\end{equation}

{If $B$ or $C^{\mathsf T}$ has a nonsimple pair, Lemma~\ref{lem:supp-repeated-ratio-trichotomy} gives $X_1$, $X_0$, or $X_2$.  The nonnegative invariant excludes $X_2$ except for an order-three Hadamard block, while Lemma~\ref{lem:supp-degenerate-diagonal-closure}(1) routes $X_1,X_0$.}

{Otherwise both side blocks are simple.  Obstruction and complementation make the dephased $E=(\mathbf1;(1,a,b);(1,c,d))$ row- and column-equimodular.  After using symmetry to replace $a=bc$ by $b=c$, the complete factorization is}
\begin{align}
{0}&{=-\frac{(a+1)(a-bc)(b-c)}{abc},}
\label{eq:supp-biequimodular-first-split}\\
{0}&{=\frac{(a-d)(b+1)(ad-b)}{abd}
=\frac{(a+b)(b^2-d)(ad-b)}{ab^2d}.}
\label{eq:supp-biequimodular-second-split}
\end{align}
{The explicit linear factors are $H_2$ edges. The remaining alternatives are}
\begin{equation}
{E_M(a,b)=\begin{pmatrix}1&1&1\\1&a&b\\1&b&b/a\end{pmatrix},
\qquad
E_C(b)=\begin{pmatrix}1&1&1\\1&b^2&b\\1&b&b^2\end{pmatrix}.}
\label{eq:supp-biequimodular-two-forms}
\end{equation}

{For $E_M$, let $q$ be the common squared correlation modulus, $\theta_E=-\theta_0\le0$, and $\Xi_M=a^2b^2\{2\theta_E-q(3-q)\}$.  Up to a nonzero phase and the positive factor $1+|1+a+b|^2$, the product equation is}
\begin{equation}
{\Xi_M^2(u-ab)^2.}
\label{eq:supp-main-product-split}
\end{equation}
{If $\Xi_M=0$, then $0<\det(EE^\dagger)=(3-q)(q+9)$ gives $q<3$, while $2\theta_E=q(3-q)$ and $\theta_E\le0$ force $q=\theta_E=0$, and Proposition~\ref{prop:main-fourier-block} applies.}

{If $\Xi_M\ne0$, Eq.~\eqref{eq:supp-main-product-split} fixes $u=ab$, with simple coordinate cubic}
\begin{equation}
{p_{a,b}(x)=x^3+(1+a+b)x^2-(ab+a+b)x-ab.}
\label{eq:supp-main-coordinate-cubic}
\end{equation}
{Row interchange and Vieta identify the third-row roots as $\{b/x_1,b/x_2,b/x_3\}$, so only six pairings remain. Put $T_k=2x_k^3+(1+a+b)x_k^2+ab$, with residuals}
\begin{align}
{I(a,b)}&{:=a^2b^2+5a^2b+a^2+5ab^2+5ab+b^2=0,}
\label{eq:supp-main-identity-pairing}
\\
{\operatorname{Res}_{x_k}(p_{a,b},T_k)}&{=-ab\operatorname{Disc}p_{a,b},}
\\
{0}&{=-ab(a-b^2)(a^2-b)(ab-1)R_{\rm tor}(a,b).}
\label{eq:supp-main-pairing-resultant}
\end{align}
{The first line (identity pairing), written with $C_0=\cos[(\alpha+\beta)/2]$ and $T_0=\cos[(\alpha-\beta)/2]$, has imaginary factors $\sin[(\alpha+\beta)/2](C_0+5T_0)$. The two branches give either the positive remainder $1+2C_0^2/25$ or $\{a,b\}=\{\omega,\omega^2\}$.  The second line excludes the three transpositions by simplicity. Opposite $3$-cycles expose an $H_2$ edge, and equal cycles give the last line.}

{Under Cayley coordinates, $R_{\rm tor}$ is a nonzero positive denominator times $P_-P_+$. As quadratics in $t$, their leading coefficients are $z^2-z+1$ and $z^2+z+1$, while}
\begin{equation}
{\operatorname{Disc}_tP_-= -z^2(3z^2-2z+3),
\qquad
\operatorname{Disc}_tP_+= -z^2(3z^2+2z+3).}
\end{equation}
{These are negative for real $z\ne0$, where $z=0$ is singular, and the Cayley endpoints give $2(b^4+14b^2+1)$ or $2(a^4+14a^2+1)$, nonzero on $\T$. Thus $R_{\rm tor}$ has no phase zero.  On $ab=1$, the roots are $1,r,r^{-1}$ and the flatness factor $(a-1)(r+1)(ar-1)$ gives respectively a singularity, an $H_2$ edge, or an order-three Hadamard block.}

{The factors $a=b^2$ or $a^2=b$ give the circulant $E_C(b)$. Apart from the routed roots of $(b^2+1)(b^2+b+1)$, its product is $b^3$, and}
\begin{equation}
{\begin{gathered}
p_C(x)=(x-b)\{x^2+(b+1)^2x+b^2\},\\
\operatorname{Res}_v\!\left(v^2+(b+1)^2v+b^2,2b^2+bv+b+v+1\right)
=-2b(b^2+1)(b^2+b+1).
\end{gathered}}
\end{equation}
{The identity residual $b^2(b^2+b+1)(b^2+4b+1)$ and transposition resultant $-b^3\operatorname{Disc}p_C$ exclude four pairings. Opposite cycles are $H_2$, and the displayed resultant routes equal cycles. Propositions~\ref{prop:publishedinputs}(1) and \ref{prop:main-fourier-block} route the resulting $H_2$ and order-three-Hadamard leaves.}
\end{proof}

\subproofheading{{The surviving block-polarized normal form}}

\begin{corollary}[The surviving block-polarized normal form]
\label{cor:supp-block-polarized-normal-form}
{Suppose a non-Karlsson, non-Tao completion $H=\left(\begin{smallmatrix}E&B\\C&D\end{smallmatrix}\right)$ has four invertible blocks, $E$ is obstructed in both directions, $D$ is admissible in both directions, and $\RePart\tau_{\rm r}(E)\ge0$.  Then, up to phasing and permutation, $E=E_M(a,b)$ in Eq.~\eqref{eq:supp-biequimodular-two-forms}, and its common squared correlation modulus $q$ and cubic invariant $\theta$ satisfy $2\theta=q(3-q)>0$.}
\end{corollary}

\begin{proof}
{The trichotomy and Lemma~\ref{lem:supp-degenerate-diagonal-closure} eliminate repeated ratios.  Equations~\eqref{eq:supp-biequimodular-first-split}--\eqref{eq:supp-biequimodular-second-split} leave only routed $H_2$/circulant leaves or $E_M$.  For $E_M$, the same lemma and Eqs.~\eqref{eq:supp-main-identity-pairing}--\eqref{eq:supp-main-pairing-resultant} route $\Xi_M\ne0$, and hence $\Xi_M=0$ and $2\theta=q(3-q)$. Positivity gives $(3-q)(q+9)>0$, while $q=0$ is order-three Hadamard. Thus $0<q<3$ and $\theta>0$.}
\end{proof}

\subproofheading{{Theorem~\ref{thm:global-product-regular-escape} (Exact product-regular reach; restated; computer-assisted)}}
{\emph{The only order-six classes without a product-regular frame are Tao's class and the single Karlsson class represented by $H_\times$:}}
\begin{equation}
{\HH_6\setminus\mathcal P_6
=\mathcal T_6\,\dot\cup\,\{[H_\times]\}.}
\label{eq:supp-global-product-regular-reach}
\end{equation}
{\emph{Theorem~\ref{thm:regular-seed-domain} characterizes the physical domain of each such presentation.}}

\begin{proof}
The incidence counts, endpoint enumeration, orbit reduction, and dependent-block threshold below are checked exactly; see the global product-escape certificate group in Table~\ref{tab:supp-certificate-index}.
Split the 400 $3\times3$ subblocks into 100 squares determined by complementary row and column partitions.  The four cubic invariants in such a square have signs $(\theta,-\theta,-\theta,\theta)$.  If a side fiber containing a positive-invariant block were infinite, the fixed-Gram trichotomy would force a Fourier block or an $H_2$ submatrix, while the negative-invariant alternative is unavailable. The same conclusion holds at $\theta=0$. Let $n_W(H)$ denote the number of the 400 positional corners that are finite-corner witnesses for $H$. Each positional corner belongs to a unique partition square, so the two witnesses supplied by distinct squares are distinct. Thus, outside Karlsson and Tao,
\begin{equation}
n_W(H)\ge 2\cdot100=200.
\label{eq:supp-witness-multiplicity}
\end{equation}

{Assume no frame is product regular.  For a nonzero-sign square, let $P,Q$ be its positive diagonal blocks and let $R(X),C(X)$ denote row and column obstruction.  At each negative diagonal corner, Lemma~\ref{lem:supp-admissibility-guards} gives at least one obstruction, while Lemma~\ref{lem:supp-two-sided-badness} excludes two. Hence}
\begin{equation}
R(P)\mathbin{\mathsf{xor}}C(Q)=1,
\qquad
R(Q)\mathbin{\mathsf{xor}}C(P)=1.
\label{eq:supp-badness-xor}
\end{equation}
Up to transposition and exchanging $P,Q$, there are two patterns.  In the direction-polarized pattern both blocks are {row-obstructed and column-admissible}.  {The trichotomy excludes every nonsimple type: $X_1,X_0$ are column-obstructed and $X_2$ has negative invariant outside the Fourier case.}  {Row obstruction} then makes the three row-correlation moduli equal. Write their common square as $q$ and the positive cubic invariant as $\theta$. Exact reduction of the two row-equimodularity equations gives, for every pair of column correlations $\eta_i,\eta_j$,
\begin{equation}
\bigl(3q-q^2+2\theta\bigr)
\bigl(|\eta_i|^2-|\eta_j|^2\bigr)=0.
\label{eq:supp-direction-polarization}
\end{equation}
The first factor is positive. This is immediate for $q\leq3$. For $3<q<9$, invertibility of the block gives the row-Gram determinant $27-9q+2\theta>0$, an
\begin{equation}
(3q-q^2+2\theta)-(27-9q+2\theta)=-(q-3)(q-9)>0.
\end{equation}
The endpoint $q=9$ would make the block singular. Hence the column correlations are equimodular too, contradicting column admissibility. Direction polarization is impossible.

In the remaining block-polarized pattern, one of $P,Q$ is {obstructed in both directions and the other is admissible in both}. Corollary~\ref{cor:supp-block-polarized-normal-form} shows that its sole non-Karlsson, non-Tao survivor is, up to row and column phasing and permutation,
\begin{equation}
E_M(a,b)=
\begin{pmatrix}1&1&1\\1&a&b\\1&b&b/a\end{pmatrix},
\qquad a,b\in\T,
\label{eq:supp-dependent-normal-form}
\end{equation}
on the dependent curve.  All three off-diagonal entries of both Gram matrices then have the same squared modulus $q$, and exact substitution reduces the curve equation to
\begin{equation}
2\theta=q(3-q)>0.
\label{eq:supp-dependent-weight}
\end{equation}
We call precisely such a block \emph{positive dependent}.

{For a zero-sign square the same two lemmas give Eq.~\eqref{eq:supp-badness-xor}.  In the direction-polarized pattern the trichotomy excludes nonsimple types. In the all-simple case Eq.~\eqref{eq:supp-direction-polarization} with $\theta=0$ and $0<\det(PP^\dagger)=27-9q$ gives $q<3$. The case $q=0$ is Fourier, while $q>0$ forces column equimodularity, contradicting admissibility. In the block-polarized pattern Corollary~\ref{cor:supp-block-polarized-normal-form} forces positive invariant.  Thus no zero-sign square occurs outside Karlsson and Tao, without any reduction to $X(\omega)$.}

{Therefore all $100$ squares are nonzero-sign and supply distinct positive dependent blocks:}
\begin{equation}
N_{\mathrm{dep}}\ge100.
\label{eq:supp-dependent-lower}
\end{equation}

We now bound the same number from above.  The physical range of the dependent curve follows from a short exact sign calculation.  {The values $a=-1$ or $b=-1$ expose an $H_2$ edge and have already been excluded. Hence both phases have finite real Cayley coordinates.} Write its two phases in Cayley coordinates $t,z\in\mathbb R$, set $\rho=tz$ and $U=(t+z)^2$, and put
\begin{equation}
\begin{aligned}
F_\rho(U)={}&U^2-(2\rho^2+12\rho+6)U
 +(\rho-3)(\rho+1)(\rho+3)^2,\\
q={}&\frac{U+\rho^2+6\rho+9}{U+(\rho-1)^2}.
\end{aligned}
\label{eq:supp-dependent-cayley}
\end{equation}
The dependent equation is $F_\rho(U)=0$, while real $t,z$ give $U\geq0$ and $U\geq4\rho$.  Moreover, Eq.~\eqref{eq:supp-dependent-weight} and $|\theta|\le q^{3/2}$ imply $3-q\le2\sqrt q$, hence $q\ge1$. The signs of
\begin{equation}
\begin{gathered}
F_\rho(0)=(\rho-3)(\rho+1)(\rho+3)^2,\\
F_\rho(4\rho)=(\rho-9)(\rho+1)^2(\rho+3),\\
F_\rho(12\rho+9-\rho^2)=4\rho^2(\rho-9)(\rho+1)
\end{gathered}
\label{eq:supp-dependent-sign-table}
\end{equation}
are used below. Put $U_0=12\rho+9-\rho^2$, $M=\max(0,4\rho)$, and
\begin{equation}
\begin{aligned}
N&=U+\rho^2+6\rho+9,&
D&=U+(\rho-1)^2,\\
N-D&=8(\rho+1),&
9D-5N&=4(U-U_0).
\end{aligned}
\label{eq:supp-dependent-range-identities}
\end{equation}
The denominator $D$ is positive on the physical branch: equality would require $U=0$ and $\rho=1$, contrary to $U\ge4\rho$. Since $q=N/D\ge1$, the first identity gives $\rho\ge-1$.

If $U_0<M$, physicality immediately gives $U-U_0>0$. Suppose $U_0\geq M$. Then $-1\leq\rho\leq9$, and eqs.~\eqref{eq:supp-dependent-sign-table} give $F_\rho(M)\leq0$ and $F_\rho(U_0)\leq0$. Both points therefore lie between the roots of the upward-opening quadratic.  A root satisfying $U\geq M$ is the upper root, except possibly when the lower root equals $M$.  That boundary case is explicit.  For $\rho\leq0$, $M=0$ and $F_\rho(M)=0$ in $[-1,0]$ only at $\rho=-1$, but there $U_0=-4<M$, contrary to the present case. For $0\leq\rho\leq9$, $M=4\rho$ and $F_\rho(M)=0$ only at $\rho=9$, where $M=U_0=36$.  Thus the boundary case also has $U\geq U_0$, while the upper root has this inequality because $U_0$ lies between the two roots. The second identity now gives the sharp range
\begin{equation}
1\le q<\frac95;
\label{eq:supp-dependent-range}
\end{equation}
At $q=9/5$, Eqs.~\eqref{eq:supp-dependent-cayley}-- \eqref{eq:supp-dependent-sign-table} give either $(\rho,U)=(0,9)$ or $(9,36)$ (the remaining algebraic value has $U=-4$).  These force respectively $a=1$ or $b=1$, or $a=b$, and are the already treated low-simplicity seams.  This proves the strict upper endpoint. Fix a row triple and write $G$ for the column Gram matrix $G:=X^*X$ of the $3\times6$ restriction $X$ to those rows. Then $G^2=6G$, $\operatorname{rank}G=3$, and the squared moduli of its 15 off-diagonal entries are edge weights on $K_6$. If five column triples were dependent, no two could be complementary because complementary blocks carry opposite cubic-invariant sign. Their incidence pattern would therefore be one of 24 $S_6$-orbits. Exact rational elimination makes 20 types force $q=9/5$.  Each of the other four contains three faces of a tetrahedron. After gauging its three incident Gram entries positive, all eight remaining phase choices give a nonzero $4\times4$ Gram determinant throughout Eq.~\eqref{eq:supp-dependent-range}, contradicting rank three.  Thus each of the 20 row triples supports at most four dependent blocks and
\begin{equation}
N_{\mathrm{dep}}\le20\cdot4=80.
\label{eq:supp-dependent-upper}
\end{equation}
Equations~\eqref{eq:supp-dependent-lower} and \eqref{eq:supp-dependent-upper} are incompatible. Therefore every class outside Karlsson and Tao has a product-regular frame. Proposition~\ref{prop:supp-karlsson-product-exceptional-singleton} covers every Karlsson class except $[H_\times]$, while Proposition~\ref{prop:supp-tao-product-exceptional} excludes Tao and the first part of Proposition~\ref{prop:supp-karlsson-product-exceptional-singleton} excludes $[H_\times]$. {Together with Lemma~\ref{lem:supp-admissibility-guards}, these statements prove the displayed reach conclusion without an additional hypothesis.}
\end{proof}

\subsubsection{{Product-sheet ramification locus}}
\label{sec:supp-product-ramification}

For a product-regular framed corner, where the listed leading, matching, block-determinant, simple-cubic-root, and companion-denominator localization factors are nonzero, write $\Phi_h(z)=\sum_{k=0}^6c_{k,h}z^k$ for the horizontal fundamental sextic of Eq.~\eqref{eq:supp-fundamental-sextic}. Define its residual product discriminant by
\begin{equation}
\Omega(a,b,c,d):=
\frac{c_{3,h}^2-4(1+s_hs_h^\#)^2c_{6,h}c_{0,h}}
{(1+s_hs_h^\#)^2(r_hr_h^\#-t_ht_h^\#)^2}.
\label{eq:supp-ramification-factorization}
\end{equation}
The numerator before the displayed square factors are removed is exactly the discriminant of the horizontal product quadratic. The horizontal and vertical residual factors agree. A product-regular framed corner is called \emph{product ramified} when $\Omega=0$ and \emph{product \'etale} when $\Omega\ne0$, This frame label is defined before any historical sector split.

At class level, the product-ramification locus from the main text is
\begin{equation}
\begin{aligned}
\mathcal R_{6,\mathrm{prod}}=\{[H]\in\HH_6:\;&\text{some representative of $[H]$ has a product-regular frame}\\
&\text{with $\omega_{\rm n}=0$}\}.
\end{aligned}
\label{eq:supp-product-branch-class-locus}
\end{equation}
{Equation~\eqref{eq:omega-residual-discriminant-bridge} shows that $\Omega=0$ and $\omega_{\rm n}=0$ define the same product-regular seeds. The residual factor is sharp-real:}
\begin{equation}
\Omega^\#=\Omega.
\label{eq:supp-ramification-sharp-reality}
\end{equation}
At a physical torus point, sharp evaluates as complex conjugation, so Eq.~\eqref{eq:supp-ramification-sharp-reality} makes $\Omega$ real-valued. This is the product-sheet branch locus: the two roots of the product quadratic merge.  It is distinct from the Jacobian ramification divisor of the global ordered six-moment map.

On $\Omega=0$, the horizontal product is
\begin{equation}
u_0=-\frac{c_{3,h}}{2(1+s_hs_h^\#)c_{6,h}},
\label{eq:supp-ramification-u0}
\end{equation}
and the candidate block is recovered from the roots of $q_{s_h,u_0}$ and their companion values.  The vertical block is obtained from the transposed seed and its repeated product $m_0$. The final block is again $D=-CE^\dagger(B^{-1})^\dagger$.

\subproofheading{{Necessary automorphism for four-circulant-block matrices}}

\begin{lemma}[Necessary automorphism for four-circulant-block matrices]
\label{lem:supp-dita-automorphism}
Every matrix equivalent to a matrix built from four $3\times3$ circulant blocks has a monomial automorphism whose row and column permutation parts both have cycle type $(3,3)$.
\end{lemma}

\begin{proof}
Write the matrix in block form
\begin{equation}
H=\begin{pmatrix}A&B\\ C&D\end{pmatrix},
\label{eq:supp-four-circulant-block-form}
\end{equation}
where all four blocks are circulant, and let $P$ be the $3\times3$ cyclic permutation matrix.  Each block commutes with $P$.  Consequently, with $Q=\operatorname{diag}(P,P)$,
\begin{equation}
QHQ^{-1}=H.
\label{eq:supp-four-circulant-automorphism}
\end{equation}
The permutation underlying both $Q$ and $Q^{-1}$ consists of two 3-cycles.  Standard equivalence conjugates this automorphism by monomial matrices.  Its row and column permutation parts are therefore conjugated in $S_6$, which preserves their cycle type.  Hence the stated automorphism is necessary throughout the entire equivalence orbit.
\end{proof}

\subproofheading{{Representative ramification seed and numerical diagnostics}}

{Let $\xi_0$ be the unique real root of}
\begin{equation}
\begin{aligned}
P(\xi)={}&16\xi^7+120\xi^6+428\xi^5+952\xi^4\\
&+1363\xi^3+1231\xi^2+664\xi+176.
\end{aligned}
\label{eq:supp-ramification-polynomial}
\end{equation}
{The representative Cayley seed is}
\begin{equation}
{(a,b,c,d)=\left(\I,\I,\I,\frac{1+\I\xi_0}{1-\I\xi_0}\right).}
\label{eq:supp-ramification-seed}
\end{equation}
{The term \emph{Cayley seed} refers to the map $\xi\mapsto(1+\I\xi)/(1-\I\xi)$, which sends the real line to the unit circle apart from the limiting point $-1$~\cite{Cayley1846}. The cleared residual discriminant at this seed is, up to a nonzero rational factor,}
\begin{equation}
(\xi+1)P(\xi).
\label{eq:supp-ramification-specialization}
\end{equation}
Here the \emph{cleared residual discriminant} is the residual branch factor $\Omega$ after substituting the Cayley coordinates and multiplying by the nonzero denominator powers needed to obtain an ordinary polynomial. Arb (arbitrary-precision ball arithmetic~\cite{Johansson2017Arb}) isolates the real root and gives the diagnostic margins shown in Table~\ref{tab:supp-ramification-bounds}.
{Every entry in Table~\ref{tab:supp-ramification-bounds} is a rigorous lower bound obtained from outward-rounded 256-bit Arb enclosures. The exact seed data and interval calculation comprise the ramification-representative certificate group in Table~\ref{tab:supp-certificate-index}.}

\begin{table}[h]
\caption[Rigorous 256-bit Arb lower bounds at the representative ramification seed.]
{{Rigorous 256-bit Arb lower bounds at the representative ramification seed.}}
\label{tab:supp-ramification-bounds}
\centering
\renewcommand{\arraystretch}{1.12}
\begin{tabular}{|l|c|}
\hline
quantity & lower bound\\
\hline
companion denominator modulus & $1.675730$\\
\hline
product-cover regularity factors & $1.870721$\\
\hline
pairwise Cayley-root separation & $0.745366$\\
\hline
block determinant modulus & $1.930433$\\
\hline
all $225$ Karlsson-equation moduli & $0.110717$\\
\hline
Tao witness $\lvert\rho^3-1\rvert$ & $1.996596$\\
\hline
all $1600$ block-automorphism witnesses & $1.067118$\\
\hline
\end{tabular}
\end{table}
The positive margins show numerically that the displayed completion is product regular and separated from the Karlsson, Tao, and four-circulant tests used here. They are reported as evidence about this representative seed, not as a theorem about a neighborhood or the dimension of its class-space image. No classification or product-coverage result depends on these numerical comparisons.

\ifdefined\SupplementOnly
\bibliography{ref}
\fi

\end{document}